\documentclass{llncs}
\usepackage[utf8]{inputenc}
\usepackage{amssymb,amsmath,stmaryrd}
\usepackage{mathtools}   
\usepackage{cases}
\usepackage{array}
\usepackage[table]{xcolor} 
\usepackage{blkarray} 
\usepackage{tikz} 
\usetikzlibrary{arrows,shapes,snakes,backgrounds,decorations.pathreplacing}
\usepackage{pdfpages}
\usepackage{wrapfig}
\usepackage{framed}

\newif\iflongversion 
\longversiontrue 
\iflongversion
\usepackage{fullpage}
\else
\fi

\usepackage{float}
\usepackage{dsfont}
\DeclareSymbolFont{cmmathcal}{OMS}{cmsy}{m}{n}
\DeclareSymbolFontAlphabet{\mathcal}{cmmathcal}
\usepackage{colortbl}
\usepackage[colorlinks=true, citecolor=blue,linkcolor=red,urlcolor=black]{hyperref}
\usepackage{algorithm}
\usepackage[noend]{algorithmic}
\usetikzlibrary{arrows,shapes,snakes,backgrounds}
\usepackage{paralist}
\usepackage{multirow}
\usepackage{enumerate}
\usepackage[lofdepth,lotdepth]{subfig}
\makeatletter
\renewcommand*{\@fnsymbol}[1]{\ensuremath{\ifcase#1\or \star\or \dagger\or \ddagger\or
       \mathsection\or \mathparagraph\or \|\or **\or \dagger\dagger
       \or \ddagger\ddagger \else\@ctrerr\fi}}
\makeatother

\let\oldmarginpar\marginpar
\renewcommand\marginpar[1]{\-\oldmarginpar[\raggedleft\scriptsize #1]%
{\raggedright\scriptsize #1}}

\newcommand\raisepunct[1]{\,\mathpunct{\raisebox{0.5ex}{#1}}}

\newcommand{\states}{\ensuremath{S} }
\newcommand{\statesSize}{\ensuremath{\vert\states\vert} }
\newcommand{\state}{\ensuremath{s} }
\newcommand{\statesOne}{\ensuremath{S_{1}} }
\newcommand{\statesTwo}{\ensuremath{S_{2}} }
\newcommand{\initState}{\ensuremath{s_{{\sf init}}} }
\newcommand{\edges}{\ensuremath{E} }
\newcommand{\edge}{\ensuremath{e} }
\newcommand{\succStates}[1]{\ensuremath{{\sf Succ}(#1)} }
\newcommand{\succStatesFull}{\ensuremath{\succStates{\state} = \{ \state' \in \states \mid (\state, \state') \in \edges\}} }
\newcommand{\weight}{\ensuremath{w} }
\newcommand{\largestW}{\ensuremath{W} }
\newcommand{\bits}{\ensuremath{V} }
\newcommand{\graph}{\ensuremath{\mathcal{G}} }
\newcommand{\graphFull}{\ensuremath{\mathcal{G} = (\states, \edges, \weight)} }
\newcommand{\game}{\ensuremath{G} }
\newcommand{\gameFull}{\ensuremath{\game = \left( \graph, \statesOne, \statesTwo\right)} }
\newcommand{\ec}{\ensuremath{U} }
\newcommand{\ecsSet}{\ensuremath{\mathcal{E}} }
\newcommand{\winningECs}{\ensuremath{\mathcal{W}} }
\newcommand{\losingECs}{\ensuremath{\mathcal{L}} }

\newcommand{\maxWinningECs}{\ensuremath{\mathcal{U}_{\textsc{w}}} }
\newcommand{\reduc}{\ensuremath{\downharpoonright} }
\newcommand{\proj}[1]{\ensuremath{{\sf proj}_{#1}} }
\newcommand{\infVisited}[1]{\ensuremath{{\sf Inf}(#1)} }

\newcommand{\markovChain}{\ensuremath{M} }
\newcommand{\mcTrans}{\ensuremath{\delta} }
\newcommand{\markovChainFull}{\ensuremath{\markovChain = (\graph, \mcTrans)} }
\newcommand{\mpTrans}{\ensuremath{\Delta} }
\newcommand{\markovProcess}{\ensuremath{{P} }}
\newcommand{\markovProcessFull}{\ensuremath{{P = (\graph, \statesOne, \statesProb, \mpTrans)} }}
\newcommand{\statesProb}{\ensuremath{S_{\mpTrans}} }
\newcommand{\thresholdWC}{\ensuremath{\mu} }
\newcommand{\thresholdExp}{\ensuremath{\nu} }
\newcommand{\optimalWC}{\ensuremath{\mu^{\ast}} }
\newcommand{\optimalExp}{\ensuremath{\nu^{\ast}} }
\newcommand{\truncatedTarget}{\ensuremath{T} }
\newcommand{\truncatedSum}[1]{\ensuremath{{\sf TS}_{#1}} }

\newcommand{\integ}{\ensuremath{\mathbb{Z}} }
\newcommand{\nat}{\ensuremath{\mathbb{N}} }
\newcommand{\natStrict}{\ensuremath{\mathbb{N}_{0}} }
\newcommand{\rat}{\ensuremath{\mathbb{Q}} }
\newcommand{\reals}{\ensuremath{\mathbb{R}} }
\newcommand{\realsAndInfinity}{\ensuremath{\mathbb{R} \cup \{-\infty,\, \infty\}} }
\newcommand{\player}{\ensuremath{\mathcal{P}} }
\newcommand{\playerOne}{\ensuremath{\mathcal{P}_{1}} }
\newcommand{\playerTwo}{\ensuremath{\mathcal{P}_{2}} }

\newcommand{\play}{\ensuremath{\pi} }
\newcommand{\plays}[1]{\ensuremath{{\sf Plays}(#1)} }
\newcommand{\prefixes}[1]{\ensuremath{\mathsf{Prefs}(#1)} }
\newcommand{\prefixesArg}[2]{\ensuremath{\mathsf{Prefs}_{#2}(#1)} }
\newcommand{\prefix}{\ensuremath{\rho} }
\newcommand{\last}[1]{\ensuremath{\mathsf{Last}(#1)} }
\newcommand{\first}[1]{\ensuremath{\mathsf{First}(#1)} }
\newcommand{\mpay}{\ensuremath{\mathsf{MP}} }
\newcommand{\tpay}{\ensuremath{\mathsf{TP}} }
\newcommand{\tpayFiniteFull}{\ensuremath{\tpay(\prefix) = \sum_{i = 0}^{i = n - 1} \weight((\state_{i}, \state_{i+1}))} }
\newcommand{\mpayFiniteFull}{\ensuremath{\mpay(\prefix) = \frac{1}{n} \tpay(\prefix)} }
\newcommand{\mpayFull}{\ensuremath{\mpay(\play) = \liminf_{n \rightarrow \infty} \mpay (\play(n))}}
\newcommand{\tpayFull}{\ensuremath{\tpay(\play) = \liminf_{n \rightarrow \infty} \tpay (\play(n))}}
\newcommand{\strat}{\ensuremath{\lambda} }
\newcommand{\stratStoch}{\ensuremath{\lambda^{{\sf stoch}}_{2}} }

\newcommand{\stratWC}{\ensuremath{\strat_{1}^{\textit{wc}}} }
\newcommand{\stratExp}{\ensuremath{\strat_{1}^{\textit{e}}} }
\newcommand{\stratComb}{\ensuremath{\strat_{1}^{\textit{cmb}}} }
\newcommand{\stratSecure}{\ensuremath{\strat_{1}^{\textit{sec}}} }
\newcommand{\stratWNS}{\ensuremath{\strat_{1}^{\textit{wns}}} }
\newcommand{\stratGlobal}{\ensuremath{\strat_{1}^{\textit{glb}}} }
\newcommand{\stepsWC}{\ensuremath{L} }
\newcommand{\stepsExp}{\ensuremath{K} }
\newcommand{\stepsGlobal}{\ensuremath{N} }
\newcommand{\cmbSum}{\ensuremath{\mathsf{Sum}} }
\newcommand{\typeA}{\ensuremath{\textit{(a)}} }
\newcommand{\typeB}{\ensuremath{\textit{(b)}} }
\newcommand{\strats}{\ensuremath{\Lambda} }
\newcommand{\stratsMemoryless}{\ensuremath{\Lambda^{M}} }
\newcommand{\stratsFinite}{\ensuremath{\Lambda^{F}} }
\newcommand{\stratsPureMemoryless}{\ensuremath{\Lambda^{PM}} }
\newcommand{\stratsPureFinite}{\ensuremath{\Lambda^{PF}} }

\newcommand{\outcomesMC}[2]{\ensuremath{\mathsf{Outs}_{#1}(#2)} }
\newcommand{\outcomesMDP}[3]{\ensuremath{\mathsf{Outs}_{#1}(#2,#3)} }
\newcommand{\outcomesGame}[4]{\ensuremath{\mathsf{Outs}_{#1}(#2,#3,#4)} }

\newcommand{\NPinter}{\ensuremath{\text{NP} \cap \text{coNP}}}
\newcommand{\PTIME}{\ensuremath{\text{P}}}
\newcommand{\NPTIME}{\ensuremath{\text{NP}}}
\newcommand{\PPTIME}{\ensuremath{\text{PP}}}

\newcommand{\playValue}{\ensuremath{f} }
\newcommand{\dist}{\ensuremath{p} }
\newcommand{\dists}{\ensuremath{\mathcal{D}} }
\newcommand{\event}{\ensuremath{\mathcal{A}} }

\newcommand{\proba}{\ensuremath{\mathbb{P}} }
\newcommand{\expect}{\ensuremath{\mathbb{E}} }

\newcommand{\supp}{\ensuremath{{\sf Supp}} }

\newcommand{\mooreMachine}[1]{\ensuremath{\mathcal{M}(#1)} }
\newcommand{\mooreMachineFull}[1]{\ensuremath{\mooreMachine{#1} = (\mooreMem, \mooreInitMem, \mooreUpd, \mooreNext)} }
\newcommand{\mooreMem}{\ensuremath{{\sf Mem}} }
\newcommand{\mooreMemElem}{\ensuremath{{\sf m}} }
\newcommand{\mooreInitMem}{\ensuremath{\mooreMemElem_{0}} }
\newcommand{\mooreUpd}{\ensuremath{\alpha_{{\sf u}}} }
\newcommand{\mooreUpdHat}{\ensuremath{\hat{\alpha}_{{\sf u}}} }
\newcommand{\mooreNext}{\ensuremath{\alpha_{{\sf n}}} }

\newcommand{\mpAlgoName}{\ensuremath{\textsc{BWC\_MP}} }
\newcommand{\mwecAlgoName}{\ensuremath{\textsc{MWEC}} }
\newcommand{\mpAlgo}[5]{\ensuremath{\mpAlgoName(#1,#2,#3,#4,#5)} }
\newcommand{\yes}{\ensuremath{\textsc{Yes}} }
\newcommand{\no}{\ensuremath{\textsc{No}} }

\newcommand{\BWC}{\text{BWC}}

\newcommand{\expDecFct}[2]{\ensuremath{\mathcal{F}(#1, #2)} }
\newcommand{\gameTS}{\ensuremath{\game_{\thresholdWC}} }
\newcommand{\tsFailSymbol}{\ensuremath{\top} }
\newcommand{\attr}{\ensuremath{\mathsf{Attr}} }

\newcommand{\kthSet}{\ensuremath{A} }
\newcommand{\kthSetSize}{\ensuremath{n} }
\newcommand{\kthSizeFct}{\ensuremath{h} }
\newcommand{\kthSizeFctFull}{\ensuremath{\kthSizeFct\colon \kthSet \rightarrow \natStrict} }
\newcommand{\kthSetsNbr}{\ensuremath{K} }
\newcommand{\kthSetMaxSum}{\ensuremath{L} }
\newcommand{\kthSubset}{\ensuremath{C} }
\newcommand{\kthElem}{\ensuremath{a} }
\newcommand{\kthNewSizeFct}{\ensuremath{\kthSizeFct_{\kthSetSize}} }
\newcommand{\kthRandomPath}[1]{\ensuremath{p_{#1}} }
\newcommand{\kthPathSizeFct}{\ensuremath{t} }
\newcommand{\kthPathBound}{\ensuremath{T} }
\newcommand{\kthWCState}{\ensuremath{\state_{wc}} }
\newcommand{\kthExpState}{\ensuremath{\state_{e}} }
\newcommand{\kthWeightA}{\ensuremath{x_{1}} }
\newcommand{\kthWeightB}{\ensuremath{x_{2}} }
\newcommand{\kthWeightC}{\ensuremath{x_{3}} }
\newcommand{\kthPathMax}{\ensuremath{Q} }
\newcommand{\kthLBExp}{\ensuremath{{\sf LB}_{e} } }
\newcommand{\kthLBWC}{\ensuremath{{\sf LB}_{wc} } }
\newcommand{\kthUBExp}{\ensuremath{{\sf UB}_{e} } }
\newcommand{\kthUBWC}{\ensuremath{{\sf UB}_{wc} } }
\newcommand{\edgesNonZero}{\ensuremath{\edges_{\mpTrans}} }
\newcommand{\negligibleStates}{\ensuremath{\states_{{\sf neg}}} }
\newcommand{\gameNonZero}{\ensuremath{\game_{\mpTrans}} }
\newcommand{\markovProcessNonZero}{\ensuremath{\markovProcess_{\mpTrans}} }
\newcommand{\mwecAlgoSet}{\ensuremath{\mathcal{M}_{\textsc{w}} }}

\floatstyle{plain}
\newfloat{myalgo}{tbhp}{mya}
{\begin{myalgo}[#1]
\centering
\begin{minipage}{#2}
\begin{algorithm}[H]}%
{\end{algorithm}
\end{minipage}
\end{myalgo}}

\renewcommand{\arraystretch}{1.2}

\spnewtheorem{assumption}{Assumption}{\bfseries}{\itshape}

\let\doendproof\endproof
\renewcommand\endproof{~\hfill\qed\doendproof}

\title{Meet Your Expectations With Guarantees:\\Beyond Worst-Case Synthesis in Quantitative Games\thanks{E.~Filiot is a F.R.S.-FNRS Research Associate, M.~Randour is a F.R.S.-FNRS Postdoctoral Researcher, J.-F.~Raskin is supported by ERC Starting Grant (279499: inVEST). Work partly supported by European project CASSTING (FP7-ICT-601148).}}

\author{V\'eronique Bruy\`ere\inst{1} \and Emmanuel Filiot\inst{2} \and Mickael Randour\inst{1,2,3} \and \mbox{Jean-Fran\c{c}ois Raskin\inst{2}}}

\institute{
Computer Science Department, Université de Mons (UMONS), Belgium\\
\and D\'epartement d'Informatique, Universit\'e Libre de Bruxelles (ULB), Belgium\\
\and LSV - CNRS \& ENS Cachan, France
}

\begin{document}

\maketitle

\begin{abstract}
We extend the quantitative synthesis framework by going beyond the worst-case. On the one hand, classical analysis of two-player games involves an adversary (modeling the environment of the system) which is purely antagonistic and asks for strict guarantees. On the other hand, stochastic models like Markov decision processes represent situations where the system is faced to a purely randomized environment: the aim is then to optimize the expected payoff, with no guarantee on individual outcomes. We introduce the beyond worst-case synthesis problem, which is to construct strategies that guarantee some quantitative requirement in the worst-case while providing a higher expected value against a particular stochastic model of the environment given as input. 
This problem is relevant to produce system controllers that provide nice expected performance in the everyday situation while ensuring a strict (but relaxed) performance threshold even in the event of very bad (while unlikely) circumstances.
We study the beyond worst-case synthesis problem for two important quantitative settings: the mean-payoff and the shortest path. 
In both cases, we show how to decide the existence of finite-memory strategies satisfying the problem and how to synthesize one if one exists. We establish algorithms and we study complexity bounds and memory requirements.
\end{abstract}

\section{Introduction}
\label{sec:introduction}

\smallskip\noindent\textbf{Classical models.} Two-player zero-sum quantitative games~\cite{EM79,ZP96,BCDGR11} and Markov decision processes (abbrev.~MDPs)~\cite{Puterman94,chatterjee_MEMICS11} are two popular formalisms for modeling decision making in adversarial and uncertain environments respectively. In the former, two players compete with opposite goals (zero-sum), and we want strategies for player~1 (the system) that ensure a given \textit{minimal performance against all possible strategies} of player~2 (its environment). In the latter, the system plays against a stochastic model of its environment, and we want strategies that ensure a \textit{good expected overall performance}.  Those two models are well studied and simple optimal memoryless strategies exist for classical objectives such as mean-payoff~\cite{liggett_SR69,EM79,filar1997} or shortest path~\cite{bertsekas_MOR1991,deAlfaro_CONCUR1999}. But both models have clear weaknesses: strategies that are good for the worst-case may exhibit suboptimal behaviors in probable situations while strategies that are good for the expectation may be terrible in some unlikely but possible situations.

\smallskip\noindent\textbf{What if we want both?} In practice, we would like to have strategies that are both ensuring $(a)$  some worst-case threshold no matter how the adversary behaves (i.e., against any arbitrary strategy) and $(b)$ a good expectation against the expected behavior of the adversary (given as a stochastic model).  This is the subject of this paper: we show how to construct finite-memory strategies that ensure both $(a)$ and $(b)$. We consider finite-memory strategies for player~1 as they can be implemented in practice (as opposed to infinite-memory ones). Player~2 is not restricted in his choice of strategies, but we will see that simple strategies suffice. Our problem, the \textbf{beyond worst-case synthesis problem}, is interesting for any quantitative measure, but we give here a thorough study of two classical ones: the {\em mean-payoff}, and the {\em shortest path}. Our results are summarized in Table~\ref{summaryTable}.

\renewcommand{\arraystretch}{1.4}
\begin{table}[htb]
\small
\centering
\begin{tabular}{|c|c||c|c|c|}
\cline{3-5} \multicolumn{2}{c|}{} & worst-case & ~expected value~ & \textbf{BWC}\\
\hline
\hline
\multirow{2}{*}{mean-payoff} & ~~complexity~~ & ~~$\NPinter$~~ & $\PTIME$ & \textbf{$\NPinter$}\\
\cline{2-5}
& memory & \multicolumn{2}{c|}{memoryless} & \textbf{pseudo-poly.}\\
\hline
\multirow{2}{*}{~~shortest path~~} & complexity & \multicolumn{2}{c|}{$\PTIME$} & ~~\textbf{pseudo-poly. / $\NPTIME$-hard}~~\\
\cline{2-5}
& memory & \multicolumn{2}{c|}{memoryless} & \textbf{pseudo-poly.}\\
\hline
\end{tabular}
\vspace{2mm}
\caption{Overview of decision problem complexities and (tight) memory requirements for winning strategies of player~1 in games (worst-case), MDPs (expected value) and the BWC setting (combination).}
\label{summaryTable}
\end{table}

\smallskip\noindent\textbf{Example.} Let us consider the weighted game in Fig.~\ref{fig:exampleTS} to illustrate the {\em shortest path} context. Circle states belong to player~1, square states to player~2, integer labels are durations in minutes, and fractions are probabilities that model the expected behavior of player~2. Player~1 wants a strategy to go from ``home'' to ``work'' such that ``work'' is \textit{guaranteed} to be reached within 60 minutes (to avoid missing an important meeting), and player~1 would also like to minimize the expected time to reach ``work''. First, note that the strategy that minimizes the expectation is to take the car (expectation is 33 minutes) but this strategy is excluded as there is a possibility to arrive after 60 minutes (in case of heavy traffic). Bicycle is safe but the expectation of this solution is 45 minutes. We can do better with the following strategy: try to take the train, if the train is delayed three times consecutively, then go back home and take the bicycle. This strategy is safe as it always reaches ``work'' within 58 minutes and its expectation is $\approx 37,45$ minutes (so better than taking directly the bicycle). Observe that this simple example already shows that, unlike the situation for classical games and MDPs, strategies using memory are strictly more powerful than memoryless ones. Our algorithms are able to decide the existence of (and synthesize) such finite-memory strategies.
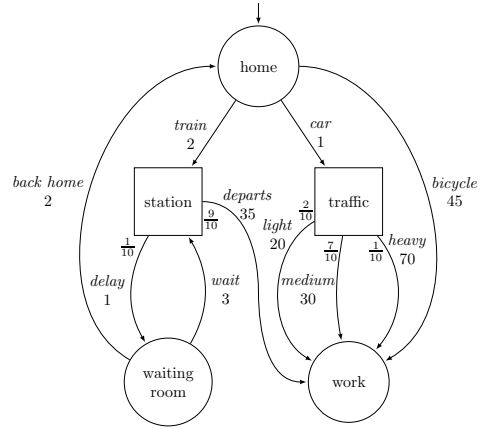
\begin{wrapfigure}{r}{65mm}
  \centering
\scalebox{0.6}{\begin{tikzpicture}[->,>=latex,shorten >=1pt,auto,node
    distance=2.5cm,bend angle=45,font=\normalsize]
    \tikzstyle{p1}=[draw,circle,text centered,minimum size=10mm, text width=15mm]
    \tikzstyle{p2}=[draw,rectangle,text centered,minimum size=15mm]
    \tikzstyle{empty}=[]
    \node[p1] (1) at (0,0) {home};
    \node[p2] (2) at (-2,-3) {station};
    \node[p2] (3) at (2,-3) {traffic};
    \node[p1] (4) at (-2,-7) {waiting room};
    \node[p1] (5) at (2,-7) {work};
    \node[empty] (a) at (-2.9,-4) {$\frac{1}{10}$};
    \node[empty] (b) at (-1.05,-3.4) {$\frac{9}{10}$};
    \node[empty] (c) at (1.05,-3.2) {$\frac{2}{10}$};
    \node[empty] (d) at (1.6,-4.1) {$\frac{7}{10}$};
    \node[empty] (e) at (2.6,-4.1) {$\frac{1}{10}$};
    \coordinate[shift={(0mm,5mm)}] (init) at (1.north);
    \path
    (1) edge node[left,align=center] {\textit{train}\\2} (2)
    (1) edge node[right,align=center] {\textit{car}\\1} (3) 
    (init) edge (1)
    ;
	\draw[->,>=latex] (4) to[out=150,in=180] node[left,align=center] {\textit{back home}\\2} (1);
	\draw[->,>=latex] (1) to[out=0,in=30] node[right,align=center] {\textit{bicycle}\\45} (5);
	\draw[->,>=latex] (2) to[out=240,in=120] node[left,align=center] {\textit{delay}\\1} (4);
	\draw[->,>=latex] (4) to[out=60,in=300] node[right,align=center] {\textit{wait}\\3} (2);
	\draw[->,>=latex] (3) to[out=210,in=150] node[left, very near start,align=center] {\textit{light}\\20} (5);
	\draw[->,>=latex] (3) to[out=260,in=100] node[left,align=center,xshift=1mm] {\textit{medium}\\30} (5);
	\draw[->,>=latex] (3) to[out=310,in=50] node[right, very near start,align=center,xshift=-1mm] {\textit{heavy}\\70} (5);
	\draw [->,>=latex] (2) to[out=0,in=90] node[right, very near start,align=center] {\textit{departs}\\35}(0, -5) to[out=270,in=180] (5);
      \end{tikzpicture}}
      \caption{Player 1 wants to minimize its expected time to reach ``work'', but while ensuring it is less than one hour in all cases.}
\label{fig:exampleTS}
\vspace{-6mm}
\end{wrapfigure}

\smallskip\noindent\textbf{Contributions.} Our main results are the following. 
First, for the mean-payoff value, we provide an algorithm (Thm.~\ref{thm:mp_decisionProblem}) that implies $\NPinter$-membership of the problem, which would reduce to $\PTIME$ if mean-payoff games were proved to be in~$\PTIME$, a long-standing open problem~\cite{BCDGR11,Chatterjee201525}. Pseudo-polynomial memory may be necessary and always suffices (Thm.~\ref{thm:mp_memoryRequirements}). Finally, we observe that infinite-memory strategies are strictly more powerful that finite-memory strategies (Sect.~\ref{subsec:mpInfiniteMemory}).
Second, for the shortest path, we provide a pseudo-polynomial time algorithm (Thm.~\ref{thm:ts_pseudoPoly}), and show that the associated decision problem is $\NPTIME$-hard (Thm.~\ref{thm:ts_NPHardness}). According to a very recent result by Haase and Kiefer~\cite{HaasePP}, our reduction even proves $\PPTIME$-hardness, which suggests that the problem does not belong to $\NPTIME$ at all otherwise the polynomial hierarchy would collapse. Pseudo-polynomial memory may be necessary and always suffices (Thm.~\ref{thm:ts_memory}). 
In the case of the shortest path problem, infinite-memory strategies give no additional power in comparison with finite-memory strategies (Rem.~\ref{rem:ts_infMem}). 

\smallskip\noindent\textbf{Related work.} This paper extends previous works presented in conferences~\cite{bruyere_STACS2014,DBLP:journals/corr/BruyereFRR14} and in a PhD thesis~\cite{Ran14}: it gives a full
presentation of the technical details, along with additional results. Our problems generalize the corresponding problems for two-player zero-sum games and MDPs. In mean-payoff games, optimal memoryless worst-case strategies exist and the best known complexity is $\NPinter$~\cite{EM79,ZP96,BCDGR11}. For the shortest path, we consider game graphs with strictly positive weights ant try to minimize the cost to target: memoryless strategies also suffice, and the problem is in $\PTIME$~\cite{DBLP:dblp_journals/mst/KhachiyanBBEGRZ08}. In MDPs, optimal strategies for the expectation are studied in~\cite{Puterman94,filar1997} for the mean-payoff and the shortest path: in both cases, memoryless strategies suffice and they can be computed in polynomial time.

Our strategies are {\em strongly risk averse}: they avoid at all cost outcomes that are below a given threshold (no matter what is their probability), and inside the set of those {\em safe} strategies, we maximize expectation. To the best of our knowledge, we are the first to consider such strategies. Other different notions of risk have been studied for MDPs: for example in~\cite{WL99}, the authors want to find policies which minimize the probability (risk) that the total discounted rewards do not exceed a specified value (target), or in~\cite{FKR95} the authors want policies that achieve a specified value of the long-run limiting average reward at a specified probability level (percentile). The latter problem has been recently extended significantly in the framework of \textit{percentile queries}, which provide elaborate guarantees on the performance profile of strategies in multi-dimensional MDPs~\cite{RRS15a}. While all those strategies limit risk, they only ensure {\em low probability} for bad behaviors but they do not ensure their absence, furthermore, they do not ensure good expectation either.

Another body of work is the study of strategies in MDPs that achieve a trade-off between the expectation and the variance over the outcomes (e.g.,~\cite{brazdil_LICS2013} for the mean-payoff,~\cite{mannor_ICML2011} for the cumulative reward), giving a statistical measure of the stability of the performance. In our setting, we strengthen this requirement by asking for \textit{strict guarantees on individual outcomes}, while maintaining an appropriate expected payoff.

A survey of rich behavioral models extending the classical approaches for MDPs --- including the beyond worst-case framework presented here --- was published in~\cite{RRS15b}, with a focus on the shortest path problem.

In this paper, we focus on \textit{finite}-memory strategies. The synthesis of \textit{infinite}-memory strategies for the beyond worst-case mean-payoff problem was recently addressed in~\cite{CR15}.

\smallskip\noindent\textbf{Structure of the paper.} In Sect.~\ref{sec:preliminaries}, we introduce the necessary definitions. In Sect.~\ref{sec:problem}, we formally define the beyond worst-case synthesis problem. Sect.~\ref{sec:mean_payoff} and Sect.~\ref{sec:shortest_path} are respectively devoted to the solutions for the mean-payoff and the shortest path. We conclude (Sect.~\ref{sec:conclusion}) with a comparative note on the two solutions.

\section{Preliminaries}
\label{sec:preliminaries}

\smallskip\noindent\textbf{Weighted directed graphs.} A {\em weighted directed graph} is a tuple $\graphFull$ where (i) $\states$ is the set of vertices, called {\em states}; (ii) $\edges \subseteq \states \times \states$ is the set of directed edges; and (iii) $\weight\colon \edges \rightarrow \integ$ is the weight labeling function. Since we only work with directed graphs in the following, we omit the adjective and talk about \textit{weighted graphs}. Also, in the sequel, we almost exclusively work with \textit{finite} graphs, i.e., graphs for which the set of states $\states$ is finite. Given a state $\state \in \states$, we denote by $\succStatesFull$ the set of successors of $\state$ by edges in $\edges$. We assume that graphs are non-blocking, i.e., for all $\state \in \states$, $\succStates{\state} \neq \emptyset$. We denote by $\largestW$ the largest absolute weight that appears in the graph. We assume that weights are encoded in binary and denote by $\bits = \lceil \log_{2} \largestW \rceil$ the number of bits of their encoding.

A \textit{play} in $\graph$ from an initial state $\initState \in \states$ is an infinite sequence of states $\play = \state_{0}\state_{1}\state_{2}\ldots{}$ such that $\state_{0} = \initState$ and $(s_{i}, s_{i+1}) \in \edges$ for all $i \geq 0$. The \textit{prefix} up to the $n$-th state of $\play$ is the finite sequence $\play(n) = s_{0}s_{1}\ldots{}s_{n}$. We resp. denote the first and last states of the prefix by $\first{\play(n)} = s_{0}$ and $\last{\play(n)} = s_{n}$. For a play $\play$, we naturally extend the notation to $\first{\play}$. The set of plays of $\graph$ is denoted by $\plays{\graph}$ and the corresponding set of prefixes is denoted by $\prefixes{\graph}$. Given a play $\play \in \plays{\graph}$, we denote by $\infVisited{\play} \subseteq \states$ the set of states that are visited infinitely often along the play.

Given a function $\playValue\colon \plays{\graph} \rightarrow \realsAndInfinity$, the \textit{value} of a play $\play$ is denoted by $\playValue(\play)$. We consider two classical value functions, the \textit{total-payoff} and the \textit{mean-payoff}, defined as follows. The \textit{total-payoff} of a prefix $\prefix = s_{0}s_{1}\ldots{}s_{n}$ is $\tpayFiniteFull$, and its \textit{mean-payoff} is $\mpayFiniteFull$. This is naturally extended to plays by considering the limit behavior: the total-payoff of a play $\play$ is $\tpayFull$ and its mean-payoff is $\mpayFull$.
Given a graph $\graph$ where all weights are strictly positive (i.e., $\weight\colon \edges \rightarrow \natStrict$) and a target set of states $\truncatedTarget \subseteq \states$, we define the \textit{truncated sum up to $\truncatedTarget$} as $\truncatedSum{\truncatedTarget} \colon \plays{\graph} \rightarrow \nat \cup \{ \infty \}$, $\truncatedSum{\truncatedTarget} (\play = s_{0}s_{1}s_{2}\ldots{}) = \sum_{i = 0}^{n-1} \weight((s_{i}, s_{i+1}))$, with $n$ the first index such that $s_{n} \in \truncatedTarget$, and $\truncatedSum{\truncatedTarget} (\play) = \infty$ if $\play$ never reaches any state in $\truncatedTarget$. As all weights are strictly positive, it is possible to reduce the truncated sum to the total-payoff (i.e., for all $\play \in \plays{\graph}$, $\truncatedSum{\truncatedTarget} (\play) = \tpay(\play)$) by making all states of $\truncatedTarget$ absorbing with a self-loop of zero weight. That is, for all $\state \in \truncatedTarget$, we have that $\succStates{\state} = \{\state\}$ and $\weight((\state, \state)) = 0$.

\smallskip\noindent\textbf{Probability distributions.} Given a finite set $A$, a (rational) \textit{probability distribution} on $A$ is a function $\dist \colon A \rightarrow [0, 1] \cap \rat$ such that $\sum_{a\in A} \dist(a) = 1$. We denote the set of probability distributions on $A$ by $\dists(A)$. The \textit{support} of the probability distribution $\dist$ on $A$ is $\supp(\dist) = \left\lbrace a \in A \;\vert\; \dist(a) > 0\right\rbrace$.

\smallskip\noindent\textbf{Two-player games.} We consider two-player turn-based games and denote the two \textit{players} by $\playerOne$ and $\playerTwo$. A finite \textit{two-player game} is a tuple $\gameFull$ composed of (i) a finite weighted graph $\graphFull$; and (ii) a partition of its states $\states$ into $\statesOne$ and $\statesTwo$ that resp. denote the sets of states belonging to $\playerOne$ and $\playerTwo$. A prefix $\play(n)$ of a play $\play$ belongs to $\player_{i}$, $i \in \lbrace 1, 2\rbrace$, if $\last{\play(n)} \in \states_{i}$. The set of prefixes that belong to $\player_{i}$ is denoted by $\prefixesArg{\game}{i}$. We sometimes denote by~$\vert\game\vert$ the size of a game, defined as a polynomial function of $\vert\states\vert$, $\vert\edges\vert$ and $\bits = \lceil \log_{2} \largestW \rceil$.

\smallskip\noindent\textbf{Strategies.} Let $\gameFull$ be a two-player game, a \textit{strategy} for player $\player_{i}$, $i \in \lbrace 1, 2\rbrace$, is a function $\strat_{i} \colon \prefixesArg{\game}{i} \rightarrow \dists(\states)$ such that for all $\prefix \in \prefixesArg{\game}{i}$, we have $\supp(\strat_{i}(\prefix)) \subseteq \succStates{\last{\prefix}}$. A strategy is called \textit{pure} if it is deterministic, i.e., if its support is a singleton for all prefixes. When a strategy $\strat_{i}$ of $\player_{i}$ is pure, we sometimes simplify its notation and write $\strat_{i}(\prefix) = \state$ instead of $\strat_{i}(\prefix)(\state) = 1$, for any $\prefix \in \prefixesArg{\game}{i}$ and the unique state $\state \in \supp(\strat_{i}(\prefix))$.

A strategy $\strat_{i}$ for $\player_{i}$ has \textit{finite memory} if it can be encoded by a stochastic finite state machine with outputs, called {\em stochastic Moore machine}, $\mooreMachineFull{\strat_{i}}$, where (i) $\mooreMem$ is a finite set of memory elements, (ii)~$\mooreInitMem \in \mooreMem$ is the initial memory element, (iii) $\mooreUpd \colon \mooreMem \times \states \to \mooreMem$ is the update function, and (iv) $\mooreNext \colon \mooreMem \times \states_{i} \to \dists(\states)$ is the next-action function. If the game is in $\state \in \states_{i}$ and $\mooreMemElem \in \mooreMem$ is the current memory element, then the strategy chooses $\state'$, the next state of the game, according to the probability distribution $\mooreNext(\mooreMemElem, \state)$. When the game leaves a state $\state \in \states$, the memory is updated to $\mooreUpd(\mooreMemElem, \state)$. Hence updates are deterministic and outputs are potentially stochastic. Formally, $(\mooreMem, \mooreInitMem, \mooreUpd, \mooreNext)$ defines the strategy $\strat_{i}$ such that $\strat_{i}(\prefix \cdot \state) = \mooreNext(\mooreUpdHat(\mooreInitMem, \prefix), \state)$ for all $\prefix \in \prefixes{\graph}$ and $\state \in \states_{i}$, where $\mooreUpdHat$ extends $\mooreUpd$ to sequences of states starting from $\mooreInitMem$ as expected. Note that pure finite-memory strategies have deterministic next-action functions. A strategy is \textit{memoryless} if $\vert \mooreMem\vert = 1$, i.e., it does not depend on the history but only on the current state of the game.

We resp. denote by $\strats_{i}(\game), \stratsFinite_{i}(\game), \stratsPureFinite_{i}(\game), \stratsMemoryless_{i}(\game)$ and $\stratsPureMemoryless_{i}(\game)$ the sets of general (i.e., possibly randomized and infinite-memory), finite-memory, pure finite-memory, memoryless and pure memoryless strategies for player $\player_{i}$ on the game $\game$. We do not write $\game$ in this notation when the context is clear. A~play~$\play$ is said to be \textit{consistent} with a strategy $\strat_{i} \in \strats_{i}$ if for all $n \geq 0$ such that $\last{\play(n)} \in \states_{i}$, we have $\last{\play(n+1)} \in \supp(\strat_{i}(\play(n))$.

\smallskip\noindent\textbf{Markov decisions processes.} 
A finite \textit{Markov decision process} (MDP) is a tuple $\markovProcessFull$ where (i) $\graphFull$ is a finite weighted graph, (ii) $\statesOne$ and $\statesProb$ define a partition of the set of states $\states$ into states of $\playerOne$ and \textit{stochastic states}, and (iii) $\mpTrans\colon \statesProb \rightarrow \dists(\states)$ is the transition function that, given a stochastic state $\state \in \statesProb$, defines the probability distribution $\mpTrans(s)$ over the possible successors of $s$, such that for all states $s \in \statesProb$, $\supp(\mpTrans(s)) \subseteq \succStates{\state}$.

In contrast to some other classical definitions of MDPs in the literature, we explicitly allow that, for some states $s \in \statesProb$, $\supp(\mpTrans(s)) \subsetneq \succStates{\state}$: some edges of the graph $\graph$ are assigned probability zero by the transition function. This is important as far as modeling is concerned, as in our context, transition functions will be defined according to a stochastic model for the environment of a system, and we cannot reasonably assume that such a model always involves all possible actions of the environment: it is possible that some actions are only used by an antagonistic environment, in our two-player game view. Hence, we study the most general setting --- where edges of probability zero are allowed. Consequently, given the MDP $\markovProcess$, we define the subset of edges $\edgesNonZero = \{ (\state_{1}, \state_{2}) \in \edges \mid \state_{1} \in \statesProb \Rightarrow \state_{2} \in \supp(\mpTrans(\state_{1}))\}$, representing all edges that either start in a state of $\playerOne$, or are chosen with non-zero probability by the transition function $\mpTrans$.

An MDP can be seen as a two-player game where $\playerOne$ is playing against a probabilistic adversary using a fixed randomized memoryless strategy $\mpTrans$ in states of the set $\statesProb$. Hence MDPs are sometimes called $1\frac{1}{2}$-player games. The notions of prefixes belonging to $\playerOne$ and of strategies for $\playerOne$ are naturally extended to MDPs.

\smallskip\noindent\textbf{End-components.}
We define \textit{end-components} (ECs) of an MDP as subgraphs in which $\playerOne$ can ensure to stay despite stochastic states \cite{de1997formal}. Formally, let $\markovProcessFull$ be an MDP, with $\graphFull$ its underlying graph. An EC in~$\markovProcess$ is a set $\ec \subseteq \states$ such that (i) the subgraph $(\ec, \edges_{\mpTrans} \cap (\ec \times \ec))$ is strongly connected, with $\edges_{\mpTrans}$ defined as before, i.e., stochastic edges with probability zero are treated as non-existent; and (ii) for all $\state \in \ec \cap \statesProb$, $\supp(\mpTrans(\state)) \subseteq \ec$, i.e., in stochastic states, all outgoing edges either stay in $\ec$ or belong to $\edges \setminus \edgesNonZero$ (that is, the probability of leaving $\ec$ from a state $\state \in \statesProb$ is zero). The set of all ECs of $\markovProcess$ is denoted $\ecsSet \subseteq 2^{\states}$.

\smallskip\noindent\textbf{Markov chains.} A finite \textit{Markov chain} (MC) is a tuple $\markovChainFull$ where (i) $\graphFull$ is a finite weighted graph; and (ii) $\mcTrans \colon \states \rightarrow \dists(\states)$ is the transition function that, given a state $\state \in \states$, defines the probability distribution $\mcTrans(s)$ over the successors of $s$, such that for all states $s \in \states$, $\supp(\mcTrans(s)) \subseteq \succStates{\state}$.

In a Markov chain $\markovChainFull$, an \textit{event} is a measurable set of plays $\event \subseteq \plays{\graph}$. Every event has a uniquely defined probability \cite{vardi_FOCS85} (Carathéodory's extension theorem induces a unique probability measure on the Borel $\sigma$-algebra over $\plays{\graph}$). We denote by $\proba^{\markovChain}_{\initState}(\event)$ the probability that a play belongs to $\event$ when the Markov chain $\markovChain$ starts in $\initState \in \states$ and is executed for an infinite number of steps. Given a measurable value function $\playValue \colon \plays{\graph} \rightarrow \realsAndInfinity$, we denote by $\expect^{\markovChain}_{\initState}(\playValue)$ the \textit{expected value} or \textit{expectation} of $\playValue$ over a play starting in $\initState$. The $\sigma$-algebra is defined through cylinder sets of prefixes: each prefix $\prefix$ defines a set of plays $\play$ such that $\prefix$ is a prefix of $\play$ \cite{baier_MIT08}. Hence, the notions of probability and expected value can naturally be used over prefixes by considering the plays belonging to their cylinder set.

\smallskip\noindent\textbf{Projections.} Given a set $A_{i}$, $1 \leq i \leq k$ of a cartesian product $A_{1} \times \ldots{} \times A_{k}$, we define the \textit{projection} over $A_{i}$, denoted by $\proj{A_{i}}\colon A_{1} \times \ldots{} \times A_{k} \rightarrow A_{i}$, as the mapping from elements $\overline{a} = (a_{1}, \ldots{}, a_{k})$ to $\proj{A_{i}}(\overline{a}) = a_{i}$.

\smallskip\noindent\textbf{Outcomes.}
Let $\markovChain = (\graph, \mcTrans)$ be a Markov chain, with $\graph = (\states, \edges, \weight)$ its underlying graph. Given an initial state $\initState \in \states$, we define the set of its possible \textit{outcomes} as
\begin{equation*}
\outcomesMC{\markovChain}{\initState} = \left\lbrace \play = \state_{0}\state_{1}\state_{2}\ldots{} \in \plays{\graph} \;\vert\; \state_{0} = \initState \wedge \forall\, n \in \nat,\, s_{n+1} \in \supp(\mcTrans(\state_{n}))\right\rbrace.
\end{equation*}
Note that if $\mcTrans$ is deterministic (i.e., if the support is a singleton) in all states, we obtain a unique play $\play = s_{0}s_{1}s_{2}\ldots{}$ as the unique possible outcome.

Let $\game = (\graph, \statesOne, \statesTwo)$ be a two-player game, with $\graph = (\states, \edges, \weight)$ its underlying graph. Given two strategies, $\strat_{1} \in \strats_{1}$ and $\strat_{2} \in \strats_{2}$, and an initial state $\initState \in \states$, we extend the notion of outcomes as follows:
\begin{equation*}
\outcomesGame{\game}{\initState}{\strat_{1}}{\strat_{2}} = \left\lbrace \play = \state_{0}\state_{1}\state_{2}\ldots{} \in \plays{\graph} \;\vert\; \state_{0} = \initState \wedge \play \text{ is consistent with } \strat_{1} \text{ and } \strat_{2}\right\rbrace.
\end{equation*}
Observe that when fixing the strategies, we obtain an MC denoted by $\game[\strat_{1}, \strat_{2}]$. This MC is finite if both~$\strat_{1}$ and~$\strat_{2}$ are finite-memory strategies. Let $\mooreMachine{\strat_{1}} = (\mooreMem_{1}, \mooreMemElem_{1}, \mooreUpd^{1}, \mooreNext^{1})$ and $\mooreMachine{\strat_{2}} = (\mooreMem_{2}, \mooreMemElem_{2}, \mooreUpd^{2}, \mooreNext^{2})$ be the Moore machines of two such strategies. The set of states of the resulting MC is obtained through the product of the memory elements of the strategies given as Moore machines and the states of the game, i.e., $\states \times \mooreMem_{1} \times \mooreMem_{2}$; and its transition function is defined based on the distributions prescribed by the strategies and in order to accurately account for the memory updates. Notice that the outcomes of $\game$ and $\game[\strat_{1}, \strat_{2}]$ are different objects by nature: the former are plays on a graph defined by the set of states $\states$ while the latter are plays on a graph defined by $\states \times \mooreMem_{1} \times \mooreMem_{2}$. Still, there exists a bijection between outcomes of the MC and their \textit{traces} in the initial game, thanks to the projection operator (Lemma \ref{lem:projBijective}).
\begin{lemma}
\label{lem:projBijective}
Let $\gameFull$ be a game, with $\graphFull$ its underlying graph. Let $\strat_{1} \in \stratsFinite_{1}$ and $\strat_{2} \in \stratsFinite_{2}$ be the finite-memory strategies of the players. Then there is a bijection between outcomes in $\game$ and outcomes in the resulting Markov chain $\game[\strat_{1}, \strat_{2}]$.
\end{lemma}

\begin{proof}
Let $\initState \in \states$ be the initial state of $\game$, $\mooreMachine{\strat_{1}} = (\mooreMem_{1}, \mooreMemElem_{1}, \mooreUpd^{1}, \mooreNext^{1})$ and $\mooreMachine{\strat_{2}} = (\mooreMem_{2}, \mooreMemElem_{2}, \mooreUpd^{2}, \mooreNext^{2})$ be the Moore machines. Consider an outcome in $\game[\strat_{1}, \strat_{2}]$: it is a sequence of states from $\states \times \mooreMem_{1} \times \mooreMem_{2}$. Obviously, its projection on the set $\states$ is unique and defines the outcome in the sense of $\game$.

Conversely, consider an outcome in $\game$: it is of the form $\state_{0}\state_{1}\state_{2}\ldots{} \in \states^{\omega}$, with $\state_{0} = \initState$. We claim there is a unique corresponding outcome in $\game[\strat_{1}, \strat_{2}]$, written $(\state_{0}, \mooreMemElem_{1}^{0}, \mooreMemElem_{2}^{0}) (\state_{1}, \mooreMemElem_{1}^{1}, \mooreMemElem_{2}^{1})\ldots{} \in \states \times \mooreMem_{1} \times \mooreMem_{2}$, with  $(\state_{0}, \mooreMemElem_{1}^{0}, \mooreMemElem_{2}^{0}) = (\initState, \mooreMemElem_{1}, \mooreMemElem_{2})$. Indeed, it suffices to see that the update functions of the Moore machines, $\mooreUpd^{1}$ and $\mooreUpd^{2}$, are deterministic functions. Hence, it is easy to reconstruct the outcome of $\game[\strat_{1}, \strat_{2}]$ based on its projection on $\states$ as it suffices to apply the effect of the update functions on the memory at each step.
\end{proof}
Hence, we obtain the following equality:
$\outcomesGame{\game}{\initState}{\strat_{1}}{\strat_{2}} = \proj{\states}\left( \outcomesMC{\game[\strat_{1}, \strat_{2}]}{(\initState, \mooreMemElem_{1}, \mooreMemElem_{2})}\right)$.
Based on this, and for the sake of readability, we abuse the notation and write $\outcomesMC{\game[\strat_{1}, \strat_{2}]}{\initState}$ equivalently to refer to this set of outcomes. Similar abuse is taken for value functions and initial states.

Back to the outcomes of the game: note that if both strategies $\strat_{1}$ and $\strat_{2}$ are pure, the resulting Markov chain only involves Dirac distributions ($\mcTrans$ is deterministic) and the set $\outcomesGame{\game}{\initState}{\strat_{1}}{\strat_{2}}$ is composed of a unique play $\play = s_{0}s_{1}s_{2}\ldots{}$ such that for all $n \geq 0$, $i \in \lbrace 1, 2\rbrace$, if $s_{n} \in \states_{i}$, then we have $\strat_{i}(\state_{n}) = \state_{n+1}$.

Let $\markovProcess = (\graph, \statesOne, \statesProb, \mpTrans)$ be a Markov decision process, with $\graph = (\states, \edges, \weight)$ its underlying graph. Again, we can fix the strategy $\strat_{1}$ of $\playerOne$ and obtain the Markov chain $\markovProcess[\strat_{1}]$. Let $\mooreMachine{\strat_{1}} = (\mooreMem, \mooreInitMem, \mooreUpd, \mooreNext)$. The set of outcomes starting in $\initState \in \states$ is defined as $\outcomesMDP{\markovProcess}{\initState}{\strat_{1}} = \proj{\states}\left( \outcomesMC{\markovProcess[\strat_{1}]}{(\initState, \mooreInitMem)}\right)$. Again, we abuse the notation and write $\outcomesMC{\markovProcess[\strat_{1}]}{\initState}$ equivalently.

Finally, back to the two-player game $\game$, if we fix the strategy $\strat_{i}$ of only one player $\player_{i}$, $i \in \{1, 2\}$, we obtain not a Markov chain, but a Markov decision process for the remaining player $\player_{3-i}$. This MDP is denoted by $\game[\strat_{i}]$. We define its set of outcomes as 
\begin{equation*}
\outcomesMDP{\game}{\initState}{\strat_{i}} = \bigcup_{\strat_{3-i} \,\in\, \strats_{3-i}} \outcomesMDP{\game[\strat_{i}]}{\initState}{\strat_{3-i}} = \bigcup_{\strat_{3-i} \,\in\, \strats_{3-i}} \outcomesMC{\game[\strat_{1}, \strat_{2}]}{\initState}.
\end{equation*}

\smallskip\noindent\textbf{Attractors.} Given a game $\gameFull$, $\graphFull$, the \textit{attractor} for $\playerOne$ of a set $A \subseteq \states$ in $\game$ is denoted by $\attr_{\game}^{\playerOne}(A)$ and computed as the fixed point of the sequence
$\attr_{\game}^{\playerOne,\,n+1}(A) = \attr_{\game}^{\playerOne,\,n}(A) \cup \{s \in \states_{1} \,\vert\, \exists\, (s,t) \in \edges,\, t \in \attr_{\game}^{\playerOne,\,n}(A)\} \cup \{s \in \states_{2} \,\vert\, \forall\, (s,t) \in \edges,\, t \in \attr_{\game}^{\playerOne,\,n}(A)\}$, with $\attr_{\game}^{\playerOne,\,0}(A) = A$. It is exactly the set of states from which $\playerOne$ can ensure to reach $A$ no matter what $\playerTwo$ does. That is,
$\attr_{\game}^{\playerOne}(A) = \left\lbrace \state \in \states \mid \exists\, \strat_{1} \in \strats_{1}(\game),\, \forall\, \strat_{2} \in \strats_{2}(\game),\, \forall\, \play = \state_{0}\state_{1}\state_{2}\ldots{} \in \outcomesGame{\game}{\state}{\strat_{1}}{\strat_{2}},\, \state_{0} = \state,\, \exists\, i \in \nat,\, \state_{i} \in A\right\rbrace$.
We define symmetrically $\attr_{\game}^{\playerTwo}(A)$, the attractor for $\playerTwo$.

\smallskip\noindent\textbf{Subgraphs, subgames and sub-MDPs.} Given a graph $\graphFull$ and a subset of states $A \subseteq \states$, we define the induced subgraph $\graph \reduc A = (A, \edges \cap (A \times A), \weight)$ naturally. Subgames and sub-MDPs are defined similarly by considering their induced subgraphs. It is to note that subgames and sub-MDPs can only be properly defined if the induced subgraphs contain no deadlock and if the transition functions remain well-defined in the case of MDPs (i.e., if the probabilities on outgoing edges still sum up to one in all stochastic states of the sub-MDP).

\smallskip\noindent\textbf{Worst-case synthesis.} Given a two-player game $\gameFull$, with $\graphFull$ its underlying graph, an initial state $\initState \in \states$, a value function $\playValue\colon \plays{\graph} \rightarrow \realsAndInfinity$, and a rational threshold $\thresholdWC \in \rat$, the {\em worst-case threshold problem} asks to decide if $\playerOne$ has a strategy $\strat_{1} \in \strats_{1}$ such that
for all $\strat_{2} \in \strats_{2}$, for all $\play \in \outcomesGame{\game}{\initState}{\strat_{1}}{\strat_{2}}$, we have that $\playValue(\play) \geq \thresholdWC$.

For the mean-payoff value function, pure memoryless optimal\footnote{A strategy for $\player_{i}$, $i \in \{1, 2\}$, is said to be \textit{optimal} if it ensures a threshold higher or equal to the threshold ensured by any other strategy of the same player. The threshold ensured by an optimal strategy is called the \textit{optimal value}.} strategies exist for both players~\cite{liggett_SR69,EM79}. Hence, deciding the winner is in $\NPinter$, and it was furthermore shown to be in UP~$\cap$~coUP~\cite{ZP96,jurdzinski98,gawlitza2009}. Whether the problem is in $\PTIME$ is a long-standing open problem \cite{BCDGR11,Chatterjee201525}. Total-payoff value functions also yield pure memoryless optimal strategies for both players \cite{gimbert2004} and the associated decision problem is in UP~$\cap$~coUP~\cite{gawlitza2009}. For the truncated sum function, which can be seen as a particular instance of total-payoff, it can be shown that the decision problem takes polynomial time~\cite{DBLP:dblp_journals/mst/KhachiyanBBEGRZ08}, as a winning strategy of $\playerOne$ should avoid all cycles (because they yield strictly positive costs), hence usage of attractors and comparison of the worst possible sum of costs with the threshold suffices.

\smallskip\noindent\textbf{Expected value synthesis.} Given MDP $\markovProcessFull$, with $\graphFull$ its underlying graph, initial state $\initState \in \states$, measurable value function $\playValue\colon \plays{\graph} \rightarrow \realsAndInfinity$, and rational threshold $\thresholdExp \in \rat$, the {\em expected value threshold problem} asks to decide if $\playerOne$ has a strategy $\strat_{1} \in \strats_{1}$ such that $\expect_{\initState}^{\markovProcess[\strat_{1}]}(f) \geq \thresholdExp$.

Optimal expected mean-payoff in MDPs can be achieved by memoryless strategies, and the corresponding decision problem can be solved in polynomial time through linear programming~\cite{filar1997}. The truncated sum value function has been studied in the literature under the name of \textit{shortest path problem}: again, memoryless strategies suffice and the problem is solvable in polynomial time via linear programming \cite{bertsekas_MOR1991,deAlfaro_CONCUR1999}.

\section{Beyond Worst-Case Synthesis}
\label{sec:problem}

We here define the \textit{beyond worst-case synthesis problem}. Our goal is to study the synthesis of finite-memory strategies that, \textit{simultaneously}, ensure a value greater than some threshold $\thresholdWC$ in the worst-case situation (i.e., against any strategy of the adversary), and ensure an expected value greater than some threshold $\thresholdExp$ against a given finite-memory stochastic model of the adversary (e.g., representing commonly observed behavior of the environment). 

\begin{definition}
\label{def:bwc_problem}
Given a game $\gameFull$, with $\graphFull$ its underlying graph, an initial state $\initState \in \states$, a finite-memory stochastic model $\stratStoch \in \stratsFinite_{2}$ of the adversary, represented by a stochastic Moore machine, a measurable value function $\playValue\colon \plays{\graph} \rightarrow \realsAndInfinity$, and two rational thresholds $\thresholdWC, \thresholdExp \in \rat$, the {\em beyond worst-case ($\BWC$) problem} asks to decide if $\playerOne$ has a finite-memory strategy $\strat_{1} \in \stratsFinite_{1}$ such that
    \begin{numcases}{}
      \forall\, \strat_{2} \in \strats_{2},\, \forall\, \play \in \outcomesGame{\game}{\initState}{\strat_{1}}{\strat_{2}},\, \playValue(\play) > \thresholdWC\label{eq:thresholdWC}\\
      \expect_{\initState}^{\game[\strat_{1}, \stratStoch]}(f) > \thresholdExp\label{eq:thresholdExp}
    \end{numcases}
and the $\BWC$ synthesis problem asks to synthesize such a strategy if one exists.
\end{definition}
We take the convention to ask for values strictly greater than the thresholds in order to ease the formulation of our results in the following. Indeed, we will show that for some thresholds, it is possible to synthesize strategies that ensure $\varepsilon$-close values, for any $\varepsilon > 0$, while it is not feasible to achieve the exact threshold (Sect.~\ref{subsec:mpInfiniteMemory}). Using the strict inequality, we avoid tedious manipulation of such $\varepsilon$ in our proofs. Notice that we can assume $\thresholdExp > \thresholdWC$, otherwise the problem reduces to the classical worst-case analysis as follows. Assume $\thresholdWC \geq \thresholdExp$ and $\strat_{1}^{pm} \in \stratsPureMemoryless_{1}$ satisfies the worst-case threshold (recall memory is not necessary for the worst-case requirement alone). Consider the MC $\game[\strat_{1}^{pm}, \stratStoch]$. By Eq.~\eqref{eq:thresholdWC} and Lemma~\ref{lem:projBijective}, we have that for all $\play \in \outcomesMC{\game[\strat_{1}^{pm}, \stratStoch]}{\initState}$, $\playValue(\play) > \thresholdWC$. Hence, regardless of how the probability is defined in the MC, we have that $\expect_{\initState}^{\game[\strat_{1}^{pm}, \stratStoch]}(f) > \thresholdWC \geq \thresholdExp$ and Eq.~\eqref{eq:thresholdExp} is trivially satisfied.

\section{Mean-Payoff Value Function}
\label{sec:mean_payoff}

The first value function that we study is the mean-payoff. We present an algorithm, $\mpAlgoName$ (Alg.~\ref{alg:mp}), for deciding the corresponding $\BWC$ problem. Its cornerstones are highlighted in Sect.~\ref{subsec:mpApproach}. A running example is presented in Sect.~\ref{subsec:mpExample}. Sections~\ref{subsec:mpAssumptions} through~\ref{subsec:mpGlobal} are devoted to a detailed justification of this algorithm and the proof of its correctness. In Sect.~\ref{subsec:mpComplexity}, we prove that our algorithm implies $\NPinter$-membership and that it is optimal with regard to the complexity of the worst-case problem. Thus, the $\BWC$ framework for mean-payoff remarkably provides additional modeling power without negative impact on the complexity class. In Sect.~\ref{subsec:mpMemoryRequirements}, we prove that pseudo-polynomial memory is sufficient and in general necessary for finite-memory strategies satisfying the $\BWC$ mean-payoff problem (polynomial in the size of the game and the stochastic model, and in the values of weights and thresholds). Finally, we show in Sect.~\ref{subsec:mpInfiniteMemory} that infinite-memory strategies are strictly more powerful than finite-memory ones for $\playerOne$. This is in contrast with the worst-case and the expected value settings, where memoryless strategies always suffice.

\subsection{\textbf{The approach in a nutshell}}
\label{subsec:mpApproach}

Algorithm $\mpAlgoName$ is described in Alg.~\ref{alg:mp}. We give an intuitive sketch of its functioning in the following.

\smallskip\noindent\textbf{Inputs and outputs.} The algorithm takes as input: a game $\game^{i}$, a finite-memory stochastic model of the adversary $\strat_{2}^{i}$, a worst-case threshold $\thresholdWC^{i}$, an expected value threshold $\thresholdExp^{i}$, and an initial state $\initState^{i}$. Its output is $\yes$ if and only if there exists a finite-memory strategy of $\playerOne$ satisfying the $\BWC$ problem (Def.~\ref{def:bwc_problem}).

The output as described in Alg.~\ref{alg:mp} is boolean: the algorithm answers whether a satisfying strategy exists or not, but does not explicitely construct it (to avoid tedious formalization within the pseudocode). Nevertheless, we present how to synthesize such a winning strategy in Sect.~\ref{subsec:mpGlobal}. We sketch its operation in the following and we highlight the role of each step of the algorithm in the construction of this winning strategy, as producing a witness winning strategy is a straightforward by-product of the process we apply to decide satisfaction of the $\BWC$ problem.

\captionsetup[algorithm]{font=small}
\renewcommand{\algorithmicrequire}{\textbf{Input:}}
\renewcommand{\algorithmicensure}{\textbf{Output:}}

\begin{algorithm}[t]
\caption{$\mpAlgo{\game^{i}}{\strat_{2}^{i}}{\thresholdWC^{i}}{\thresholdExp^{i}}{\initState^{i}}$}
\label{alg:mp}
\begin{algorithmic}[1]
\small
\REQUIRE $\game^{i} = \left( \graph^{i}, \statesOne^{i}, \statesTwo^{\textit{i}}\right)$ a game, $\graph^{i} = \left( \states^{i}, \edges^{i}, \weight^{i}\right)$ its underlying graph, $\strat_{2}^{i} \in \stratsFinite_{2}(\game^{i})$ a finite-memory stochastic model of the adversary, $\mooreMachineFull{\strat_{2}^{i}}$ its Moore machine, $\thresholdWC^{i} = \frac{a}{b}, \thresholdExp^{i} \in \rat$, $\thresholdWC^{i} < \thresholdExp^{i}$, resp.~the worst-case and the expected value thresholds, and $\initState^{i} \in \states^{i}$ the initial state
\ENSURE $\yes$ if and only if $\playerOne$ has a finite-memory strategy  $\strat_{1} \in \stratsFinite_{1}(\game^{i})$ satisfying the $\BWC$ problem from $\initState^{i}$, for the thresholds pair $(\thresholdWC^{i}, \thresholdExp^{i})$ and the mean-payoff value function

\vspace{2mm}
\COMMENT{\textit{Preprocessing}}
\IF{$\thresholdWC^{i} \neq 0$} \label{alg:mp_thresholdsTest}
	\STATE Modify the weight function of $\graph^{i}$ s.t. $\forall\, \edge \in \edges^{i},\, \weight^{i}_{{\sf new}}(\edge) := b\cdot\weight^{i}(\edge) - a$, and consider the new thresholds pair $(0,\, \thresholdExp := b\cdot\thresholdExp^{i} -a)$\label{alg:mp_thresholds}
\ENDIF
\STATE Compute $\states_{\textit{WC}} := \left\lbrace \state \in \states^{i} \mid \exists\, \strat_{1} \in \strats_{1}(\game^{i}),\, \forall\, \strat_{2} \in \strats_{2}(\game^{i}),\, \forall\, \play \in \outcomesGame{\game^{i}}{\state}{\strat_{1}}{\strat_{2}},\, \mpay(\play) > 0\right\rbrace$\label{alg:mp_winningStates}
\IF{$\initState^{i} \not\in \states_{\textit{WC}}$}\label{alg:mp_losingNo}
	\RETURN $\no$\label{alg:mp_losingNoB}
\ELSE
	\STATE Let $\game^{w} := \game^{i} \reduc \states_{\textit{WC}}$ be the subgame induced by worst-case winning states\label{alg:mp_reduc}
	\STATE Build $\game := \game^{w} \otimes \mooreMachine{\strat_{2}^{i}} = (\graph, \statesOne, \statesTwo)$, $\graphFull$, $\states \subseteq \left( \states_{\textit{WC}} \times \mooreMem\right)$, the game obtained by product with the Moore machine, and $\initState := (\initState^{i}, \mooreInitMem)$ the corresponding initial state\label{alg:mp_memory}
	\STATE Let $\stratStoch \in \stratsMemoryless_{2}(\game)$ be the memoryless transcription of $\strat_{2}^{i}$ on $\game$\label{alg:mp_stochModel}
	\STATE Let $\markovProcess := \game[\stratStoch] = (\graph, \statesOne, \statesProb = \statesTwo, \mpTrans = \stratStoch)$ be the MDP obtained from $\game$ and $\stratStoch$\label{alg:mp_mdp}
\ENDIF

\vspace{2mm}
\COMMENT{\textit{Main algorithm}}
\STATE Compute $\maxWinningECs$ the set of maximal winning end-components of $\markovProcess$\label{alg:mp_main}
\STATE Build $\markovProcess' = (\graph', \statesOne, \statesProb, \mpTrans)$, where $\graph' = (\states, \edges, \weight')$ and $\weight'$ is defined as follows:
\begin{equation*}
\forall\, e = (\state_{1}, \state_{2}) \in \edges,\, \weight'(e) := \begin{cases}\weight(e) \text{ if } \exists\: \ec \in \maxWinningECs \text{ s.t. } \{\state_{1}, \state_{2}\} \subseteq \ec\\0 \text{ otherwise} \end{cases}
\end{equation*}\label{alg:mp_modifyWeights}
\STATE Compute the maximal expected value $\optimalExp$ from $\initState$ in $\markovProcess'$\label{alg:mp_maxExp}
\IF{$\optimalExp > \thresholdExp$}\label{alg:mp_maxExpComp}
	\RETURN $\yes$
\ELSE
	\RETURN $\no$
\ENDIF\label{alg:mp_main_end}
\end{algorithmic}
\end{algorithm}

\smallskip\noindent\textbf{Preprocessing.} The first part of the algorithm (lines~\ref{alg:mp_thresholdsTest}-\ref{alg:mp_mdp}) is a preprocessing of the game $\game^{i}$ and the stochatic model $\strat_{2}^{i}$ given as inputs in order to apply the second part of the algorithm (lines~\ref{alg:mp_main}-\ref{alg:mp_main_end}) on a modified game $\game$ and stochastic model $\stratStoch$, simpler to manipulate. We show in the following that the answer to the $\BWC$ problem on the modified game is $\yes$ if and only if it is also $\yes$ on the input game, and we present how a winning strategy of $\playerOne$ in $\game$ can be transferred to a winning strategy in $\game^{i}$.

The preprocessing is composed of four main steps. First, we modify the weight function of $\graph^{i}$ in order to consider the equivalent $\BWC$ problem with thresholds $(0,\, \thresholdExp)$ instead of $(\thresholdWC^{i},\, \thresholdExp^{i})$. This classical trick is used to get rid of explicitely considering the worst-case threshold in the following, as it is equal to zero. 
Second, observe that any strategy that is winning for the $\BWC$ problem must also be winning for the classical worst-case problem. Such a strategy cannot allow visits of any state from which $\playerOne$ cannot ensure winning against an antagonistic adversary because mean-payoff is a prefix-independent (i.e., for all $\prefix \in \prefixes{\graph}$, $\play \in \plays{\graph}$, we have that $\mpay(\prefix\cdot\play) = \mpay(\play)$) objective (hence it is not possible to ``win'' it over the finite prefix up to such a state). Hence, we reduce our study to $\game^{w}$, the subgame induced by worst-case winning states in $\game^{i}$ (lines~\ref{alg:mp_winningStates} and~\ref{alg:mp_reduc}). Obviously, if from the initial state $\initState^{i}$, $\playerOne$ cannot win the worst-case problem, then the answer to the $\BWC$ problem is $\no$ (lines~\ref{alg:mp_losingNo}-\ref{alg:mp_losingNoB}).
Third, we build the game $\game$ which states are defined by the product of the states of $\game^{w}$ and the memory elements of the Moore machine $\mooreMachine{\strat_{2}^{i}}$ (line~\ref{alg:mp_memory}). Intuitively, we expand the initial game by integrating the memory of the stochastic model of $\playerTwo$ in the graph. Note that this does not modify the power of the adversary. Fourth, the finite-memory stochastic model $\strat^{i}_{2}$ on $\game^{i}$ clearly translates to a memoryless stochastic model $\stratStoch$ on $\game$ (line~\ref{alg:mp_stochModel}). This will help us obtain elegant proofs for the second part of the algorithm.

\smallskip\noindent\textbf{Analysis of end-components.} The second part of the algorithm (lines~\ref{alg:mp_main}-\ref{alg:mp_main_end}) hence operates on a game $\game$ such that from all states, $\playerOne$ has a strategy to achieve a strictly positive mean-payoff value (recall $\thresholdWC = 0$). We consider the MDP $\markovProcess = \game[\stratStoch]$ and notice that the underlying graphs of $\game$ and $\markovProcess$ are the same thanks to $\stratStoch$ being memoryless. The following steps analyze \textit{end-components} in the MDP, i.e., strongly connected subgraphs in which $\playerOne$ can ensure to stay when playing against the stochastic adversary (cf. Sect.~\ref{sec:preliminaries}).

The motivation to the analysis of ECs is the following. It is well-known that under any arbitrary strategy $\strat_{1} \in \strats_{1}$ of $\playerOne$ in~$\markovProcess$, the probability that states visited infinitely often along an outcome constitute an EC is one~\cite{courcoubetis_JACM1995,de1997formal}. Recall that the mean-payoff is prefix-independent, therefore the value of any outcome only depends on those states that are seen infinitely often. Hence, the expected mean-payoff in $\markovProcess[\strat_{1}]$ depends \textit{uniquely} on the value obtained in the ECs. Inside an EC, we can compute the maximal expected value that can be achieved by $\playerOne$, and this value is the same in all states of the EC~\cite{filar1997}.

Consequently, in order to satisfy the expected value requirement (Eq.~\eqref{eq:thresholdExp}), an acceptable strategy for the $\BWC$ problem has to favor reaching ECs with a sufficient expectation, but under the constraint that it should also ensure satisfaction of the worst-case requirement (Eq.~\eqref{eq:thresholdWC}). As we will show, this constraint implies that some ECs with high expected values may still need to be avoided because they do not permit to guarantee the worst-case requirement. This is the cornerstone of the classification of ECs that follows.

\smallskip\noindent\textbf{Classification of end-components.} Let $\ecsSet \subseteq 2^{\states}$ be the set of all ECs in MDP $\markovProcess$. By definition, only edges in $\edgesNonZero$, as defined in Sect.~\ref{sec:preliminaries}, are involved to determine which sets of states form an EC in $\markovProcess$. As such, for any EC $\ec \in \ecsSet$, there may exist edges from $\edges \setminus \edgesNonZero$ starting in $\ec$, such that $\playerTwo$ can force leaving $\ec$ when using an arbitrary strategy. Still these edges will never be used by the stochastic model $\stratStoch$. This remark will be important to the definition of strategies of $\playerOne$ that guarantee the worst-case requirement, as $\playerOne$ needs to be able to react to the hypothetic use of such an edge. We will see in the following that it is also the case \textit{inside} an EC.

Now, we want to consider the ECs in which $\playerOne$ can ensure that the worst-case requirement will be fulfilled (i.e., without having to leave the EC): we call them \textit{winning} ECs. Indeed, the others will need to be eventually avoided, hence will have zero impact on the expectation of a finite-memory strategy satisfying the $\BWC$ problem. So we call the latter \textit{losing} ECs. The subtlety of this classication is that it involves considering the ECs both in the MDP~$\markovProcess$, and in the game~$\game$. Formally, let $\ec \in \ecsSet$ be an EC. It is \textit{winning} if, in the subgame $\game \reduc \ec$, from all states, $\playerOne$ has a strategy to ensure a strictly positive mean-payoff against any strategy of $\playerTwo$ \textit{that only chooses edges which are assigned non-zero probability by $\stratStoch$}, or equivalently, edges in $\edgesNonZero$. We denote $\winningECs \subseteq \ecsSet$ the set of such ECs. Non-winning ECs are \textit{losing}: in those, whatever the strategy of $\playerOne$ played against the stochastic model $\stratStoch$ (or any strategy with the same support), there exists at least one outcome for which the mean-payoff is not strictly positive (even if its probability is zero, its mere existence is not acceptable for the worst-case requirement).

\smallskip\noindent\textbf{Maximal winning end-components.} Based on these definitions, see that line~\ref{alg:mp_main} of algorithm $\mpAlgoName$ does not actually compute the set $\winningECs$ containing all winning ECs, but rather the set $\maxWinningECs \subseteq \winningECs$, defined as $\maxWinningECs = \{\ec \in \winningECs \mid \forall\, \ec' \in \winningECs,\, \ec \subseteq \ec' \Rightarrow \ec = \ec'\}$, i.e., the set of \textit{maximal} winning ECs.

The intuition on \textit{why we can} restrict our study to this subset is as follows. If an EC $\ec_{1} \in \winningECs$ is included in another EC $\ec_{2} \in \winningECs$, i.e., $\ec_{1} \subseteq \ec_{2}$, we have that the maximal expected value achievable in $\ec_{2}$ is at least equal to the one achievable in~$\ec_{1}$. Indeed, $\playerOne$ can reach $\ec_{1}$ with probability one (by virtue of $\ec_{2}$ being an EC and $\ec_{1} \subseteq \ec_{2}$) and stay in it forever with probability one (by virtue of $\ec_{1}$ being an EC): hence the expectation would be equal to what can be obtained in~$\ec_{1}$ thanks to the prefix-independence of the mean-payoff. This property implies that it is sufficient to consider maximal winning ECs in our computations.

As for \textit{why we do it}, observe that the complexity gain is critical. The number of winning ECs can be as large as~$\vert\winningECs\vert \leq \vert\ecsSet\vert \leq 2^{\vert\states\vert}$, that is, exponential in the size of the input. Yet, the number of maximal winning ECs is bounded by $\vert\maxWinningECs\vert \leq \vert\states\vert$ as they are disjoint by definition. Indeed, for any two winning ECs with a non-empty intersection, their union also constitutes an EC, and is still winning because $\playerOne$ can essentially stick to the EC of his choice. The computation of the set $\maxWinningECs$ is executed by a recursive subalgorithm calling polynomially-many times an $\NPinter$ oracle solving the worst-case threshold problem. Roughly sketched, this algorithm computes the maximal end-component decomposition of an MDP (in polynomial time~\cite{DBLP:journals/jacm/ChatterjeeH14}), then checks for each EC $\ec$ in the decomposition (their number is polynomial) if $\ec$ is winning or not, which requires a call to an $\NPinter$ oracle solving the worst-case threshold problem on the corresponding subgame. If $\ec$ is losing, it may still be the case that a sub-EC $\ec' \subsetneq \ec$ is winning. Therefore we recurse on the MDP reduced to $\ec$, where states from which $\playerTwo$ can win in $\ec$ have been removed (they are a no-go for~$\playerOne$). Hence the stack of calls is also at most polynomial.

\smallskip\noindent\textbf{Ensure reaching winning end-components.} As discussed, under any arbitrary strategy of $\playerOne$, states visited infinitely often form an EC with probability one. Now, if we take a \textit{finite-memory} strategy that \textit{satisfies} the $\BWC$ problem (Def.~\ref{def:bwc_problem}), we can precise this result and state that they form a \textit{winning} EC with probability one. Equivalently, we have that the probability that an outcome $\play$ is such that $\infVisited{\play} = \ec$ for some $\ec \in \ecsSet \setminus \winningECs$ is zero. The equality is crucial. It may be the case, with non-zero probability, that $\infVisited{\play} = \ec' \subsetneq \ec$ for some $\ec' \in \winningECs$ and $\ec \in \ecsSet \setminus \winningECs$ (hence the recursive algorithm to compute $\maxWinningECs$). It is clear that $\playerOne$ should not visit all the states of a losing EC forever, as then he would not be able to guarantee the worst-case threshold inside the corresponding subgame (we show in Sect.~\ref{subsec:mpInfiniteMemory} that with infinite memory, there may still be some incentive to stay in a losing EC).

We denote $\negligibleStates = \states \setminus \bigcup_{\ec \in \maxWinningECs} \ec$ the set of states that, with probability one, are only seen a finite number of times when a $\BWC$ satisfying strategy is played, and call them \textit{negligible} states.

Our ultimate goal here is to build a modified MDP $\markovProcess'$, sharing the same graph and ECs as $\markovProcess$, such that a classical optimal strategy for the expected value problem on $\markovProcess'$ will naturally avoid losing ECs and prescribe which winning ECs are the most interesting to reach for a $\BWC$ strategy on the initial game $\game$ and MDP $\markovProcess$. Observe that the expected value obtained in $\markovProcess$ by any $\BWC$ satisfying strategy of $\playerOne$ only depends on the weights of edges involved in winning ECs, or equivalently, in maximal winning ECs (as the set of outcomes that are not trapped in them has measure zero). Consequently, we build $\markovProcess'$ by modifying the weight function of $\markovProcess$ (line~\ref{alg:mp_modifyWeights}): we keep the weights unchanged in edges that belong to some $\ec \in \maxWinningECs$, and we put them to zero everywhere else, i.e., on any edge involving a negligible state. Weight zero is taken because it is lower than the expectation granted by winning ECs, which is strictly greater than zero by definition.

\smallskip\noindent\textbf{Reach the highest valued winning end-components.} We compute the maximal expected mean-payoff $\thresholdExp^{\ast}$ that can be achieved by $\playerOne$ in the MDP $\markovProcess'$, from the corresponding initial state (line~\ref{alg:mp_maxExp}). This computation takes polynomial time and memoryless strategies suffice to achieve the maximal value~\cite{filar1997}.

As discussed before, such a strategy reaches an EC of $\markovProcess'$ with probability one. Basically, we build a strategy that favors reaching ECs with high associated expectations in $\markovProcess'$. We argue that the ECs reached with probability one by this strategy are necessarily winning ECs. Clearly, if a winning EC is reachable instead of a losing one, it will be favored because of the weights definition in $\markovProcess'$ (expectation is strictly higher in winning ECs). It remains to check if the set of winning ECs is reachable with probability one from any state in $\states$. That is the case because of the preprocessing. Indeed, we know that all states are winning for the worst-case requirement. Clearly, from any state in $A = \states \setminus \bigcup_{\ec \in \ecsSet} \ec$, $\playerOne$ cannot ensure to stay in $A$ (otherwise it would form an EC) and thus must be able to win the worst-case requirement from reached ECs. Now for any state in $B = \bigcup_{\ec \in \ecsSet} \ec \setminus \bigcup_{\ec \in \maxWinningECs} \ec$, i.e., states in losing ECs and not in any sub-EC winning, $\playerOne$ cannot win the worst-case by staying in $B$, by definition of losing EC. Since we know $\playerOne$ can ensure the worst-case by hypothesis, he must be able to reach $C = \bigcup_{\ec \in \maxWinningECs} \ec$ from any state in $B$, as claimed.

\smallskip\noindent\textbf{Inside winning end-components.} Based on that, winning ECs are reached with probability one. Let us first consider what we can say about such ECs if we assume that $\edgesNonZero = \edges$, i.e., if the stochastic model maps all possible edges to non-zero probabilities. 
We establish a finite-memory \textit{combined strategy} of $\playerOne$ that ensures~(i)~worst-case satisfaction while yielding (ii) an expected value $\varepsilon$-close to the maximal expectation inside the component. For two well-chosen parameters $\stepsExp, \stepsWC \in \nat$, it is informally defined as follows: in phase~$\typeA$, play a memoryless expected value optimal strategy for $\stepsExp$ steps and memorize $\cmbSum \in \integ$, the sum of weights along these steps; in phase $\typeB$, if $\cmbSum > 0$, go to~$\typeA$, otherwise play a memoryless worst-case optimal strategy for $\stepsWC$ steps, then go to $\typeA$. In phases $\typeA$, $\playerOne$ tries to increase its expectation and approach its optimal one, while in phase $\typeB$, he compensates, if needed, losses that occured in phase $\typeA$. The two memoryless strategies exist on the subgame induced by the EC: by definition of ECs, based on~$\edgesNonZero$, the stochastic model of $\playerTwo$ will never be able to force leaving the EC against the combined strategy.

A key result of our paper is the existence of values for $\stepsExp$ and~$\stepsWC$ such that (i) and (ii) are verified. We see plays as sequences of periods, each starting with phase~$\typeA$. First, for any $\stepsExp$, it is possible to define $\stepsWC(\stepsExp)$ such that any period composed of phases~$\typeA+\typeB$ ensures a mean-payoff at least $1/(\stepsExp+\stepsWC) > 0$. Periods containing only phase $\typeA$ trivially induce a mean-payoff at least~$1/\stepsExp$ as they are not followed by phase $\typeB$. Both rely on the weights being integers. As the length of any period is bounded by $(\stepsExp+\stepsWC)$, the inequality remains strict for the mean-payoff of any play, implying~(i). Now, consider parameter $\stepsExp$. Clearly, when~$\stepsExp \rightarrow \infty$, the expectation over a phase $\typeA$ tends to the optimal one. Nevertheless, phases~$\typeB$ also contribute to the overall expectation of the combined strategy, and (in general) lower it so that it is strictly less than the optimal for any $\stepsExp, \stepsWC \in \nat$. Hence to prove (ii), we not only need that the probability of playing phase $\typeB$ decreases when $\stepsExp$ increases, but also that it decreases faster than the increase of $\stepsWC$, needed to ensure~(i), so that overall, the contribution of phases~$\typeB$ tends to zero when $\stepsExp \rightarrow \infty$. This is indeed the case and is proved using results bounding the probability of observing a mean-payoff significantly (more than some $\varepsilon$) different than the optimal expectation along a phase $\typeA$ of length $\stepsExp \in \nat$: this probability decreases exponentially when $\stepsExp$ increases~\cite{tracol_ORL2009,glynn_SPL2002} (related to the notions of Chernoff bounds and Hoeffding's inequality in MCs), while $\stepsWC$ only needs to be polynomial in $\stepsExp$.

Now, consider what happens if $\edgesNonZero \subsetneq \edges$. Then, if $\playerTwo$ uses an arbitrary strategy, he can take edges of probability zero, i.e., in $\edges \setminus \edgesNonZero$, either staying in the EC, or leaving it. In both cases, this must be taken into account in order to satisfy Eq.~\eqref{eq:thresholdWC} as it may involve dangerous weights (recall that zero-probability edges are not considered when an EC is classified as winning or not). Fortunately, if this were to occur, $\playerOne$ could switch to a worst-case winning memoryless strategy, which exists in all states thanks to the preprocessing, to preserve the worst-case requirement. Regarding the expected value (Eq.~\eqref{eq:thresholdExp}), this has no impact as it occurs with probability zero against $\stratStoch$. The strategy to follow in winning ECs hence adds this reaction procedure to the combined strategy: we call it the \textit{witness-and-secure strategy}.

\smallskip\noindent\textbf{Global strategy synthesis.} In summary, (a) losing ECs should be avoided and will be by a strategy that optimizes the expectation on the MDP~$\markovProcess'$; (b) in winning ECs, $\playerOne$ can obtain ($\varepsilon$-closely) the expectation of the EC \textit{and} ensure the worst-case threshold.

Hence, we finally compare the value $\thresholdExp^{\ast}$ with the expected value threshold $\thresholdExp$ (line~\ref{alg:mp_maxExpComp}): (i) if it is strictly higher, we conclude that there exists a finite-memory strategy satisfying the $\BWC$ problem, and (ii) if it is not, we conclude that there does not exist such a strategy.

To prove (i), we establish a finite-memory strategy in $\game$, called \textit{global strategy}, of $\playerOne$ that ensures a strictly positive mean-payoff against an antagonistic adversary, and ensures an expected mean-payoff $\varepsilon$-close to $\thresholdExp^{\ast}$ (hence $> \thresholdExp$) against the stochastic adversary modeled by $\stratStoch$ (i.e., in $\markovProcess$). The intuition is as follows. We play the memoryless optimal strategy of the MDP~$\markovProcess'$ for a sufficiently long time, defined by a parameter $\stepsGlobal \in \nat$, in order to be with probability close to one in a winning EC (the convergence is exponential by results on absorption times in MCs~\cite{grinstead_AMS1997}). Then, if we are inside a winning EC, we switch to the witness-and-secure strategy which, as sketched in the previous paragraph, ensures the worst-case and the expectation thresholds. If we are not yet in a winning EC, then we switch to a worst-case winning strategy in $\game$, which always exists by hypothesis. Thus the mean-payoff of plays that do not reach winning ECs is strictly positive. Since in winning ECs we are $\varepsilon$-close to the maximal expected value of the EC, we can conclude that it is possible to play the optimal expectation strategy of MDP $\markovProcess'$ for sufficiently long to obtain an overall expected value which is arbitrarily close to $\thresholdExp^{\ast}$, and still guarantee the worst-case threshold in all outcomes.

To prove (ii), it suffices to understand that only ECs have an impact on the expectation, and that losing ECs cannot be used forever without endangering the worst-case requirement. Note that given a winning strategy on $\game$, it is possible to build a corresponding winning strategy on $\game^{i}$ by reintegrating the memory elements of the Moore machine in the memory of the winning strategy of $\playerOne$.

\smallskip\noindent\textbf{Complexity bounds.} The input size of the algorithm depends on the size of the game, the size of the Moore machine for the stochastic model, and the encodings of weights and thresholds. We can prove that all computing steps require (deterministic) polynomial time except for calls to an oracle solving the worst-case threshold problem, which is in $\NPinter$~\cite{ZP96,jurdzinski98} and not known to be in $\PTIME$. Hence, the overall complexity of the BWC problem is $\NPinter$ and may collapse to $\PTIME$ if the worst-case problem were to be proved in $\PTIME$.

We also establish that the $\BWC$ problem is at least as difficult as the worst-case problem thanks to a polynomial time reduction from the latter to the former. Thus, $\mpAlgoName$ membership to $\NPinter$ can be seen as optimal regarding our current knowledge of the worst-case threshold problem.

\begin{theorem}
\label{thm:mp_decisionProblem}
The beyond worst-case problem for the mean-payoff value function is in $\NPinter$ and at least as hard as deciding the winner in mean-payoff games.
\end{theorem}

\begin{remark}[approximation of the optimal value]
Given a worst-case threshold $\thresholdWC \in \rat$, a natural question is whether we can
maximize the expectation of finite-memory strategies that satisfy this threshold. 
However, there is no best expectation value in general, as increasing the size of the memory
may also strictly increase the expectation. Nevertheless, the least upper bound
of all the expected value thresholds that can be achieved by finite-memory strategies can 
be approached up to an $\varepsilon$, for all $\varepsilon>0$. Formally,
assume that the worst-case threshold $\mu$ can be satisfied, and let
 $\thresholdExp_{\top}$ be the least upper bound of the set $\{ \thresholdExp \in \rat \mid \exists\, \strat_{1} \in \stratsPureFinite_{1}$ that
satisfies the $\BWC$ problem for thresholds $(\thresholdWC, \thresholdExp) \}$ (it exists since this set is trivially bounded by $W$ and it is non-empty).
Without knowing $\thresholdExp_{\top}$ a priori, it is still possible to approach it $\varepsilon$-closely, for
all $\varepsilon > 0$, by a dichotomic search with a polynomial (in $\bits = \log_{2} \largestW$, the length of the encoding of weights) number of steps, initialized to the interval $[\mu, W]$.
\end{remark}

\subsection{\textbf{Running example}}
\label{subsec:mpExample}

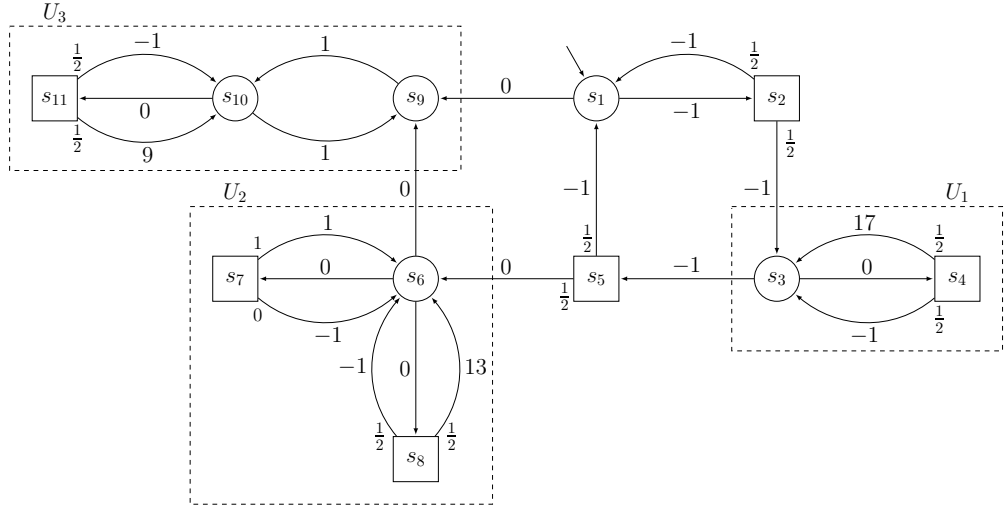
\begin{figure}[htb]
  \centering   
  \scalebox{0.6}{\begin{tikzpicture}[->,>=latex,shorten >=1pt,auto,node
    distance=2.5cm,bend angle=45,font=\Large]
    \tikzstyle{p1}=[draw,circle,text centered,minimum size=10mm]
    \tikzstyle{p2}=[draw,rectangle,text centered,minimum size=10mm]
    \tikzstyle{empty}=[]
    \node[p1] (1) at (0,0) {$\state_{9}$};
    \node[p1] (2) at (4,0) {$\state_{1}$};
    \node[p2] (3) at (8,0) {$\state_{2}$};
    \node[p1] (4) at (8,-4) {$\state_{3}$};
    \node[p2] (5) at (12,-4) {$\state_{4}$};
    \node[p2] (6) at (4,-4) {$\state_{5}$};
    \node[p1] (7) at (0,-4) {$\state_{6}$};
    \node[p2] (8) at (-4,-4) {$\state_{7}$};
    \node[p1] (9) at (-4,0) {$\state_{10}$};
    \node[p2] (10) at (-8,0) {$\state_{11}$};
    \node[p2] (11) at (0,-8) {$\state_{8}$};
    \node[empty] (swec) at (-8, 1.9) {$\ec_{3}$};
    \node[empty] (wwec) at (-4, -2.1) {$\ec_{2}$};
    \node[empty] (lec) at (12, -2.1) {$\ec_{1}$};
    \node[empty] (proba5a) at (11.6, -3.1) {$\frac{1}{2}$};
    \node[empty] (proba5b) at (11.6, -4.9) {$\frac{1}{2}$};
    \node[empty] (proba3a) at (8.3, -1) {$\frac{1}{2}$};
    \node[empty] (proba3b) at (7.5, 0.95) {$\frac{1}{2}$};
    \node[empty] (proba8a) at (-3.5, -3.2) {{\large $1$}};
    \node[empty] (proba8b) at (-3.5, -4.8) {{\large $0$}};
    \node[empty] (proba10a) at (-7.5, 0.9) {$\frac{1}{2}$};
    \node[empty] (proba10b) at (-7.5, -0.9) {$\frac{1}{2}$};
    \node[empty] (proba11a) at (0.8, -7.5) {$\frac{1}{2}$};
    \node[empty] (proba11b) at (-0.8, -7.5) {$\frac{1}{2}$};
    \node[empty] (proba6a) at (3.8, -3.1) {$\frac{1}{2}$};
    \node[empty] (proba6b) at (3.3, -4.4) {$\frac{1}{2}$};
    \coordinate[shift={(-3mm,8mm)}] (init) at (2.north west);
    \path
    (2) edge node[above] {$0$} (1)
    (6) edge node[above] {$0$} (7)
    (4) edge node[above] {$-1$} (6)
    (3) edge node[left] {$-1$} (4)
    (6) edge node[left] {$-1$} (2)
    (4) edge node[above] {$0$} (5)
    (7) edge node[above] {$0$} (8)
    (7) edge node[left] {$0$} (1)
    (init) edge (2)
    ;
	\draw[->,>=latex] (3) to[out=140,in=40] node[above] {$-1$} (2);
	\draw[->,>=latex] (2) to[out=0,in=180] node[below] {$-1$} (3);
	\draw[->,>=latex] (5) to[out=140,in=40] node[above] {$17$} (4);
	\draw[->,>=latex] (5) to[out=220,in=320] node[below] {$-1$} (4);
	\draw[->,>=latex] (8) to[out=40,in=140] node[above] {$1$} (7);
	\draw[->,>=latex] (8) to[out=320,in=220] node[below] {$-1$} (7);
	\draw[->,>=latex] (1) to[out=140,in=40] node[above] {$1$} (9);
	\draw[->,>=latex] (9) to[out=320,in=220] node[below] {$1$} (1);
	\draw[->,>=latex] (9) to[out=180,in=0] node[below] {$0$} (10);
	\draw[->,>=latex] (10) to[out=40,in=140] node[above] {$-1$} (9);
	\draw[->,>=latex] (10) to[out=320,in=220] node[below] {$9$} (9);
	\draw[->,>=latex] (7) to[out=270,in=90] node[left] {$0$} (11);
	\draw[->,>=latex] (11) to[out=130,in=230] node[left] {$-1$} (7);
	\draw[->,>=latex] (11) to[out=50,in=310] node[right] {$13$} (7);
	\draw[dashed,-] (-9,1.6) -- (1,1.6) -- (1,-1.6) -- (-9,-1.6) -- (-9,1.6);
	\draw[dashed,-] (7,-2.4) -- (13,-2.4) -- (13,-5.6) -- (7,-5.6) -- (7,-2.4);
	\draw[dashed,-] (-5,-2.4) -- (1.7,-2.4) -- (1.7,-9) -- (-5,-9) -- (-5,-2.4);
      \end{tikzpicture}}
      \caption{Mean-payoff game with maximal winning ECs $\ec_{2}$ and $\ec_{3}$. End-component $\ec_{1}$ is losing.}
\label{fig:mpRunningExample}
  \end{figure}

In order to illustrate several notions and strategies, we will consider the game depicted in Fig.~\ref{fig:mpRunningExample} throughout Sect.~\ref{sec:mean_payoff}. States of $\playerOne$ are represented by circles and states of $\playerTwo$ by squares. The stochastic model of $\playerTwo$ is memoryless and is described by the probabilities written close to the start of outgoing edges. For example, in $\state_{2}$, the stochastic model chooses edge $(\state_{2}, \state_{1})$ with probability $1/2$ and edge $(\state_{2}, \state_{3})$ with probability $1/2$. Each edge is assigned a weight represented by an integer number midway. Formally, our definition of the set~$\edges$ allows only one edge from any given state $\state$ to any state $\state'$, hence asks that multiple edges with different values be split by adding dummy states (states with exactly one ingoing edge and one outgoing edge). Note that in order to preserve the same mean-payoff values for paths in the graph, we need to split every edge and copy its weight in both halfs. This restriction is w.l.o.g. and applied for the sake of readability in technical proofs. Still, our graphical representation oversteps it to maintain compactness.

We consider the $\BWC$ problem with the worst-case threshold $\thresholdWC = 0$. Observe that this game satisfies the assumptions guaranteed at the end of the preprocessing part of the algorithm. That is, the worst-case threshold is zero, a worst-case winning strategy of $\playerOne$ exists in all states (e.g., the memoryless strategy choosing edges $(\state_{1}, \state_{9})$, $(\state_{3}, \state_{5})$, $(\state_{6}, \state_{9})$, $(\state_{9}, \state_{10})$ and $(\state_{10}, \state_{9})$ in their respective starting states), and the stochastic model is memoryless, as explained above.

\subsection{\textbf{Preprocessing - simplifying assumptions}}
\label{subsec:mpAssumptions}

We discuss here the preprocessing part of the algorithm (lines~\ref{alg:mp_thresholdsTest}-\ref{alg:mp_mdp}). The goal is to be able to execute the main algorithm (lines~\ref{alg:mp_main}-\ref{alg:mp_main_end}) with the following hypotheses: (a) the worst-case threshold is zero, (b) in all states, $\playerOne$ has a strategy to satisfy the worst-case requirement, and (c) the stochastic model of the adversary is memoryless. This preprocessing is sound and complete: $\playerOne$ has a strategy for the $\BWC$ problem in the input $\game^{i}$ (for thresholds $(\thresholdWC^{i}, \thresholdExp^{i})$ and against stochastic model $\strat_{2}^{i}$) if and only if he also has one in the preprocessed game $\game$  (for thresholds $(0, \thresholdExp)$ and against stochastic model $\stratStoch$). We break down the proof in three lemmas, one for each hypothesis.

\smallskip\noindent\textbf{Thresholds.} First, Lemma~\ref{lem:mp_thresholdChange} states that the worst-case threshold can be taken equal to zero thanks to a slight modification of the weight function (lines~\ref{alg:mp_thresholdsTest}-\ref{alg:mp_thresholds}). From now on, we thus assume that $\thresholdWC = 0$.

\begin{lemma}
\label{lem:mp_thresholdChange}
Let $\gameFull$ be a two-player game, $\graphFull$ its underlying graph, $\initState \in \states$ the initial state, $\strat_{2}^{f} \in \stratsFinite_{2}$ a finite-memory stochastic model of $\playerTwo$, and $(\thresholdWC = \frac{a}{b}, \thresholdExp) \in \rat^{2}$ a pair of thresholds, with $a \in \integ$ and $b \in \natStrict$. Then $\playerOne$ has a satisfying strategy for the $\BWC$ mean-payoff problem in $\game$ if and only if $\playerOne$ has a satisfying strategy when considering the thresholds pair $(0, \thresholdExp' = b\cdot\thresholdExp - a)$ and the weight function $\weight'$ such that $\forall\, \edge \in \edges,\, \weight'(\edge) = b\cdot\weight(\edge) - a$.
\end{lemma}

\begin{proof}
Recall the mean-payoff of a play is defined as the (infimum) limit of the mean weight over prefixes of increasing length. Hence, the affine transformation applied on weights carries over to play values: for all $\play \in \plays{\graph}$, we have that $\mpay_{\weight'}(\play) = b\cdot\mpay_{\weight}(\play) - a$, where the subscript denotes which weight function we consider. Let $\strat_{1} \in \strats_{1}$ be a strategy of $\playerOne$. First, consider the worst-case requirement of Eq.~\eqref{eq:thresholdWC}. We claim that the following equivalence holds:
$\forall\, \strat_{2} \in \strats_{2},\, \forall\, \play \in \outcomesGame{\game}{\initState}{\strat_{1}}{\strat_{2}},\, \mpay_{\weight}(\play) > \thresholdWC = \frac{a}{b} \;\Longleftrightarrow\; \forall\, \strat_{2} \in \strats_{2},\, \forall\, \play \in \outcomesGame{\game}{\initState}{\strat_{1}}{\strat_{2}},\, \mpay_{\weight'}(\play) > 0$.
This is trivial by applying the affine transformation. Second, consider the expected value requirement of Eq.~\eqref{eq:thresholdExp}. We claim that
\[\expect_{\initState}^{\game[\strat_{1}, \strat_{2}^{f}]}(\mpay_{\weight}) > \thresholdExp \;\Longleftrightarrow\; \expect_{\initState}^{\game[\strat_{1}, \strat_{2}^{f}]}(\mpay_{\weight'}) > \thresholdExp' = b\cdot\thresholdExp - a.\]
The expected value operator is well-known to be linear. Thus, by extracting the affine transformation, we obtain that $\expect_{\initState}^{\game[\strat_{1}, \strat_{2}^{f}]}(\mpay_{\weight'}) = b\cdot\expect_{\initState}^{\game[\strat_{1}, \strat_{2}^{f}]}(\mpay_{\weight}) - a$, which proves our point and concludes the proof.
\end{proof}

\smallskip\noindent\textbf{Worst-case winning.} Second, we prove in Lemma~\ref{lem:mp_noLosingStates} that a necessary condition for a strategy to satisfy the $\BWC$ problem is to avoid visiting states that are losing for the worst-case mean-payoff requirement. This justifies lines~\ref{alg:mp_winningStates}-\ref{alg:mp_reduc} of the algorithm. Observe that the graph of $\game^{w} = \game^{i} \reduc \states_{\textit{WC}}$ contains no deadlock as otherwise it would contradict the fact that $\playerOne$ can satisfy the worst-case threshold problem from states in $\states_{\textit{WC}}$ in the game $\game^{i}$. Also note that $\game^{w}[\strat_{2}^{i}]$ remains a well-defined MDP as there exists no edge from states $\state \in \states_{2}^{i} \cap \states_{\textit{WC}}$ to states $\state' \in \states \setminus \states_{\textit{WC}}$, otherwise $\playerTwo$ could win the game from $\state$ by reaching $\state'$, by prefix-independence of the mean-payoff, and so $\state$ would no belong to $\states_{\textit{WC}}$.

\begin{lemma}
\label{lem:mp_noLosingStates}
Let $\gameFull$ be a two-player game, $\graphFull$ its underlying graph, $\initState \in \states$ the initial state, $\strat_{2}^{f} \in \stratsFinite_{2}$ a finite-memory stochastic model of $\playerTwo$, and $\thresholdExp \in \rat$ the expected value threshold. Let $\states_{\textit{WC}} \subseteq \states$ be the set of winning states for $\playerOne$ for the worst-case threshold problem. If $\strat_{1} \in \stratsFinite_{1}$ is a satisfying strategy for the $\BWC$ mean-payoff problem, then
$\forall\, \strat_{2} \in \strats_{2},\, \forall\, \play \in \outcomesGame{\game}{\initState}{\strat_{1}}{\strat_{2}},\, \forall\, \state \in \states \setminus \states_{\textit{WC}},\, \state \not\in \play$.
\end{lemma}

\begin{proof}
Formally, let $\states_{\textit{WC}} := \left\lbrace \state \in \states \mid \exists\, \strat_{1} \in \strats_{1},\, \forall\, \strat_{2} \in \strats_{2},\, \forall\, \play \in \outcomesGame{\game}{\state}{\strat_{1}}{\strat_{2}},\, \mpay(\play) > 0\right\rbrace$. We claim that a winning strategy of $\playerOne$ for the $\BWC$ problem must avoid states that are not in $\states_{\textit{WC}}$ and prove it by contradiction. Indeed, let $\strat_{1} \in \stratsFinite_{1}$ be such a strategy. Assume the claim is false: there exist $\strat_{2} \in \strats_{2}$, $\play \in \outcomesGame{\game}{\initState}{\strat_{1}}{\strat_{2}}$ and $\state \in \states \setminus \states_{\textit{WC}}$ such that $\state \in \play$. By definition of $\states_{\textit{WC}}$ and determinacy of mean-payoff games~\cite{EM79,martin_AM75}, it is the case that in state~$\state$, $\playerTwo$ has a winning strategy for the worst-case objective: there exists $\strat'_{2} \in \strats_{2}$ such that for all $\strat'_{1} \in \strats_{1}$, there exists $\play' \in \outcomesGame{\game}{\state}{\strat'_{1}}{\strat'_{2}}$ such that $\neg(\mpay(\play') > 0)$, i.e., $\mpay(\play') \leq 0$. Hence, consider the strategy $\strat''_{2}$ of $\playerTwo$ that plays according to $\strat_{2}$ up to the first visit of $\state$ and according to $\strat'_{2}$ afterwards.  Clearly, there exists an outcome $\play'' \in \outcomesGame{\game}{\initState}{\strat_{1}}{\strat''_{2}}$ such that $\play'' = \prefix \cdot \play'$ and $\mpay(\play') \leq 0$ for some $\prefix \in \prefixes{\graph}$. Since the mean-payoff objective is prefix-independent, we have that $\mpay(\play'') \leq 0$. Thus, the strategy $\strat_{1}$ of $\playerOne$ does not satisfy Eq.~\eqref{eq:thresholdWC}, which contradicts the hypothesis and concludes the proof.
\end{proof}

\smallskip\noindent\textbf{Memoryless stochastic model.} Finally, we show in Lemma~\ref{lem:mp_memorylessStochModel} that we can study the equivalent $\BWC$ problem on the game obtained by product of the original game and the Moore machine of the stochastic model, using a memoryless stochastic model instead of the finite-memory one (lines~\ref{alg:mp_memory}-\ref{alg:mp_mdp} of $\mpAlgoName$). Having a memoryless stochastic model proves useful in the main steps of the algorithm (lines~\ref{alg:mp_main}-\ref{alg:mp_main_end}) as it guarantees that the game $\game$ and the MDP $\game[\stratStoch]$ possess the same underlying graph.

Let $\gameFull$ be a two-player game, $\graphFull$ its underlying graph, $\initState \in \states$ the initial state, $\strat_{2}^{f} \in \stratsFinite_{2}(\game)$ a finite-memory stochastic model of $\playerTwo$, $\mooreMachineFull{\strat_{2}^{f}}$ its Moore machine. We define the product game $\game' = \game \otimes \mooreMachine{\strat_{2}^{f}} = (\graph', \statesOne', \statesTwo')$, with $\graph' = (\states', \edges', \weight')$, as follows.
\begin{itemize}
\item $\states' = \states \times \mooreMem$, $\statesOne' = \statesOne \times \mooreMem$, $\statesTwo' = \statesTwo \times \mooreMem$;
\item $\forall\, \state_{1}, \state_{2} \in \states,\, \forall\, \mooreMemElem_{1}, \mooreMemElem_{2} \in \mooreMem$, $\left((\state_{1}, \mooreMemElem_{1}), (\state_{2}, \mooreMemElem_{2})\right) \in \edges'$ $\Leftrightarrow$ $(\state_{1}, \state_{2}) \in \edges$ $\wedge$ $\mooreUpd(\mooreMemElem_{1}, \state_{1}) = \mooreMemElem_{2}$;
\item $\forall\, e = \left((\state_{1}, \mooreMemElem_{1}), (\state_{2}, \mooreMemElem_{2})\right) \in \edges'$, $\weight'(e) = \weight((\state_{1}, \state_{2}))$;
\item $\initState' = (\initState, \mooreInitMem)$ is the new initial state.
\end{itemize}
Given the finite-memory stochastic model $\strat_{2}^{f} \in \stratsFinite_{2}(\game)$ of $\playerTwo$, we transcript it into a memoryless strategy $\strat_{2}^{m} \in \stratsMemoryless_{2}(\game')$ on the product game such that
$\forall\, \state_{1} \in \statesTwo,\, \forall\, \left((\state_{1}, \mooreMemElem_{1}), (\state_{2}, \mooreMemElem_{2})\right) \in \edges',\, \strat_{2}^{m}((\state_{1}, \mooreMemElem_{1}))((\state_{2}, \mooreMemElem_{2})) = \mooreNext(\mooreMemElem_{1}, \state_{1})(\state_{2})$.
Basically, we have integrated the finite memory of $\strat_{2}^{f}$ into the states of $\game'$ and defined the remaining corresponding memoryless strategy $\strat_{2}^{m}$ on $\game'$.
The following lemma holds.

\begin{lemma}
\label{lem:mp_memorylessStochModel}
Let $\thresholdExp \in \rat$ be the expected value threshold. The two following statements are equivalent.
\begin{enumerate}[(a)] 
\item $\playerOne$ has a strategy to satisfy the $\BWC$ problem on $\game$ against the finite-memory stochastic model $\strat_{2}^{f}$.
\item $\playerOne$ has a strategy to satisfy the $\BWC$ problem on $\game'$ against the memoryless stochastic model $\strat_{2}^{m}$.
\end{enumerate}
\end{lemma}

\begin{proof}
We first prove \textit{(a)} $\Rightarrow$ \textit{(b)}. Assume $\strat_{1} \in \stratsFinite_{1}(\game)$ is a satisfying strategy for the $\BWC$ problem on $\game$, i.e., it satisfies Eq.~\eqref{eq:thresholdWC} and~\eqref{eq:thresholdExp} against the stochastic model $\strat_{2}^{f}$. We build a corresponding strategy $\strat'_{1} \in \stratsFinite_{1}(\game')$ that is winning against $\strat_{2}^{m}$ in $\game'$ as follows:
$\forall\, \prefix' = (\state_{0}, \mooreMemElem_{0}) (\state_{1}, \mooreMemElem_{1}) \ldots{} (\state_{k}, \mooreMemElem_{k}) \in \prefixesArg{\game'}{1},\, (\state_{k+1}, \mooreMemElem_{k+1}) \in \states' \text{ such that } \left( (\state_{k}, \mooreMemElem_{k}), (\state_{k+1}, \mooreMemElem_{k+1})\right) \in \edges',\;
\strat'_{1}(\prefix')\left((\state_{k+1}, \mooreMemElem_{k+1})\right) = \strat_{1}(\proj{\states}(\prefix'))(\state_{k+1})$.
Note that strategy $\strat'_{1}$ is well-defined as $\strat_{1}$ is. Consider the worst-case requirement on $\game'$ (Eq.~\eqref{eq:thresholdWC}). Let $\play' \in \outcomesMDP{\game'}{\initState'}{\strat'_{1}}$ be any outcome consistent with the newly defined strategy $\strat'_{1}$. By definition of~$\strat'_{1}$, we have that $\play = \proj{\states}(\play') \in \outcomesMDP{\game}{\initState}{\strat_{1}}$ is an outcome consistent with $\strat_{1}$ in $\game$. Hence, by hypothesis, we have that $\mpay(\play) > 0$. By definition of $\weight'$, we have that $\mpay_{\weight'}(\play') = \mpay_{\weight}(\play)$, thus $\mpay(\play') > 0$ and the worst-case requirement is satisfied. Now consider the expected value, Eq.~\eqref{eq:thresholdExp}. By hypothesis, we have that $\expect_{\initState}^{\game[\strat_{1}, \strat_{2}^{f}]}(\mpay) > \thresholdExp$. We claim that $\expect_{\initState'}^{\game'[\strat'_{1}, \strat_{2}^{m}]}(\mpay) > \thresholdExp$. Indeed, observe that there is a bijection between outcomes in $\outcomesMC{\game[\strat_{1}, \strat_{2}^{f}]}{\initState}$ and $\outcomesMC{\game'[\strat'_{1}, \strat_{2}^{m}]}{\initState'}$ using the projection operator because the memory update function $\mooreUpd$ of the Moore machine $\mooreMachine{\strat_{2}^{f}}$ is deterministic. As the values of plays are preserved by the changes to $\weight'$ and the probability measures of cylinder sets are also preserved by definition of $\strat^{m}_{2}$, the claim is verified.

Second, we show that \textit{(b)} $\Rightarrow$ \textit{(a)}. Assume $\strat'_{1} \in \stratsFinite_{1}(\game')$ is a satisfying strategy for the $\BWC$ problem on $\game'$, against the stochastic model $\strat_{2}^{m}$. We build a corresponding strategy $\strat_{1} \in \stratsFinite_{1}(\game)$ that is winning against $\strat_{2}^{f}$ in $\game$. Thanks to the update function of $\mooreMachine{\strat_{2}^{f}}$ being deterministic, given a prefix $\prefix = \state_{0}\state_{1}\ldots{}\state_{k} \in \prefixes{\game}$, there is a unique corresponding prefix $\prefix' \in \prefixes{\game'}$ such that $\prefix = \proj{\states}(\prefix')$. Hence we define $\strat_{1}$ as follows: $\forall\, \prefix = \state_{0} \state_{1} \ldots{} \state_{k} \in \prefixesArg{\game}{1},\, \state_{k+1} \in \states \text{ such that } (\state_{k}, \state_{k+1}) \in \edges,\;
\strat_{1}(\prefix)(\state_{k+1}) = \strat'_{1}(\prefix')((\state_{k+1},\mooreMemElem_{k+1}))$,
with $\prefix' = (\state_{0}, \mooreMemElem_{0}) (\state_{1}, \mooreMemElem_{1}) \ldots{} (\state_{k}, \mooreMemElem_{k})$ the unique prefix in $\prefixesArg{\game'}{1}$ such that $\prefix = \proj{\states}(\prefix')$ and $\mooreMemElem_{k+1} = \mooreUpd(\mooreMemElem_{k}, \state_{k})$. Strategy $\strat_{1}$ is well-defined. The rest follows by similar arguments as for $(a) \Rightarrow (b)$.
\end{proof}

Given a satisfying strategy in the product game, the proof of Lemma~\ref{lem:mp_memorylessStochModel} describes how to obtain a corresponding satisfying strategy in the original game. Hence, strategies obtained through algorithm $\mpAlgoName$ are preserved alongside the answer to the $\BWC$ problem when going back to the original game.
We now study the main steps of the algorithm, assuming the simplifying assumptions given by the preprocessing.

\subsection{\textbf{Classification of end-components}}
\label{subsec:mpClassification}

In the following, we consider a game $\game$ where all states are winning for the worst-case requirement (Eq.~\eqref{eq:thresholdWC}), a memoryless stochastic model $\stratStoch \in \stratsMemoryless_{2}$, and $\markovProcess = \game[\stratStoch]$, the MDP obtained when this stochastic model is followed by~$\playerTwo$. As sketched in Sect.~\ref{subsec:mpApproach}, the crux is the analysis of the ECs of $\markovProcess$.

\smallskip\noindent\textbf{Long-run appearance of end-components.} Lemma~\ref{lem:ecsProbaOne} recalls a well-known property of MDPs: under any arbitrary strategy of $\playerOne$, the set of states visited infinitely often along a play almost-surely constitutes an EC.  Notice the abuse of notation as discussed in Sect.~\ref{sec:preliminaries}.

\begin{lemma}[\cite{courcoubetis_JACM1995,de1997formal}]
\label{lem:ecsProbaOne}
Let $\markovProcessFull$ be an MDP, $\graphFull$ its underlying graph, $\ecsSet \subseteq 2^{\states}$ the set of its end-components, $\initState \in \states$ the initial state, and $\strat_{1} \in \strats_{1}(\markovProcess)$ an arbitrary strategy of $\playerOne$. Then, we have that
$\proba^{\markovProcess[\strat_{1}]}_{\initState}\left( \{\play \in \outcomesMC{\markovProcess[\strat_{1}]}{\initState} \mid \infVisited{\play} \in \ecsSet\}\right) = 1$.
\end{lemma}
Hence the expected value $\expect_{\initState}^{\markovProcess[\strat_{1}]}(\mpay)$ depends exclusively on the values of end-components (because the mean-payoff of any play belongs to $\left[-\largestW, \largestW\right]$ and thus the value of plays not entering ECs, which set has probability measure zero, cannot be infinite).

\smallskip\noindent\textbf{Winning end-components.} First, we introduce some notations. We respectively denote $\gameNonZero$ and $\markovProcessNonZero$ the game and MDP where the underlying graph is limited to the subset of edges which are assigned non-zero probability by $\mpTrans = \stratStoch$, i.e., $\edgesNonZero \subseteq \edges$. By definition, ECs are computed with regard to~$\edgesNonZero$, hence the ECs of $\markovProcessNonZero$ are exactly equal to the ECs of~$\markovProcess$. Moreover, we have that
$\strats_{2}(\gameNonZero) = \big\lbrace \strat_{2} \in \strats_{2}(\game) \mid \forall\, \prefix \cdot \state \in \prefixesArg{\game}{2},\, \supp(\strat_{2}(\prefix\cdot\state)) \subseteq \supp(\mpTrans(\state)) \big\rbrace$,
whereas available choices are unchanged for $\playerOne$ in $\gameNonZero$ as edges in $\edges \setminus \edgesNonZero$ all start in states of $\playerTwo$. 

Using these notations, we now define \textit{winning} ECs as the ECs $\ec \in \ecsSet$ where $\playerOne$ can ensure satisfaction of the worst-case requirement in the corresponding subgame $\gameNonZero \reduc \ec$. Note that we consider $\gameNonZero$ as we only need to consider strategies of $\playerTwo$ that share the support of the stochastic model. That is because our goal is to ensure that $\playerOne$ can benefit from (the maximal expectation achievable in) these ECs, which is only safe with regard to the worst-case requirement if~$\playerOne$ can force that all outcomes that may occur when facing the stochastic model yield a strictly positive mean-payoff. Note that reacting to an arbitrary strategy of $\playerTwo$, i.e., a strategy in $\strats_{2}(\game)$, as required in Eq.~\eqref{eq:thresholdWC}, will be considered in the following sections: for now we only care about $\stratStoch$ and satisfaction of the expected value requirement as specified in Eq.~\eqref{eq:thresholdExp}.

\begin{definition}
\label{def:classificationECs}
Let $\gameFull$ be a two-player game, $\graphFull$ its underlying graph, $\stratStoch \in \stratsMemoryless_{2}$ a memoryless stochastic model of $\playerTwo$, $\markovProcess = \game[\stratStoch] = (\graph, \statesOne, \statesProb = \statesTwo, \mpTrans = \stratStoch)$ the resulting MDP and $\gameNonZero$ the game reduced to non-zero probability edges. Let $\ec \in \ecsSet$ be an end-component of $\markovProcess$. Then, we have that
\begin{itemize}
\item $\ec \in \winningECs\;$, the {\em winning ECs}, if the following holds:
$\exists\, \strat_{1} \in \strats_{1}(\gameNonZero \reduc \ec),\, \forall\, \strat_{2} \in \strats_{2}(\gameNonZero \reduc \ec),\, \forall\, \state \in \ec,\, \forall\, \play \in \outcomesGame{(\gameNonZero \reduc \ec)}{\state}{\strat_{1}}{\strat_{2}},\, \mpay(\play) > 0\;;$
\item $\ec \in \losingECs$, the {\em losing ECs}, otherwise. By determinacy of mean-payoff games,
$\exists\, \strat_{2} \in \strats_{2}(\gameNonZero \reduc \ec),\, \forall\, \strat_{1} \in \strats_{1}(\gameNonZero \reduc \ec),\, \exists\, \state \in \ec,\, \exists\, \play \in \outcomesGame{(\gameNonZero \reduc \ec)}{\state}{\strat_{1}}{\strat_{2}},\, \mpay(\play) \leq 0$.
\end{itemize}
\end{definition}
Note that an EC is winning if $\playerOne$ has a worst-case winning strategy from \textit{all} states. This point is important as it may well be the case that winning strategies exist in a strict subset of states of the EC. This does not contradict the definition of ECs as strongly connected subgraphs, as the latter only guarantees that every state can be reached \textit{with probability one}, and not necessarily surely. Hence one cannot call upon the prefix-independence of the mean-payoff to extend the existence of a winning strategy to all states. Such a situation can be observed on the game of Fig.~\ref{fig:mp_winningECsComputationExample}, where the EC~$\ec_{2}$ is losing (because from $\state_{1}$, the outcome $(\state_{1}\state_{3}\state_{4})^{\omega}$ can be forced by $\playerTwo$, yielding mean-payoff $-1/3 \leq 0$), while its sub-EC $\ec_{3}$ is winning. From $\state_{1}$, $\playerOne$ can ensure to reach $\ec_{3}$ almost-surely, but not surely, which is critical in this case.

\begin{figure}[htb]
  \centering   
  \scalebox{0.6}{\begin{tikzpicture}[->,>=latex,shorten >=1pt,auto,node
    distance=2.5cm,bend angle=45,font=\Large]
    \tikzstyle{p1}=[draw,circle,text centered,minimum size=10mm]
    \tikzstyle{p2}=[draw,rectangle,text centered,minimum size=10mm]
    \tikzstyle{empty}=[]
    \node[p1] (1) at (0,0) {$\state_{1}$};
    \node[p1] (2) at (4,0) {$\state_{2}$};
    \node[p2] (3) at (-2,-2) {$\state_{3}$};
    \node[p1] (4) at (2,-2) {$\state_{4}$};
    \node[p1] (5) at (-4,0) {$\state_{5}$};
    \node[empty] (ec1) at (-5.2, -1.2) {$\ec_{3}$};
    \node[empty] (ec2) at (-7.4, 1) {$\ec_{2}$};
    \node[empty] (ec2) at (6.6, 1) {$\ec_{1}$};
    \node[empty] (proba1) at (-2.7, -2.4) {$\frac{1}{2}$};
    \node[empty] (proba2) at (-1.3, -2.4) {$\frac{1}{2}$};
    \coordinate[shift={(0mm,5mm)}] (init) at (1.north);
    \path
    (1) edge node[above] {$0$} (2)
    (5) edge node[above] {$0$} (1)
    (1) edge node[left,xshift=-1mm] {$0$} (3)
    (4) edge node[right] {$-1$} (1)
    (3) edge node[below] {$0$} (4)
    (init) edge (1)
    (5) edge [loop left, out=150, in=210,looseness=3, distance=16mm] node [left] {$10$} (5)
    (2) edge [loop right, out=30, in=330,looseness=3, distance=16mm] node [right] {$1$} (2)
    ;
	\draw[->,>=latex] (3) to[out=180,in=270] node[left,xshift=-1mm] {$0$} (5);
	\draw[dashed,-] (-3.2,0.8) -- (-6.4,0.8) -- (-6.4,-0.8) -- (-3.2,-0.8) -- (-3.2,0.8);
	\draw[dashed,-] (0,1.4) -- (-7,1.4) -- (-7,-2.8) -- (3.8,-2.8) -- (0,1.4);
	\draw[dashed,-] (1.7,1.4) -- (6.2,1.4) -- (6.2,-2.8) -- (5.5,-2.8) -- (1.7,1.4);
      \end{tikzpicture}}
      \caption{End component $\ec_{2}$ is losing. The set of maximal winning ECs is $\maxWinningECs = \winningECs = \{\ec_{1}, \ec_{3}\}$.}
\label{fig:mp_winningECsComputationExample}
  \end{figure}

\smallskip\noindent\textbf{Maximality.} As discussed in Sect.~\ref{subsec:mpApproach}, we can restrict our analysis to \textit{maximal} winning ECs in the following. This is a consequence of Lemma~\ref{lem:mp_ecsMaximal}.

\begin{lemma}
\label{lem:mp_ecsMaximal}
Let $\ec_{1}, \ec_{2} \in \winningECs$ be two winning ECs in the MDP $\markovProcess$ such that $\ec_{1} \subsetneq \ec_{2}$. Let $\optimalExp_{1}, \optimalExp_{2}$ denote the respective maximal expected values achievable by $\playerOne$ in $\ec_{1}$ and $\ec_{2}$. Then, we have that $\optimalExp_{1} \leq \optimalExp_{2}$.
\end{lemma}

\begin{proof}
First notice that the maximal expectation achievable in an EC does not depend on the starting state inside the EC. Hence, assume any state $\state_{2} \in \ec_{2}$. Since $\ec_{1} \subsetneq \ec_{2}$ and $\ec_{2}$ is an end-component, $\playerOne$ can reach a state in $\ec_{1}$ with probability one from $\state_{2}$.
Hence, the contribution of plays that do not reach $\ec_{1}$ in the expected value is null. Finally, by prefix-independence of the mean-payoff, we can forget about the finite prefixes outside $\ec_{1}$ and we deduce that $\optimalExp_{2} \geq \optimalExp_{1}$.
\end{proof}

We formally define the set of \textit{maximal winning ECs} as $\maxWinningECs = \{\ec \in \winningECs \mid \forall\, \ec' \in \winningECs,\, \ec \subseteq \ec' \Rightarrow \ec = \ec'\}$.
While the total number of winning ECs $\vert\winningECs\vert \leq \vert \ecsSet\vert \leq 2^{\vert\states\vert}$ can be exponential in the number of states of the game, the number of maximal winning ECs $\vert\maxWinningECs\vert \leq \vert\states\vert$ is bounded by this number of states, as all ECs of $\maxWinningECs$ are disjoint (because the union of two winning ECs is itself a winning EC). Hence, restriction to the maximal winning ECs is a cornerstone in the overall $\NPinter$ complexity that we claim for the BWC problem.

\smallskip\noindent\textbf{Illustration.} Consider the running example in Fig.~\ref{fig:mpRunningExample}. Note that states $\state_{1}$, $\state_{2}$ and $\state_{5}$ do not belong to any EC: given any strategy of $\playerOne$ in $\markovProcess$, with probability one, any consistent outcome will only visit those states a finite number of times (Lemma~\ref{lem:ecsProbaOne}). The set of \textit{maximal winning ECs} is $\maxWinningECs = \{\ec_{2}, \ec_{3}\}$. Obviously, those ECs are disjoint. The set of winning ECs is larger, $\winningECs = \maxWinningECs \cup \{\{\state_{9}, \state_{10}\}, \{\state_{6}, \state_{7}\}\}$.

End-component $\ec_{1}$ is \textit{losing}. Indeed, in the subgame $\gameNonZero \reduc \ec_{1}$, the strategy consisting in always picking the $-1$ edge guarantees an outcome which mean-payoff is negative. Note that this edge is present in $\edgesNonZero$ as it is assigned probability $1/2$ by the stochastic model. Here, we witness why it is important to base our definition of winning ECs on the game $\gameNonZero$ rather than $\game$. Indeed, in $\game \reduc \ec_{2}$, $\playerTwo$ can also guarantee a negative mean-payoff by always choosing edges with weight $-1$. However, to achieve this, $\playerTwo$ has to pick edges that are \textit{not} in $\edgesNonZero$: this will never happen against the stochastic model and as such, this can be watched by $\playerOne$ to see if $\playerTwo$ uses an arbitrary antagonistic strategy, and dealt with. If $\playerTwo$ conforms to $\edgesNonZero$, i.e., if he plays in $\gameNonZero$, he has to pick the edge of weight $1$ in $\state_{7}$ and $\playerOne$ has a worst-case winning strategy consisting in always choosing to go in $\state_{7}$. This EC is thus classified as \textit{winning}. Note that for $\ec_{3}$, in both subgames $\game \reduc \ec_{3}$ and $\gameNonZero \reduc \ec_{3}$, $\playerOne$ can guarantee a strictly positive mean-payoff by playing $(\state_{9}\,\state_{10})^\omega$: even \textit{arbitrary} strategies of $\playerTwo$ cannot endanger $\playerOne$ in this case.

Lastly, consider the game depicted in Fig.~\ref{fig:mp_winningECsComputationExample}. While $\ec_{2}$ is a strict superset of $\ec_{3}$, the former is losing whereas the latter is winning, as explained above. Hence, the set $\maxWinningECs$ is equal to $\{\ec_{1}, \ec_{3}\}$.

\smallskip\noindent\textbf{Computation of the maximal winning end-components.} Obviously, from a complexity standpoint, to benefit from the polynomial size of $\maxWinningECs$, in contrast to the potentially exponential size of $\winningECs$, we need to compute $\maxWinningECs$ without first computing all winning ECs of $\winningECs$. We present an algorithm to do so, called $\mwecAlgoName$, in Alg.~\ref{alg:ec}. Lemma~\ref{lem:mpECsPartition} establishes that $\mwecAlgoName$ is correct and complete, and uses a polynomial number of calls to an $\NPinter$ oracle to solve mean-payoff games (i.e., decide the answer of the worst-case threshold problem), other than that implementing polynomial operations. It operates under the assumption that all edges are of non-zero probability, i.e., $\edgesNonZero = \edges$. This is w.l.o.g.~as it suffices to remove those edges as a preprocessing step (they have no impact in the definitions of ECs and winning ECs).

\begin{lemma}
\label{lem:mpECsPartition}
Let $\markovProcessFull$ be an MDP, with $\graphFull$ its underlying graph such that $\edgesNonZero = \edges$. Then algorithm $\mwecAlgoName$ computes its set of maximal winning ECs $\maxWinningECs = \mwecAlgoName(\markovProcess)$ in polynomially-many calls to an $\NPinter$ oracle for mean-payoff games.
\end{lemma}

\captionsetup[algorithm]{font=small}

\begin{algorithm}[t]
\caption{$\mwecAlgoName(\markovProcess)$}
\label{alg:ec}
\begin{algorithmic}[1]
\small
\REQUIRE $\markovProcessFull$, with $\graphFull$ such that $\edgesNonZero = \edges$
\ENSURE $\mwecAlgoSet = \maxWinningECs$, the set of maximal winning ECs of $\markovProcess$ 

\vspace{2mm}
\IF{$S = \emptyset$} \label{alg:ec_baseCase}
	\RETURN $\emptyset$
\ELSE \label{alg:ec_indCase}
\STATE Compute $\mwecAlgoSet := \{ \ec_{1}, \ldots{}, \ec_{n}\}$ the maximal EC decomposition of $\markovProcess$\label{alg:ec_maxEC}
   \FORALL{$i = 1, \ldots{}, n$}\label{alg:ec_loop}
    \STATE Compute $L_i\subseteq U_i$ the set of states from which $\playerTwo$ has a strategy to enforce $\mpay \leq 0$ in $\markovProcess \reduc \ec_{i}$, i.e.,\label{alg:ec_Li}
\begin{equation*}
L_i := \left\lbrace  \state\in \ec_{i} \mid \forall\, \strat_{1} \in \strats_{1}(\markovProcess \reduc \ec_{i}),\, \exists\, \play \in \outcomesMDP{\markovProcess \reduc \ec_{i}}{\state}{\strat_{1}},\, \mpay(\play)\leq 0\right\rbrace 
\end{equation*}
\vspace{-5mm}
\IF{$L_i\neq \emptyset$}
       \STATE $\mwecAlgoSet := (\mwecAlgoSet \setminus \{ \ec_{i}\}) \cup \mwecAlgoName(\markovProcess\reduc (\ec_i \setminus L_i))$\label{alg:ec_reccall}
\ENDIF
    \ENDFOR
\ENDIF
\RETURN $\mwecAlgoSet$
\end{algorithmic}
\end{algorithm}

Algorithm $\mwecAlgoName$ can be sketched as follows. Given a non-empty MDP ($S \neq \emptyset$), it first computes its decomposition into maximal end-components\footnote{Given an MDP $\markovProcess$ with a set of ECs $\ecsSet$, an EC $\ec$ is said to be \textit{maximal} in $\markovProcess$ if for all $\ec' \in \ecsSet$, $\ec \subseteq \ec' \Rightarrow \ec = \ec'$. This is not to be confused with the definition of maximal winning ECs, given in this section. In particular, maximal ECs need not be winning in general, whereas maximal winning ECs need not be maximal ECs in the sense we just defined (they only need to be maximal with regard to other \textit{winning} ECs).} (without distinction between winning and losing ECs). This can be obtained in polynomial time~\cite{DBLP:journals/jacm/ChatterjeeH14}. Afterwards, it checks for each of these ECs if it is winning or not, in the sense of Def.~\ref{def:classificationECs}. If the EC is winning, it is now part of the set of claimed maximal winning ECs, denoted $\mwecAlgoSet$ in the algorithm, and the algorithm will not recurse on this set of states. If the EC is losing, then it may still be the case that a sub-EC is winning, as discussed in Sect.~\ref{subsec:mpClassification}. Hence, the algorithm eliminates all worst-case losing states and executes recursively on the induced sub-MDP. It stops recursing on sub-MDPs whenever one is declared winning or empty. Since it suppresses at least one state in each call, the algorithm is ensured to stop. Moreover, the number of calls is polynomial in the size of the MDP. Deciding if an EC is winning or losing requires calling an $\NPinter$ oracle solving the worst-case threshold problem~\cite{ZP96,jurdzinski98,gawlitza2009}.

The remainder of this section is dedicated to the proof of the correctness and completeness of the algorithm, as well as an illustration of its operation. We first state several remarks on its functioning.

\begin{remark}\label{ec:rm1}
    In line~\ref{alg:ec_Li}, we have that $P\reduc U_i$ is a well-defined MDP, since $U_i$ is an EC. Similarly, in line~\ref{alg:ec_reccall}, $P\reduc (U_i\setminus L_i)$ is also a well-defined MDP. Indeed, from all states $s\in (U_i\setminus L_i) \cap \states_{1}$, 
    there exists an edge from $s$ that goes to a state of~$U_i\setminus L_i$, otherwise $s$ would be a losing state (and so would be in $L_i$). Moreover, 
    for all states $\state \in (U_i\setminus L_i) \cap \statesProb$, there is no edge from $s$ that goes in $L_i$ (otherwise $s$ would be in $L_i$) nor in $\states \setminus \ec_{i}$ (otherwise $\ec$ would not be an EC), and therefore the probability distribution $\mpTrans(s)$ is still well-defined on $P\reduc (U_i\setminus L_i)$. Additionally, this sub-MDP still verifies that there is no edge with probability zero.
\end{remark}

\begin{remark}\label{ec:rm2}
    Let $U$ be an EC of $P$, $L$ its set of losing states (as computed by line~\ref{alg:ec_Li}), and $V\subseteq U\setminus L$. 
    Then $V$ is a winning EC of $P\reduc (U\setminus L)$ if and only if $V$ is a winning EC of $P$. This follows from the same reasoning as for Rem.~\ref{ec:rm1}.
\end{remark}

\begin{remark}\label{ec:rm3}
    Let $U$ be a winning EC of $P$, strictly included in a losing EC $V$ of $P$. Let $s\in V$ be a worst-case losing state (i.e., a state from which $\playerOne$ cannot guarantee a strictly positive mean-payoff in $\markovProcess \reduc V$, as defined at line~\ref{alg:ec_Li} of the algorithm). Then we claim that $s\not\in U$. Indeed, suppose the contrary. From $s$, $\playerTwo$ can enforce $\mpay \leq 0$ against any strategy $\lambda_1\in\Lambda_1(P\reduc V)$, and a fortiori could do so against any strategy
    $\lambda_1\in \Lambda_1(P\reduc U)$ (notice that only $\playerOne$ can decide to leave $U$ as $U$ is an EC). Therefore, $U$ would not be winning.
\end{remark}

\begin{proof}
We first show that the algorithm is \textit{sound}, i.e., $\mwecAlgoSet \subseteq \maxWinningECs$. It is done by induction on the size
    of~$P$. If $P$ is empty ($S = \emptyset$), the claim is clear. Otherwise let $U\in \mwecAlgoSet$. There are two cases.
(I) \textit{$U$ is equal to some $U_i$ computed at line~\ref{alg:ec_maxEC} and has never been removed from $\mwecAlgoSet$}. It means that $L_i$ is empty, and by definition of $L_i$, that $U_i$ is winning. Also, $\ec_{i}$ is trivially a \textit{maximal} winning EC as it belongs to the maximal EC decomposition of $\markovProcess$.
(II) \textit{$U$ has been added at line~\ref{alg:ec_reccall} as the result of some recursive call $\mwecAlgoName(P\reduc (U_i\setminus L_i))$ for some
          maximal EC $U_i$ of~$P$}. Since $L_i\neq \emptyset$, the set $U_i\setminus L_i$ is strictly smaller than~$S$, and
          by induction hypothesis, $U$ is a winning EC of $P\reduc (U_i\setminus L_i)$. By Rem.~\ref{ec:rm2}, it is also a winning EC of~$P$. It remains to show that $U$ is a \emph{maximal} winning EC of~$P$. Suppose that it is not the case. Then there exists a
          strict superset $U'$ of $U$ which is a winning EC of~$P$. Clearly, $U'\subseteq U_i$ since $U_i$ is a maximal
          EC of $P$, and maximal ECs are pairwise disjoint. Morevoer, $U'$ is a subset of $U_i\setminus L_i$ by Rem.~\ref{ec:rm3}. 
          By Rem.~\ref{ec:rm2}, it is therefore a winning EC of $P\reduc (U_i\setminus L_i)$, which contradicts the maximality
          of $U$ in~$P\reduc (U_i\setminus L_i)$.

We now establish that the algorithm is \textit{complete}, i.e., $\maxWinningECs \subseteq \mwecAlgoSet$. Again, it is proved by induction on the size of~$P$. If~$P$ is empty, then
    the claim is obviously true. Now, suppose that $P$ is non-empty, and let $U\in \maxWinningECs$. There are two cases. (I) \textit{$U$ is a maximal EC of $P$}. In that case, it will be computed at line~\ref{alg:ec_maxEC} and never removed from 
          $\mwecAlgoSet$ (because the set of losing states will be empty as $U$ is winning).
(II) \textit{$U$ is not a maximal EC in $\markovProcess$}. Therefore there exists some maximal EC $U_i$ of $P$, which is losing and strictly contains~$U$. 
          Let $L_i$ be the non-empty set of worst-case losing states of $U_i$ (as computed by line~\ref{alg:ec_Li}). We have to show that
          $U$ is a maximal winning EC of $P\reduc (U_i\setminus L_i)$,  in which case we could conclude by induction
          hypothesis, i.e., $U$ would be returned by the recursive call $\mwecAlgoName(P\reduc (U_i\setminus L_i))$. By Rem.~\ref{ec:rm3}, $U$ and $L_i$ are disjoint, and therefore $U\subseteq (U_i\setminus L_i)$. By Rem.~\ref{ec:rm1} and Rem.~\ref{ec:rm2}, $U$ is a winning
          EC of $P\reduc (U_i\setminus L_i)$, since it is a winning EC of $P$. It remains
          to show that $U$ is a \emph{maximal} winning EC of $P\reduc (U_i\setminus L_i)$. Suppose that there exists a strict
          superset $U'$ of $U$ such that $U'$ is a winning EC of $P\reduc (U_i\setminus L_i)$. By Rem.~\ref{ec:rm2}, $U'$ would also be
          a winning EC of~$P$, which contradicts that $\ec \in \maxWinningECs$ by definition of $\maxWinningECs$ as the set of maximal winning ECs. It implies that $U$ is a maximal
          winning EC of $P\reduc (U_i\setminus L_i)$ and thus that $\ec \in \mwecAlgoSet$.

Finally, consider the complexity. Assume that we have an $\NPinter$ oracle solving the worst-case threshold problem~\cite{ZP96,jurdzinski98,gawlitza2009} and called in line~\ref{alg:ec_Li}. The number of recursive calls to $\mwecAlgoName$ is linear in~$|S|$,
        the number of states of $P$, because in the loop of line~\ref{alg:ec_loop}, the sets
        $U_i$ are pairwise disjoint, and because each call is executed after eliminating at least one state. Moreover, the maximal EC decomposition of $P$ can be computed in~$O(|S|^2)$~\cite{DBLP:journals/jacm/ChatterjeeH14}. Therefore, algorithm $\mwecAlgoName$ makes a polynomial number of calls to an $\NPinter$ oracle, which proves the claim.
\end{proof}

Consider the execution of algorithm $\mwecAlgoName$ on the MDP described in Fig.~\ref{fig:mp_winningECsComputationExample}. In its first call, it computes the maximal EC decomposition $\mwecAlgoSet = \{\ec_{1}, \ec_{2}\}$. Now, for $\ec_{1}$, we have that $L_{1}$ is empty and thus $\ec_{1}$ remains in $\mwecAlgoSet$. On the contrary, for $\ec_{2}$, we have that $L_{2} = \{\state_{1}, \state_{3}, \state_{4}\}$. Hence the algorithm suppresses $\ec_{2}$ from $\mwecAlgoSet$ and recurses on the sub-MDP $\markovProcess \reduc (\ec_{2} \setminus L_{2}) = \markovProcess \reduc \{\state_{5}\}$. There, the maximal EC decomposition gives the unique EC $\ec_{3}$ which is winning since $L_{3} = \emptyset$, and thus remains in $\mwecAlgoSet$. The algorithm ends with $\mwecAlgoSet = \{\ec_{1}, \ec_{3}\}$. Clearly we have that $\mwecAlgoSet = \maxWinningECs$ as proved before.

\subsection{\textbf{Winning end-components are almost-surely reached in the long-run}}
\label{subsec:mpInsideLosing}

Recall that Lemma~\ref{lem:ecsProbaOne} states that under any arbitrary strategy $\strat_{1} \in \strats_{1}$, the set of infinitely visited states of the outcome of the MDP $\markovProcess[\strat_{1}] = \game[\strat_{1},\stratStoch]$ is almost-surely equal to an EC. In this section, we refine this result and show that under any {\em finite-memory} strategy $\strat_{1}^{f} \in \stratsFinite_{1}$ {\em satisfying the $\BWC$ problem}, the set of infinitely visited states is almost-surely equal to a \textit{winning} EC. In other words, the long-run probability of \textit{negligible states}, defined as $\negligibleStates = \states \setminus \bigcup_{\ec \in \maxWinningECs} \ec = \states \setminus \bigcup_{\ec \in \winningECs} \ec$, is zero.

\begin{lemma}
\label{lem:winningECsProbaOne}
Let $\gameFull$ be a two-player game, $\graphFull$ its underlying graph, $\stratStoch \in \stratsMemoryless_{2}$ a memoryless stochastic model of $\playerTwo$, $\markovProcess = \game[\stratStoch]$ the resulting MDP and $\initState \in \states$ the initial state. Let $\strat^{f}_{1} \in \stratsFinite_{1}$ be a finite-memory strategy of $\playerOne$ that satisfies the $\BWC$ problem for thresholds $(0,\, \thresholdExp) \in \rat^{2}$. Then, we have that $\proba^{\markovProcess[\strat_{1}^{f}]}_{\initState}\left( \left\lbrace \play \in \outcomesMC{\markovProcess[\strat_{1}^{f}]}{\initState} \mid \infVisited{\play} \in \winningECs \right\rbrace \right) = 1$.
\end{lemma}

\begin{proof}
Let $\strat^{f}_{1} \in \stratsFinite_{1}$ be a finite-memory $\BWC$ satisfying strategy. By Lemma~\ref{lem:ecsProbaOne}, we have the claim is verified for $\ec \in \ecsSet = \winningECs \cup \losingECs$. It remains to show that the probability of having a losing EC, i.e., an EC in $\losingECs$, is zero.
By contradiction, assume there exists some $\ec_{\textsc{l}} \in \losingECs$ such that 
$\proba^{\markovProcess[\strat_{1}^{f}]}_{\initState}\left(\left\lbrace  \play \in \outcomesMC{\markovProcess[\strat_{1}^{f}]}{\initState} \mid \infVisited{\play} = \ec_{\textsc{l}}\right\rbrace \right) > 0$.
Since $\strat^{f}_{1}$ is finite-memory, we have that $\markovChain = \markovProcess[\strat^{f}_{1}]$ is a \textit{finite} MC. Thus, we consider the bottom strongly-connected components (BSCCs) of $\markovChain$ and the contradiction hypothesis implies that some outcomes of $\outcomesMC{\markovChain}{\initState}$ will be trapped in a BSCC corresponding to $\ec_{\textsc{l}}$ (i.e., this BSCC is reachable with non-zero probability in $\markovChain$), and visit all its states infinitely often. Since $\ec_{\textsc{l}}$ is losing, this BSCC induces plays where the mean-payoff is not strictly positive. Indeed, strategy $\stratStoch$ of $\playerTwo$ suffices to produce consistent outcomes that are worst-case losing thanks to Def.~\ref{def:classificationECs} (as only the support matters for the worst-case requirement, not the exact probabilities). By prefix-independence of the mean-payoff value function, we obtain the existence of plays of $\markovChain$, starting in $\initState$, and inducing a mean-payoff that does not satisfy the worst-case threshold. This shows that $\strat^{f}_{1}$ is not winning for the $\BWC$ problem and by contradiction, concludes our proof.
\end{proof}

The consequence of this statement is that \textit{edges involving negligible states do not contribute to the overall expectation} of finite-memory strategies satisfying the $\BWC$ problem. Based on this, we propose in Sect.~\ref{subsec:mpGlobal} a modification of the MDP $\markovProcess = \game[\stratStoch]$ that helps us synthesize satisfying strategies when they exist.

\begin{remark}
The proof of Lemma~\ref{lem:winningECsProbaOne} relies on the finite memory of the strategy. Similar reasoning cannot be applied if the strategy of $\playerOne$ uses infinite memory. Indeed, the Markov chain $\markovChain = \markovProcess[\strat^{f}_{1}]$ becomes infinite and we cannot base our analysis on BSCCs anymore (as they need not exist in general, outcomes cannot be trapped in a BSCC). As a matter of fact, we prove in Sect.~\ref{subsec:mpInfiniteMemory} that infinite-memory strategies may benefit from negligible states forming losing ECs in some cases, by staying in them forever with a non-zero probability, and thus it is not possible to neglect them in the overall expectation.
\end{remark}

\smallskip\noindent\textbf{Illustration.} Consider $\ec_{1}$ on Fig.~\ref{fig:mpRunningExample}. By Def.~\ref{def:classificationECs}, this EC is losing as always taking the edge of weight $-1$ is a winning strategy for $\playerTwo$ in $\gameNonZero \reduc \ec_{1}$. The optimal expectation achievable in $\markovProcess \reduc \ec_{1}$ by $\playerOne$ is $4$: this is higher than what is achievable in both $\ec_{2}$ and $\ec_{3}$. Note that there exists no winning EC included in $\ec_{1}$. By Lemma~\ref{lem:ecsProbaOne}, we know that any strategy of $\playerOne$ will see its expectation bounded by the maximum between the optimal expectations of the ECs $\ec_{1}$, $\ec_{2}$ and $\ec_{3}$. Lemma~\ref{lem:winningECsProbaOne} further refines this bound by restricting it to the maximum between the expectations of $\ec_{2}$ and $\ec_{3}$. Indeed, it states that $\playerOne$ cannot benefit from the expected value of $\ec_{1}$ while using finite memory, as being trapped in~$\ec_{1}$ with non-zero probability induces the existence of outcomes losing for the worst-case (here, outcomes that always take the $-1$ edge). Since $\ec_{1}$ neither helps for the worst-case nor for the expectation, there is no point in playing inside it and $\playerOne$ may as well cross it directly and try to maximize its expectation using the winning ECs, $\ec_{2}$ and $\ec_{3}$. The set of negligible states in $\markovProcess$ is $\negligibleStates = \states \setminus (\ec_{2} \cup \ec_{3}) = \{\state_{1}, \state_{2}, \state_{3}, \state_{4}, \state_{5}\}$.

In the game depicted in Fig.~\ref{fig:mp_winningECsComputationExample}, we already observed that $\ecsSet = \{\ec_{1}, \ec_{2}, \ec_{3}\}$, $\winningECs = \maxWinningECs = \{\ec_{1}, \ec_{3}\}$ and $\losingECs = \{\ec_{2}\}$. Consider the negligible state $\state_{1} \in \negligibleStates = \ec_{2} \setminus \ec_{3}$. As a consequence of Lemma~\ref{lem:winningECsProbaOne}, we have that a finite-memory strategy of $\playerOne$ may only take the edge $(\state_{1}, \state_{3})$ finitely often in order to ensure the worst-case requirement. Indeed, the losing outcome $(\state_{1}\state_{3}\state_{4})^{\omega}$ would exist (while of probability zero) if $\playerOne$ were to play this edge infinitely often. Therefore, it is clear that $\playerOne$ can only ensure that $\ec_{3}$ is reached with a probability arbitrarily close to one, and not equal to one, because at some point, he has to switch to edge $(\state_{1}, \state_{2})$ (after a bounded time since $\playerOne$ uses a finite-memory strategy).

\subsection{\textbf{Winning end-component with non-zero probabilities: combined strategy}}
\label{subsec:mpInsideStrongly}

In this section, we take a closer look at what happens inside a winning EC \textit{where the stochastic model assigns non-zero probabilities} to all possible edges. We will show how to deal with edges of probability zero in Sect.~\ref{subsec:mpInsideWeakly}. For the sake of readability, we make Assumption~\ref{assump:uniqueWEC}. Obviously, similar reasoning can be applied to all the such winning ECs in a larger game.

\begin{assumption}
\label{assump:uniqueWEC}
Let $\gameFull$ be a two-player game, $\graphFull$ its underlying graph, $\stratStoch \in \stratsMemoryless_{2}$ a memoryless stochastic model of $\playerTwo$, and $\markovProcess = \game[\stratStoch] = (\graph, \statesOne, \statesProb = \statesTwo, \mpTrans = \stratStoch)$ the resulting MDP. We assume that $\game$ is reduced to a unique maximal winning EC, i.e., $\maxWinningECs = \{\states\}$, and that $\gameNonZero = \game$, i.e., the set of edges of probability zero, $\edges \setminus \edgesNonZero$, is empty.
\end{assumption}

Our claim is that inside such a winning EC, $\playerOne$ has a \textit{finite-memory} strategy that simultaneously (a) ensures the worst-case requirement, and (b) yields an expected value which is $\varepsilon$-close to the maximal expectation of the EC. Consequently, we establish Theorem~\ref{thm:insideWinning} and Corollary~\ref{cor:insideWinning}.
\begin{theorem}
\label{thm:insideWinning}
Let $\gameFull$ be a two-player game reduced to a unique winning EC, $\graphFull$ its underlying graph, $\stratStoch \in \stratsMemoryless_{2}$ a memoryless stochastic model of $\playerTwo$ such that $\edgesNonZero = \edges$, $\markovProcess = \game[\stratStoch] = (\graph, \statesOne, \statesProb = \statesTwo, \mpTrans = \stratStoch)$ the resulting MDP, $\initState \in \states$ an initial state and $\optimalExp \in \rat$ the maximal expected value achievable by $\playerOne$ in $\markovProcess$. Then, for all~$\varepsilon > 0$, there exists a finite-memory strategy of $\playerOne$ that satisfies the $\BWC$ problem for the thresholds pair $(0,\, \optimalExp - \varepsilon)$.
\end{theorem}
The remainder of this section is dedicated to the proof of Thm.~\ref{thm:insideWinning}. It is a remarkably positive result as it essentially states that $\playerOne$ can \textit{guarantee both} the worst-case and the expected value thresholds \textit{without sacrifying any performance} (in terms of play values) except for some arbitrarily small $\varepsilon$. The key idea is to build a finite-memory strategy based on careful alternation between two memoryless strategies: one which is optimal for the worst-case, and one which is optimal for the expected value. The proof requires deep understanding of the limiting properties of Markov chains, such as the rate of convergence towards a stationary distribution. Nevertheless, we provide an intuitive sketch of the combined strategy and illustrate it on the running example in the following.

\begin{corollary}
\label{cor:insideWinning}
In a game $\game$ reduced to a winning EC, $\playerOne$ has a strategy for the $\BWC$ problem for thresholds $(0, \thresholdExp)$ against a stochastic model $\stratStoch \in \stratsMemoryless_{2}$ such that $\edgesNonZero = \edges$ if and only if the optimal expected value in $\game[\stratStoch]$ is strictly greater than $\thresholdExp$.
\end{corollary}

\begin{proof}
Consider the left-to-right implication. Assume $\strat_{1}$ is the $\BWC$ strategy. By definition, the optimal expected value $\optimalExp$ is at least equal to $\expect_{\initState}^{\markovProcess[\strat_{1}]}(\mpay)$, which is strictly greater than $\thresholdExp$ by hypothesis. Now consider the converse implication. Let $\optimalExp$ be the optimal expected value. By hypothesis, $\optimalExp > \thresholdExp$. By Thm.~\ref{thm:insideWinning}, for any $\varepsilon > 0$, there exists a strategy $\strat_{1} \in \stratsFinite_{1}$ that satisfies the $\BWC$ problem for thresholds $(0, \optimalExp - \varepsilon)$. In particular, it is possible to choose $\varepsilon \leq \optimalExp - \thresholdExp$ to obtain the claim and conclude the proof.
\end{proof}

\smallskip\noindent\textbf{Individual requirements - memoryless strategies.} First, consider the two requirements - worst-case and expected value - independently. By definition of winning ECs and the hypothesis that $\game = \gameNonZero$, $\playerOne$ has a strategy to guarantee a strictly positive mean-payoff against any strategy of the adversary in the game. As discussed in Sect.~\ref{sec:preliminaries}, pure memoryless optimal strategies exist for $\playerOne$~\cite{liggett_SR69,EM79}. In the following, we denote $\stratWC \in \stratsPureMemoryless_{1}$ such a strategy, and~$\optimalWC$ the optimal mean-payoff value, i.e., the minimal mean-payoff that $\stratWC$ ensures against any strategy of $\playerTwo$. Hence, we have that
$\optimalWC = \inf_{\strat'_{2} \in \strats_{2}} \left\lbrace \mpay(\play) \mid \play \in \outcomesGame{\game}{\initState}{\stratWC}{\strat'_{2}}\right\rbrace = \sup_{\strat_{1} \in \strats_{1}} \inf_{\strat'_{2} \in \strats_{2}} \left\lbrace \mpay(\play) \mid \play \in \outcomesGame{\game}{\initState}{\strat_{1}}{\strat'_{2}}\right\rbrace > 0$,
since we are in a winning EC for $\thresholdWC = 0$.

Similarly, we define a pure memoryless strategy maximizing the expected value in the MDP $\markovProcess = \game[\stratStoch]$ induced by applying the memoryless stochastic model $\stratStoch \in \stratsMemoryless_{2}$ on $\game$. We denote this strategy $\stratExp \in \stratsPureMemoryless_{1}$, and we write the associated expectation as $\optimalExp = \expect_{\initState}^{\markovProcess[\stratExp]}(\mpay)$. Notice that we manipulate equivalently strategies on the game and on the MDP thanks to their shared underlying graph (Sect.~\ref{subsec:mpAssumptions}). The existence of a pure memoryless optimal strategy follows from~\cite{filar1997} (see Sect.~\ref{sec:preliminaries}). However, we here require, \textit{without loss of generality}, that $\stratExp$ is chosen (there may exist several pure memoryless optimal strategies, all yielding the same expected value, by definition) in order to satisfy an additional property: we want that the Markov chain $\markovProcess[\stratExp]$ be \textit{unichain}, i.e., containing a unique recurrent class (i.e., a unique bottom strongly-connected component when considering edges which are assigned non-zero probability in the MC), and possibly some transient states. This will be useful later to apply needed technical results (Lemma~\ref{lem:mp_swecExpDecrease}). It is always possible to choose such a strategy. Intuitively, $\playerTwo$ cannot force the existence of multiple recurrent classes in the MC as it would contradict the fact that we are inside an EC (by definition, $\playerOne$ must be able to force the visit of any state with probability one). Hence, it remains to argue that $\playerOne$ has no interest in inducing multiple recurrent classes, as he cannot increase the expected value by doing so. This is clear since either the different recurrent classes yield the same expectation, in which case one suffices, or they yield different expectations, in which case an optimal strategy like $\stratExp$ will only use the one that produces the maximal expectation ($\playerOne$ has the power to restrict the MC to the class of his choice by definition of EC).

In general, one cannot satisfy the $\BWC$ problem by following only $\stratWC$ or only $\stratExp$, although they suffice when their respective requirements are considered independently. Fortunately, it is possible to build upon those two strategies in order to achieve simultaneous satisfaction with a combined finite-memory strategy.

\bigskip\noindent\textbf{Defining a combined strategy.} Based on the existence of strategies $\stratWC$ and $\stratExp$, we define a pure finite-memory strategy $\stratComb \in \stratsPureFinite_{1}$ that carefully and dynamically alternates between the two memoryless ones to ensure satisfaction of the $\BWC$ problem. Our strategy is \textit{parameterized by two naturals}: $\stepsExp$ and~$\stepsWC$.

\begin{definition}
\label{def:mp_stratComb}
In a game $\game$ satisfying Assumption~\ref{assump:uniqueWEC}, we define the {\em combined strategy} $\stratComb \in \stratsPureFinite_{1}$ as follows.
\begin{itemize}
\item[$\typeA$] Play $\stratExp$ for $\stepsExp$ steps and memorize $\cmbSum \in \integ$, the sum of weights encountered during these $\stepsExp$ steps.

\item[$\typeB$] If $\cmbSum > 0$, then go to $\typeA$. 

Else, play $\stratWC$ during $\stepsWC$ steps then go to $\typeA$.
\end{itemize}
\end{definition}
By step we mean taking any edge in the game, be it from a state of $\playerOne$ or $\playerTwo$. We define \textit{periods} as sequences played from the beginning of a phase of type $\typeA$ or $\typeB$ up to its end, i.e., the beginning of a new period. Intuitively, in a period of type $\typeA$, the strategy mimics the optimal expectation strategy. By playing~$\stratExp$ long enough, we can ensure that the mean-payoff obtained during the period is very close to $\optimalExp$, with probability close to one (Lemma~\ref{lem:mp_swecExpDecrease}). Still, we need to ensure the worst-case threshold in all cases. This may in general not be ensured by periods of type $\typeA$. Hence, we keep a memory of the running sum of weights in the period. This requires a finite number of bits of memory as the sum takes values in $\{-\stepsExp\cdot\largestW, -\stepsExp\cdot\largestW+1, \ldots{}, \stepsExp\cdot\largestW-1, \stepsExp\cdot\largestW\}$. If $\cmbSum > 0$ after the $\stepsExp$ steps, then the mean-payoff over the period satisfies the worst-case requirement (Eq.~\eqref{eq:thresholdWC}) and the strategy can immediately start a new period of type~$\typeA$. Otherwise, it is necessary to compensate in order to satisfy the worst-case requirement. In that case, the strategy mimics the optimal worst-case strategy $\stratWC$ for $\stepsWC$ steps. Such a strategy guarantees that cycles in the outcome have a strictly positive sum of weights since $\optimalWC > 0$, as discussed before. As $\cmbSum$ is lower bounded after $\stepsExp$ steps, there exists a value of $\stepsWC$ such that the total sum of weights (and thus the mean-payoff) over periods $\typeA$ + $\typeB$ is strictly positive. Because all weights are integers, we further deduce that the sum over a period is at least equal to one. By the boundedness of the length of a period, we thus prove that the overall mean-payoff along a play stays \textit{strictly} positive. Hence $\stratComb$ satisfies Eq.~\eqref{eq:thresholdWC}.

While this sketch is sufficient to see that the worst-case requirement is satisfied by strategy $\stratComb$, proving that we can choose $\stepsExp$ and $\stepsWC$ such that its expected value is $\varepsilon$-close to $\optimalExp$ is more involved. Intuitively, periods of type~$\typeB$ must not happen too frequently, nor be too long, in order to have a boundable impact on the overall expectation. The cornerstone to achieve such a moderate impact of periods of type $\typeB$ resides in the fact that a linear increase in $\stepsExp$ produces an exponential decrease in the need for a period of type~$\typeB$ (Lemma~\ref{lem:mp_swecExpDecrease}) whereas it only requires a linear increase in $\stepsWC$ to ensure the worst-case requirement (see Def.~\ref{def:mp_stepsWC} and Lemma~\ref{lem:mp_swecWC}). Note that the need for a decreasing contribution of periods of type $\typeB$ to the overall expectation explains why we need to track the current sum $\cmbSum$ and cannot settle for a simpler strategy that would play periods $\typeA$ and $\typeB$ in strict alternation (cf. Rem.~\ref{rem:mp_swecStrictNotSuf}).

In a nutshell, we claim that under Assumption~\ref{assump:uniqueWEC}, it is always possible to find values for constants $\stepsExp$ and $\stepsWC$ such that strategy $\stratComb$ satisfies the $\BWC$ problem for $(0,\, \optimalExp - \varepsilon)$, as stated in Theorem~\ref{thm:insideWinning}. Before proving it, we illustrate the combined strategy and introduce some intermediary technical results. Notice that implementing $\stratComb$ only requires finite memory as strategies $\stratExp$ and $\stratWC$ are memoryless, constants $\stepsExp$ and $\stepsWC$ have finite values, and $\cmbSum$ takes a finite number of values.

\smallskip\noindent\textbf{Illustration.} Consider the subgame $\gameNonZero \reduc \ec_{3} = \game \reduc \ec_{3}$ of the game in Fig.~\ref{fig:mpRunningExample} and the initial state $\initState = \state_{10}$. The worst-case requirement can be satisfied, that is why the EC is winning. Indeed, always choosing to go to $\state_{9}$ when in $\state_{10}$ is an optimal memoryless worst-case strategy $\stratWC$ that guarantees a mean-payoff $\optimalWC = 1$. Its expectation is $\expect_{\initState}^{\game[\stratWC,\stratStoch]}(\mpay) = 1$. On the other hand, the strategy $\stratExp$ that always selects $\state_{11}$ is optimal regarding the expected value criterion: it induces expectation\footnote{Given an irreducible MC $\markovChainFull$, with $\graphFull$, one can compute its limiting stationary distribution by finding the unique probability vector $\mathbf{v} \in \left[ 0, 1\right] ^{\vert\states\vert}$ such that $\mathbf{v}\mathbf{P}_{\mcTrans} = \mathbf{v}$, where $\mathbf{P}_{\mcTrans}$ denotes the transition matrix derived from $\mcTrans$. The expected mean-payoff can then be obtained by multiplying the row vector $\mathbf{v}$ by the column vector $\mathbf{e} \in \reals^{\vert\states\vert}$ that contains the respective expected weights over outgoing edges for each state. That is: $\forall\, \state \in \states$, $\mathbf{e}(\state) = \sum_{\state' \in \states} \mcTrans(\state)(\state') \cdot \weight((\state, \state'))$, and $\expect^{\markovChain}_{\state} = \mathbf{v}\cdot\mathbf{e}$.} $\optimalExp = \Big(0 + \big(1/2 \cdot 9 + 1/2 \cdot (-1)\big)\Big)/2 = 2$ against the stochastic model $\stratStoch$. However, it can only guarantee a mean-payoff of value $-1/2$ in the worst-case.

\begin{wrapfigure}{r}{60mm}
\vspace{-4mm}
  \centering   
  \scalebox{0.6}{\begin{tikzpicture}[->,>=latex,shorten >=1pt,auto,node
    distance=2.5cm,bend angle=45,font=\normalsize]
    \tikzstyle{p1}=[draw,circle,text centered,minimum size=10mm]
    \tikzstyle{p2}=[draw,rectangle,text centered,minimum size=10mm]
    \tikzstyle{p3}=[draw,diamond,text centered,minimum size=24mm,text width=14mm]
    \tikzstyle{empty}=[]
    \node[p3] (1) at (0,0) {$\state_{10}$\\$\cmbSum > 0$};
    \node[p3] (2) at (-4,0) {$\state_{11}$};
    \node[p3] (3) at (-4,4) {$\state_{10}$\\$\cmbSum \leq 0$};
    \node[p3] (4) at (0,4) {$\state_{9}$};
    \node[empty] (a) at (-4.2, 1.3) {$\frac{1}{2}$};
    \node[empty] (b) at (-3.4, -0.9) {$\frac{1}{2}$};
    \coordinate[shift={(8mm,0mm)}] (init) at (1.east);
    \path
    (1) edge node[above] {$0$} (2)
    (2) edge node[left] {$-1$} (3)
    (3) edge node[above] {$1$} (4)
    (4) edge node[right] {$1$} (1)
    (init) edge (1)
    ;
	\draw[->,>=latex] (2) to[out=320,in=220] node[below] {$9$} (1);
      \end{tikzpicture}}
      \caption{Markov chain $\game[\stratComb, \stratStoch]$ induced by the combined strategy $\stratComb$ and the stochastic model $\stratStoch$ over the winning EC $\ec_{3}$ of $\game$.}
\label{fig:mp_insideSWEC_MC}
\vspace{-8mm}
\end{wrapfigure}

By Theorem~\ref{thm:insideWinning}, we know that it is possible to find finite-memory strategies satisfying the $\BWC$ problem for any thresholds pair $(0,\, 2 - \varepsilon)$, $\varepsilon > 0$. In particular, consider the thresholds pair $(0,\, 3/2)$. We build a combined strategy~$\stratComb$ as described in Def.~\ref{def:mp_stratComb}. Let $\stepsExp = \stepsWC = 2$: the strategy plays the edge $(\state_{10}, \state_{11})$ once, then if the edge of value $9$ has been chosen by $\playerTwo$, it chooses $(\state_{10}, \state_{11})$ again; otherwise it chooses the edge $(\state_{10}, \state_{9})$ once and then resumes choosing $(\state_{10}, \state_{11})$. This strategy satisfies the $\BWC$ problem. In the worst-case, $\playerTwo$ always chooses the $-1$ edge, but each time he does so, the $-1$ is followed by two $+1$ thanks to the cycle $\state_{10} \state_{9} \state_{10}$. Strategy $\stratComb$ hence guarantees a mean-payoff equal to $(0 - 1 + 1 + 1)/4 = 1/4 > 0$ in the worst-case. For the expected value requirement, we can build the Markov chain $\game[\stratComb, \stratStoch]$ (Fig.~\ref{fig:mp_insideSWEC_MC}) and check that its expectation is $\expect_{\initState}^{\game[\stratComb,\stratStoch]}(\mpay) = 5/3 > 3/2$.

\begin{remark}
\label{rmk:mp_memorylessNotEnough}
Memoryless strategies do not suffice for the $\BWC$ problem, even with randomization. Indeed, the edge $(\state_{10}, \state_{11})$ cannot be assigned a non-zero probability as it would endanger the worst-case requirement (since the outcome~$(\state_{10}\state_{11})^{\omega}$ cycling on the edge of weight $-1$ would exist and have a negative mean-payoff). Hence, the only acceptable memoryless strategy is $\stratWC$, which has only an expectation of $1$.
\end{remark}

\smallskip\noindent\textbf{Technical results.} Before proving the correctness of strategy $\stratComb$, we introduce an important property verified by the MC $\game[\stratExp,\stratStoch]$. It is known that in the long-run, the probability of outcomes that induce a mean-payoff equal to the expectation of the MC is one. Lemma~\ref{lem:mp_swecExpDecrease} shows that, for sufficiently long prefixes, it is possible to bound the probability of having a mean-payoff which differs from the expected value by more than a given~$\varepsilon > 0$. In particular, it implies that for sufficiently large values of $\stepsExp$, this probability decreases exponentially with $\stepsExp$. This will help us bound the impact of periods of type $\typeB$ on the overall expectation.

\begin{lemma}[{Follows from the extension of~\cite[Thm. 2]{glynn_SPL2002} proposed in~\cite{tracol_ORL2009}}]
\label{lem:mp_swecExpDecrease}
For all initial state $\initState$, for all $\varepsilon > 0$, there exists constants \mbox{$c_{1}, c_{2} >0$} and $\stepsExp_{0} \in \nat$ such that, for all $\stepsExp \geq \stepsExp_{0}$,
\begin{equation*}
\proba^{\game[\stratExp,\stratStoch]}_{\initState}\Big( \play \in \plays{\game[\stratExp,\stratStoch]} \;\Big\vert\; \big\vert\mpay(\play(\stepsExp)) - \optimalExp\big\vert \geq \varepsilon\Big) \leq \expDecFct{\stepsExp}{\varepsilon} = \dfrac{c_{1}}{e^{c_{2}\cdot\stepsExp\cdot\varepsilon^{2}}}.
\end{equation*}
\end{lemma}
In~\cite[Thm. 2]{glynn_SPL2002}, Glynn and Ormoneit present an extension of Hoeffding's inequality~\cite{hoeffding_JASA1963} for uniformly ergodic Markov chains. Straight application of this result in our setting is not possible, as the MC $\game[\stratExp, \stratStoch]$ does not need to be aperiodic in general. Nevertheless, this result is extended to unichain MCs (possibly periodic) in~\cite{tracol_ORL2009}. Hence, Lemma~\ref{lem:mp_swecExpDecrease} reformulates the latter result in our precise context. Notice that the MC $\game[\stratExp, \stratStoch]$ is unichain thanks to the choice of $\stratExp$, as presented earlier.

\begin{proof}
For the sake of completeness, we sketch the main steps proposed by Tracol~\cite{tracol_ORL2009} to extend the results of~\cite{glynn_SPL2002}. Consider the MC $\game[\stratExp, \stratStoch]$: it can be decomposed in a set of transient states and a unique recurrent class, i.e., a bottom strongly-connected component. Assume this BSCC is periodic, of period $d \in \natStrict$. We can decompose it in $d$ aperiodic classes on which we are able to apply the bound provided by~\cite[Thm. 2]{glynn_SPL2002}. The key idea is then to obtain a unified bound by aggregating the bounds obtained for each aperiodic MC, given sufficiently large values of the constants.

A word on the constants of Lemma~\ref{lem:mp_swecExpDecrease}. Careful analysis of the proofs of~\cite[Thm. 2]{glynn_SPL2002} and~\cite[Prop. 2]{tracol_ORL2009} reveals that $c_{1}$ is exponential in $\varepsilon$ and polynomial in the characteristics of the MC, while $c_{2}$ is only polynomial in the characteristics of the MC. More importantly, constant $\stepsExp_{0}$ is polynomial in the size of the MC and polynomial in the largest weight $\largestW$ (exponential in its encoding).
\end{proof}

\smallskip\noindent\textbf{Analysis of the combined strategy.} Our goal is to show that for any $\varepsilon > 0$, there exist two naturals $\stepsExp$ and $\stepsWC$ such that~$\stratComb$ proves the correctness of Theorem~\ref{thm:insideWinning}. First, we define $\stepsWC$ as a \textit{linear} function of~$\stepsExp$. Note that the main purpose of~$\stepsExp$ is to create periods of type $\typeA$ long enough to have an expected mean-payoff close to the optimal value achieved by $\stratExp$, i.e., $\optimalExp$. The aim of $\stepsWC$ is for periods of type $\typeB$ to be long enough to compensate the possible negative effect of periods of type $\typeA$ and thus ensure the worst-case requirement. As stated before, $\stepsWC$ should not grow too quickly to preserve an overall mean-payoff which is mainly influenced by periods of type $\typeA$ (hence close to the optimal expectation $\optimalExp$).

\begin{definition}
\label{def:mp_stepsWC}
Given a natural constant $\stepsExp \in \nat$, we define $\stepsWC = \left\lfloor\dfrac{\stepsExp\cdot\largestW + \statesSize\cdot\largestW + \statesSize \cdot \optimalWC}{\optimalWC}\right\rfloor + 1$.
\end{definition}

\begin{remark}
\label{rem:mp_swecStrictNotSuf}
Obviously, $\stepsWC$ needs to be proportional to $\stepsExp$ to preserve the worst-case requirement as the amount to compensate is bounded by $\stepsExp\cdot\largestW$. Consequently, we can observe the need for the bookkeeping of $\cmbSum$ in strategy $\stratComb$ in contrast to a non-dynamic alternation scheme between periods of type $\typeA$ and $\typeB$. Indeed, in a strategy following the latter scheme, the long-term expectation would be close to $\frac{\stepsExp\cdot\optimalExp + \stepsWC\cdot\optimalWC}{\stepsExp+\stepsWC}$. As $\stepsWC$ is not constant but proportional to $\stepsExp$, one can easily see that this expression does not tend to $\optimalExp$ when $\stepsExp$ tends to infinity (which is required when $\varepsilon$ tends to zero according to Lemma~\ref{lem:mp_swecExpDecrease}). Thus, strict alternation does not suffice to satisfy the thresholds pair presented in Thm.~\ref{thm:insideWinning}.
\end{remark}

Under the value of $\stepsWC$ given in Def.~\ref{def:mp_stepsWC}, Lemma~\ref{lem:mp_swecWC} states that the worst-case requirement is satisfied. The idea is to decompose any play into an infinite sequence of periods, each of them having a bounded length and ensuring a strictly positive sum of weights, yielding an overall strictly positive mean-payoff of the play.

\begin{lemma}
\label{lem:mp_swecWC}
For any $\stepsExp \in \nat$, the combined strategy $\stratComb \in \stratsPureFinite_{1}$ is such that for all $\strat_{2} \in \strats_{2}$, for all $\play \in \outcomesGame{\game}{\initState}{\stratComb}{\strat_{2}}$, we have that $\mpay(\play) > 0$.
\end{lemma}

\begin{proof}
Let $\strat_{2}$ be an arbitrary strategy of $\playerTwo$ and $\play$ any outcome in $\outcomesGame{\game}{\initState}{\stratComb}{\strat_{2}}$. By definition of $\stratComb$, we decompose the play in a sequence of periods of type $\typeA$ and $\typeB$. That is, $\play = \prefix_{0}\prefix_{1}\prefix_{2}\ldots{}$ where, for all $i \geq 0$, $\prefix_{i}$ is a finite sequence of states that is either of length $\stepsExp$ if $\prefix_{i}$ is of type $\typeA$ or of length $\stepsWC$ if $\prefix_{i}$ is of type $\typeB$. Moreover, $\prefix_{0}$ is of type $\typeA$ and for all $i \geq 1$, $\prefix_{i}$ is of type $\typeB$ if and only if $\prefix_{i-1}$ is of type~$\typeA$ and such that $\cmbSum(\prefix_{i-1}) \leq 0$ (i.e., the sum of weights along the sequence is not strictly positive). We regroup each sequence of type $\typeB$ with its predecessor of type $\typeA$ and obtain $\play = \prefix'_{0}\prefix'_{1}\prefix'_{2}\ldots{}$ such that all sequences are either of type $\typeA$ or of type $\typeA + \typeB$.

Consider any sequence $\prefix'_{i}$ of type $\typeA$. Since it is not followed by a period of type $\typeB$, we know that $\cmbSum(\prefix'_{i}) > 0$. Now consider any sequence  $\prefix'_{i}$ of type $\typeA + \typeB$. At the end of the period of type $\typeA$, the sum is bounded from below by~$-\stepsExp\cdot\largestW$. During the period of type $\typeB$, an optimal \textit{memoryless} worst-case strategy $\stratWC$ is followed and consequently, all formed cycles have a mean-payoff at least equal to $\optimalWC$. Hence, the sum of weights over the period of type $\typeB$ is at least $(\stepsWC - \statesSize)\cdot\optimalWC - \statesSize\cdot\largestW$. By Def.~\ref{def:mp_stepsWC}, this induces that the overall sum over the sequence of type $\typeA + \typeB$ is $\cmbSum(\prefix'_{i}) > 0$.

In both cases, we have that $\cmbSum(\prefix'_{i}) > 0$. But since weights are integers, this implies a stronger inequality: $\cmbSum(\prefix'_{i}) \geq 1$. We go back to the play seen as an infinite sequence of states $\play = \state_{0}\state_{1}\state_{2}\ldots{}$. Thanks to the previous observations, we can now state that the total sum of weights up to any state $\state_{t}$, $t \geq 0$, is bounded by $\tpay(\play(t)) \geq -\big[ t \bmod (\stepsExp+\stepsWC)\big] \cdot \largestW + \left\lfloor \frac{t}{\stepsExp + \stepsWC}\right\rfloor \cdot 1 \geq -(\stepsExp+\stepsWC)\cdot \largestW + \frac{t}{\stepsExp + \stepsWC} - 1$.
Hence the mean-payoff of the play $\play$ is
$\mpay(\play) \geq \liminf_{t \rightarrow \infty} \left[ \frac{-(\stepsExp+\stepsWC)\cdot \largestW}{t} + \frac{1}{\stepsExp + \stepsWC} - \frac{1}{t}\right]  = \frac{1}{\stepsExp + \stepsWC} > 0$,
which concludes the proof.
\end{proof}

It remains to show that for any $\varepsilon > 0$, there exists a value $\stepsExp \in \nat$ such that the expected value requirement (Eq.~\eqref{eq:thresholdExp}) is also satisfied by $\stratComb$. This is proved in Lemma~\ref{lem:mp_swecExp}. Again, decomposition of plays into periods needs to be considered. Furthermore, the crux of the proof resides in the use of Lemma~\ref{lem:mp_swecExpDecrease}: it induces that when the length of periods of type~$\typeA$ grows linearly, the probability of periods of type~$\typeB$ decreases exponentially. Since the length of periods of type~$\typeB$ only grows linearly in~$\stepsExp$, the overall impact of periods~$\typeB$ in the expectation tends to zero when~$\stepsExp$ tends to infinity. Hence, the expectation of~$\stratComb$ tends to the optimal expectation~$\optimalExp$ and classical convergence analysis provides the result.

\begin{lemma}
\label{lem:mp_swecExp}
For any $\varepsilon > 0$, there exists $\stepsExp \in \nat$ such that $\,\expect_{\initState}^{\game[\stratComb,\stratStoch]}(\mpay) > \optimalExp - \varepsilon$.
\end{lemma}

\begin{proof}
For the proof, we assume that $\varepsilon \leq \optimalExp$, otherwise the claim is obviously true for any $\stepsExp \in \nat$ since the mean-payoff of any outcome consistent with $\stratComb$ is strictly positive by application of Lemma \ref{lem:mp_swecWC}. Now, for a given $\stepsExp \in \nat$, consider the corresponding finite MC $\markovChain(K) = \game[\stratComb,\stratStoch]$ where 
$\stratComb$ is defined with the parameter $K$ (i.e., periods of type~$\typeA$ are of length $K$). To obtain the claim, we prove that $\expect_{\initState}^{\markovChain(K)}(\mpay) \xrightarrow[\stepsExp \rightarrow \infty]{} \optimalExp$. Similarly to the proof of Lemma~\ref{lem:mp_swecWC}, any outcome of $\markovChain(K)$ is an outcome consistent with $\stratComb$ in $\game$ and thus can be decomposed in an infinite sequence of periods of types $\typeA$  and $\typeA+\typeB$.

To begin with, we will consider (i) the expectation over \textit{one} period of type $\typeA$, and (ii) the expectation over \textit{one} period of type $\typeA + \typeB$, as well as the respective probabilities of seeing such periods whenever a new period begins. Note that the probability of a period being of type $\typeA$ or $\typeA + \typeB$ does not depend on whether the previous one was of type $\typeA$ or $\typeA + \typeB$, by definition of $\stratComb$. Hence, the probability that some 
period has a specific type, and the expected value of that period, are not influenced by what happened over previous periods. For the sake of simplicity, we omit that different periods do not necessarily start in the same state, as the resulting impact on the expectation is negligible for sufficiently long periods.

(i) Let us first consider periods of type $\typeA$. By Def. \ref{def:mp_stratComb}, $\playerOne$ follows the strategy $\stratExp$ during those periods. Recall that $\expect_{\initState}^{\game[\stratExp,\stratStoch]}(\mpay) = \optimalExp$. Let us denote by $e_{\typeA}$ the expected mean-payoff over a period of type $\typeA$. By Lemma \ref{lem:mp_swecExpDecrease}, for any $\varepsilon > 0$, there exists $\stepsExp_{0} \in \nat$ such that for all $\stepsExp \geq \stepsExp_{0}$, this expected value is bounded from below by $e_{\typeA}\geq \left( 1 - \expDecFct{\stepsExp}{\varepsilon}\right) \cdot (\optimalExp - \varepsilon) + \expDecFct{\stepsExp}{\varepsilon}\cdot x$, with $x$ a lower bound on \textit{any} consistent outcome. Since any period of type $\typeA$ is not followed by a period of type $\typeB$ (otherwise it would be considered as a period of type $\typeA+\typeB$),  we know that the sum of weights along the 
$K$ steps of the period is greater than or equal to $1$, and therefore we can take $x \geq 1/\stepsExp$ (cf.~proof of Lemma \ref{lem:mp_swecWC}).

(ii) Now, consider periods of type $\typeA+\typeB$. As shown in the proof of Lemma \ref{lem:mp_swecWC}, we know that the expected mean-payoff over such a period, denoted by $e_{\typeA+\typeB}$, is
greater than or equal to $1/(\stepsExp+\stepsWC)$. Now consider the probability $p_{\typeA+\typeB}$ that a period is of type $\typeA+\typeB$. By definition of strategy $\stratComb$, $p_{\typeA+\typeB}$ is equal to the probability of having a total sum of weights less than or equal to $0$ after $\stepsExp$ steps of playing according to $\stratExp$. Since we assume that $\varepsilon \leq \optimalExp$ and since $\optimalExp \geq \optimalWC > 0$, we can bound from above this probability by the probability to have a mean-payoff less than $\optimalExp - \varepsilon$ over the $\stepsExp$ steps. Repeating the argument of point~(i) and applying Lemma \ref{lem:mp_swecExpDecrease}, we deduce that $p_{\typeA+\typeB} \leq \expDecFct{\stepsExp}{\varepsilon}$ for all $K\geq K_0$.

We now have sufficient arguments to study the overall expected value over $m$ periods and express it as a weighted average between the two types of periods considered. Let $p_{\typeA}$ denote the probability of periods of type $\typeA$. Note that $p_{\typeA} = 1-p_{\typeA+\typeB}$. The expected number of periods of type $\typeA$ is $m\cdot p_{\typeA}$. Similarly, $m\cdot p_{\typeA+\typeB}$ is the expected
number of periods of type $\typeA+\typeB$. The expected sum of weights over the $m$ periods is therefore $m\cdot p_{\typeA} \cdot e_{\typeA} \cdot \stepsExp + m\cdot p_{\typeA+\typeB} \cdot e_{\typeA+\typeB} \cdot (\stepsExp + \stepsWC)$ since periods of type $\typeA$ have length $K$ and periods of type~$\typeA+\typeB$ length $K+L$. The expected length of the $m$ periods is $m\cdot p_{\typeA} \cdot \stepsExp + m\cdot p_{\typeA+\typeB} \cdot (\stepsExp + \stepsWC)$. Finally, we have that

\begin{equation}
\label{eq:mp_expPeriod}
\expect_{\initState,\,m\,\text{periods}}^{\markovChain(K)}(\mpay) = \dfrac{m\cdot p_{\typeA} \cdot e_{\typeA} \cdot \stepsExp + m\cdot p_{\typeA+\typeB} \cdot e_{\typeA+\typeB} \cdot (\stepsExp + \stepsWC)}{m\cdot p_{\typeA} \cdot \stepsExp + m\cdot p_{\typeA+\typeB} \cdot (\stepsExp + \stepsWC)}.
\end{equation}
Clearly, this expression does not depend on the number of periods $m$. This is consistent with our analysis since we argued that periods of type $\typeA$ and $\typeA + \typeB$ are statistically independent. Also, note that this reasoning is only correct for complete periods. Nonetheless, any prefix $\play(n)$, $n \geq 0$, of the play is composed of a sequence of complete periods followed by a suffix of length bounded by $(\stepsExp + \stepsWC)$ and of total sum of weights bounded by $-(\stepsExp + \stepsWC)\cdot\largestW$ and $(\stepsExp + \stepsWC)\cdot\largestW$. Hence, we characterize the expected mean-payoff over the $n$ first steps of a play as
\begin{equation*}
\dfrac{\left\lfloor\dfrac{n}{l}\right\rfloor \cdot l \cdot \expect_{\initState,\,1\,\text{period}}^{\markovChain(K)}(\mpay) - (\stepsExp + \stepsWC) \cdot \largestW}{n} \leq \expect_{\initState,\,n\,\text{steps}}^{\markovChain(K)}(\mpay) \leq \dfrac{\left\lfloor\dfrac{n}{l}\right\rfloor \cdot l \cdot \expect_{\initState,\,1\,\text{period}}^{\markovChain(K)}(\mpay) + (\stepsExp + \stepsWC) \cdot \largestW}{n},
\end{equation*}
with $l = p_{\typeA} \cdot \stepsExp + p_{\typeA+\typeB} \cdot (\stepsExp + \stepsWC)$, the expected length of a period. Naturally, the finite suffix proves to be negligible when $n$ grows, hence
$\expect_{\initState}^{\markovChain(K)}(\mpay) = \liminf_{n \rightarrow \infty} \left[ \expect_{\initState,\,n\,\text{steps}}^{\markovChain(K)}(\mpay)\right] = \expect_{\initState,\,1\,\text{period}}^{\markovChain(K)}(\mpay)$.
Observe that this equation uses the equality between the expectation over the values of plays and the limit of the expectation over values of prefixes. This equality is verified for the mean-payoff value function but does not need to be true for arbitrary value functions.

Back to Eq.~\eqref{eq:mp_expPeriod}, with $m = 1$, we use $e_{\typeA+\typeB} \geq 1/(\stepsExp+\stepsWC) > 0$ and assume $\stepsExp > 0$ to obtain the inequality
\[\expect_{\initState}^{\markovChain(K)}(\mpay) \geq \frac{p_{\typeA} \cdot e_{\typeA}}{p_{\typeA} + p_{\typeA+\typeB} \cdot \left( \frac{\stepsExp + \stepsWC}{\stepsExp}\right) }\raisepunct{.}\]
Again, we assume $\stepsExp$ large enough to ensure $p_{\typeA} > 0$ (such a $\stepsExp$ exists by consequence of Lemma \ref{lem:mp_swecExpDecrease}) and get
\begin{equation*}
\expect_{\initState}^{\markovChain(K)}(\mpay) \geq \frac{e_{\typeA}}{1 + \left(\frac{p_{\typeA+\typeB}}{p_{\typeA}}\right) \cdot \left( \frac{\stepsExp + \stepsWC}{\stepsExp}\right) }\raisepunct{.}
\end{equation*}
By (i) and (ii), we have $p_{\typeA} \geq 1 - \expDecFct{\stepsExp}{\varepsilon}$, $p_{\typeA+\typeB} \leq \expDecFct{\stepsExp}{\varepsilon}$, and $e_{\typeA} \geq \left( 1 - \expDecFct{\stepsExp}{\varepsilon}\right) \cdot (\optimalExp - \varepsilon)$. We derive that
\begin{equation}
\label{eq:mp_expIneq}
\expect_{\initState}^{\markovChain(K)}(\mpay) \geq \dfrac{\left( 1 - \expDecFct{\stepsExp}{\varepsilon}\right) \cdot (\optimalExp - \varepsilon)}{1 + \frac{\expDecFct{\stepsExp}{\varepsilon} \cdot (\stepsExp + \stepsWC)}{(1 - \expDecFct{\stepsExp}{\varepsilon}) \cdot \stepsExp} }\raisepunct{.}
\end{equation}
Recall that ultimately, we want to prove that $\expect_{\initState}^{\markovChain(K)}(\mpay) \xrightarrow[\stepsExp \rightarrow \infty]{} \optimalExp$. Consider what happens when $\stepsExp \rightarrow \infty$ in Eq.~\eqref{eq:mp_expIneq}: notice that $\stepsWC$ is linear in $\stepsExp$ by Def. \ref{def:mp_stepsWC}, hence $\stepsWC \rightarrow \infty$, and that $\expDecFct{\stepsExp}{\varepsilon} \rightarrow 0$ by Lemma \ref{lem:mp_swecExpDecrease}. This does not suffice to conclude on the possible convergence of the lower bound given in Eq.~\eqref{eq:mp_expIneq}. The crux of the argument is given by Lemma \ref{lem:mp_swecExpDecrease}: $\expDecFct{\stepsExp}{\varepsilon}$ decreases exponentially for a linear increase in $\stepsExp$. Thus, we have that $\expDecFct{\stepsExp}{\varepsilon} \cdot (\stepsExp + \stepsWC) \rightarrow 0$. Therefore,
\begin{equation}
\label{eq:mp_expIneqLimit}
\lim_{\stepsExp \rightarrow \infty} \left[\dfrac{\left( 1 - \expDecFct{\stepsExp}{\varepsilon}\right) \cdot (\optimalExp - \varepsilon) }{1 + \frac{\expDecFct{\stepsExp}{\varepsilon} \cdot (\stepsExp + \stepsWC)}{(1 - \expDecFct{\stepsExp}{\varepsilon}) \cdot \stepsExp} }\right] = \optimalExp - \varepsilon.
\end{equation}
Notice two facts. First, for any $\stepsExp \in \nat$, we have that $\expect_{\initState}^{\markovChain(K)}(\mpay) \leq \expect_{\initState}^{\game[\stratExp,\stratStoch]}(\mpay) = \optimalExp$ by the optimality of~$\stratExp$. Second, Eq.~\eqref{eq:mp_expIneqLimit} is valid for any $\varepsilon$ such that $\nu^*\geq \varepsilon > 0$. Hence we observe that the sequence of expected values $(\expect_{\initState}^{\markovChain(\stepsExp)}(\mpay))_{\stepsExp \geq 0}$ is bounded from above and from below by two sequences converging to $\optimalExp$. Ergo
$\expect_{\initState}^{\markovChain(\stepsExp)}(\mpay) \xrightarrow[\stepsExp \rightarrow \infty]{} \optimalExp$.
By definition of convergence, for any $\varepsilon$ such that $\nu^*\geq \varepsilon > 0$, there exists $\stepsExp \in \nat$ such that $\expect_{\initState}^{\markovChain(\stepsExp)}(\mpay) > \optimalExp - \varepsilon$, which concludes our proof.
\end{proof}

Based on Lemma~\ref{lem:mp_swecWC} and Lemma~\ref{lem:mp_swecExp}, Theorem~\ref{thm:insideWinning} follows and concludes our analysis of winning ECs with no edges of probability zero. We extend these results to arbitrary winning ECs in the next section.

\subsection{\textbf{Starting in a winning end-component: witness-and-secure strategy}}
\label{subsec:mpInsideWeakly}

We now turn to winning ECs with potentially edges of probability zero, i.e., $\edgesNonZero \subseteq \edges$. We present in this section how to construct a finite-memory strategy that can benefit $\varepsilon$-closely from the maximal expectation achievable in such winning ECs when facing the stochastic model $\stratStoch$, while guaranteeing satisfaction of the worst-case requirement even against \textit{arbitrary} strategies of $\playerTwo$ (i.e., strategies that may use edges in $\edges \setminus \edgesNonZero$). It is crucial to notice that we now consider \textit{a complete game}, i.e., not necessarily reduced to a single winning EC as in Sect.~\ref{subsec:mpInsideStrongly}. Still, we assume that the play starts in a winning EC: consistent outcomes will stay in it when $\playerTwo$ follows $\stratStoch$ (because $\playerOne$ will have no interest to leave), but may exit the EC if $\playerTwo$ takes edges of probability zero, either by the action of $\playerTwo$ (recall there may exist edges that leave the EC in $\edges \setminus \edgesNonZero$) or the action of $\playerOne$ (which may need to leave to guarantee a strictly positive mean-payoff).

\begin{theorem}
\label{thm:insideWeakly}
Let $\gameFull$ be a two-player game, $\graphFull$ its underlying graph, $\stratStoch \in \stratsMemoryless_{2}$ a memoryless stochastic model of $\playerTwo$, $\markovProcess = \game[\stratStoch] = (\graph, \statesOne, \statesProb = \statesTwo, \mpTrans = \stratStoch)$ the resulting MDP, $\ec \in \winningECs$ a winning EC, $\initState \in \ec$ an initial state inside the EC, and $\optimalExp \in \rat$ the maximal expected value achievable by $\playerOne$ in $\markovProcess \reduc \ec$. Then, for all $\varepsilon > 0$, there exists a finite-memory strategy of $\playerOne$ that satisfies the $\BWC$ problem for the thresholds pair~$(0,\, \optimalExp - \varepsilon)$.
\end{theorem}
We prove this theorem in the following. Let us first give some key intuition. With respect to the expected value requirement of Eq.~\eqref{eq:thresholdExp}, the hypothesis is that $\playerTwo$ \textit{will} follow its stochastic model $\stratStoch$. Hence, he will only play edges in $\edgesNonZero$ and the combined strategy proposed in Sect.~\ref{subsec:mpInsideStrongly} suffices to achieve the claimed expectation and guarantee the worst-case threshold against $\stratStoch$ (basically, we can apply Thm.~\ref{thm:insideWinning} on $\gameNonZero$). Now, we still have to account for arbitrary strategies of $\playerTwo$ to satisfy the worst-case requirement of Eq.~\eqref{eq:thresholdWC}. Notice that the combined strategy suffices to ensure it against all strategies playing exclusively in $\edgesNonZero$. So, it only remains to deal with strategies that choose some edges in $\edges \setminus \edgesNonZero$. It is easy for $\playerOne$ to \textit{witness} if~$\playerTwo$ chooses an edge outside $\edgesNonZero$ (as the stochastic model is assumed known by $\playerOne$). If this happens, $\playerOne$ can \textit{secure} its mean-payoff value by switching from the combined strategy to a worst-case winning strategy, which exists in all states due to the preprocessing of the game (Sect.~\ref{subsec:mpAssumptions}).

\smallskip\noindent\textbf{Basis strategies.} We denote by $\stratSecure \in \stratsPureMemoryless_{1}(\game)$ a pure memoryless worst-case winning strategy on $\game$. This strategy exists in all states of the game due to the preprocessing, including states of the EC $\ec \in \winningECs$. Still, it may require leaving the EC to ensure a strictly positive mean-payoff (because the definition of winning ECs only consider edges in $\edgesNonZero$).
When using $\stratStoch$, $\playerTwo$ cannot force leaving the winning EC $\ec \in \winningECs$. By Thm.~\ref{thm:insideWinning}, $\playerOne$ has a finite-memory strategy on $\gameNonZero \reduc \ec$, denoted $\stratComb$, that ensures Eq.~\eqref{eq:thresholdExp}, and verifies Eq.~\eqref{eq:thresholdWC}, \textit{if we restrict} $\playerTwo$ to strategies in $\strats_{2}(\gameNonZero)$.

\smallskip\noindent\textbf{Defining a witness-and-secure strategy.} In order to prove Thm.~\ref{thm:insideWeakly}, we define a pure finite-memory \textit{witness-and-secure strategy} as follows.

\begin{definition}
\label{def:mp_stratWitAndFlee}
In a game $\game$ such that $\markovProcess = \game[\stratStoch]$, $\ec \in \winningECs$ is a winning EC and $\initState \in \ec$ is the initial state, we define the {\em witness-and-secure strategy} $\stratWNS \in \stratsPureFinite_{1}(\game)$ as follows.
\begin{itemize}
\item[(i)] Play the combined strategy $\stratComb \in \stratsPureFinite_{1}(\game_{\mpTrans} \reduc \ec)$ as long as $\playerTwo$ picks edges in $\edgesNonZero$.

\item[(ii)] As soon as $\playerTwo$ takes an edge in $\edges \setminus \edgesNonZero$,\footnote{More complex switching schemes could be used, such as only switching if the edge taken is really dangerous (i.e., part of a non strictly positive cycle), switching after a bounded number of deviations from the support, etc. But this simple scheme proves to be sufficient to realize Thm.~\ref{thm:insideWeakly} and is easier to analyze.} play the worst-case winning strategy $\stratSecure \in \stratsPureMemoryless_{1}(\game)$ forever.
\end{itemize}
\end{definition}
This strategy makes use of the combined strategy $\stratComb$ presented in Def.~\ref{def:mp_stratComb}. Hence it is similarly parameterized by two naturals $\stepsExp$ and $\stepsWC$ that respectively define the lengths of periods of type $\typeA$ and of type $\typeB$ in  Def.~\ref{def:mp_stratComb}. Again, we will show that for any $\varepsilon > 0$, we can find values for $\stepsExp$ and $\stepsWC$ such that Thm.~\ref{thm:insideWeakly} is verified. Notice that $\stratWNS$ is finite-memory since $\stratComb$ and $\stratSecure$ are too and watching for the appearance of an edge belonging to $\edges \setminus \edgesNonZero$ in the actions of the adversary only requires a finite amount of memory.

Intuitively, the witness-and-secure strategy acts as follows. As long as $\playerTwo$ conforms to $\edgesNonZero$, playing in $\game$ is essentially the same as playing in $\gameNonZero$. Hence $\stratWNS$ prescribes acting like $\stratComb$, which induces satisfaction of the $\BWC$ problem in $\gameNonZero \reduc \ec$ by Thm.~\ref{thm:insideWinning}. Two requirements must be satisfied by $\stratWNS$: (a) the worst-case and (b) the expected value. First consider (a). Two situations may occur. Either the outcome is such that $\stratWNS$ always stays in phase (i) and strategy $\stratComb$ is used forever in $\gameNonZero \reduc \ec$, in which case direct application of Thm.~\ref{thm:insideWinning} suffices to prove that Eq.~\eqref{eq:thresholdWC} is satisfied. Or, the outcome is such that $\stratWNS$ switches to phase (ii), in which case satisfaction of the worst-case requirement follows by definition of $\stratSecure$. Now consider (b). Notice that the only consistent outcomes always stay in phase~(i), by definition of the set $\outcomesGame{\game}{\initState}{\stratWNS}{\stratStoch}$ which does not allow for choices outside of $\edgesNonZero$. Hence the overall expectation is equal to the one over outcomes staying in phase~(i). By Def.~\ref{def:mp_stratWitAndFlee}, the latter is exactly the expectation of~$\stratComb$, which satisfies the threshold by Thm.~\ref{thm:insideWinning}. Thus, the existence of fitting values of $\stepsExp$ and $\stepsWC$ for Thm.~\ref{thm:insideWeakly} is guaranteed.

\smallskip\noindent\textbf{Illustration.} Consider the winning EC $\ec_{2}$ in the game of Fig.~\ref{fig:mpRunningExample} and the initial state $\state_{6} \in \ec_{2}$. Notice that $\playerOne$ can ensure a strictly positive mean-payoff in the subgame $\gameNonZero \reduc \ec_{2}$, but not in $\game \reduc \ec_{2}$. Indeed, by always choosing the $-1$ edges (which requires edge $(\state_{7}, \state_{6}) \in \edgesNonZero \setminus \edges$), $\playerTwo$ can ensure a negative mean-payoff whatever the strategy of $\playerOne$. However, there exists a strategy that ensures Eq.~\eqref{eq:thresholdWC}, i.e., yields a strictly positive mean-payoff against any strategy in $\strats_{2}(\game)$, by leaving the EC. Let $\stratSecure$ be the memoryless strategy that takes the edge $(\state_{6}, \state_{9})$ and then cycle on $(\state_{10}\state_{9})^{\omega}$ forever: it guarantees a mean-payoff of $1 > 0$.

For a moment, consider the EC $\ec_{2}$ in $\gameNonZero$. Graphically, it means that the $-1$ edge from $\state_{7}$ to $\state_{6}$ disappears. In the subgame $\gameNonZero \reduc \ec_{2}$, there are two particular memoryless strategies. The optimal worst-case strategy $\stratWC$ guarantees a mean-payoff of $1/2 > 0$ by choosing to go to $\state_{7}$. The optimal expectation strategy $\stratExp$ yields an expected mean-payoff of $3$ by choosing to go to $\state_{8}$ (notice this strategy yields the same expectation in $\markovProcessNonZero \reduc \ec_{2}$ and $\markovProcess \reduc \ec_{2}$). Based on them, we build the combined strategy $\stratComb \in \stratsPureFinite_{1}(\game_{\mpTrans} \reduc \ec_{2})$ as defined in Def.~\ref{def:mp_stratComb} and by Thm.~\ref{thm:insideWinning}, for any $\varepsilon > 0$, there are values of $\stepsExp$ and $\stepsWC$ such that it satisfies the $\BWC$ problem for thresholds $(0,\, 3-\varepsilon)$ in $\gameNonZero \reduc \ec_{2}$. For example, for parameters $\stepsExp = \stepsWC = 2$, we have that $\expect_{\state_{6}}^{(\markovProcess_{\mpTrans} \reduc \ec_{2})[\stratComb]}(\mpay) = \expect_{\state_{6}}^{(\markovProcess \reduc \ec_{2})[\stratComb]}(\mpay) = 13/6$.

We construct the witness-and-secure strategy $\stratWNS \in \stratsPureFinite_{1}(\game)$ based on $\stratComb$ and $\stratSecure$ as described by Def.~\ref{def:mp_stratWitAndFlee}. In this case, that means playing as $\stratComb$ until the $-1$ edge from $\state_{7}$ to $\state_{6}$ is taken by $\playerTwo$. As previously sketched, such a strategy ensures a worst-case mean-payoff equal to $1 > 0$ thanks to $\stratSecure$ and yields an expectation $\expect_{\state_{6}}^{\markovProcess[\stratWNS]}(\mpay) = 13/6$ for $\stepsExp = \stepsWC = 2$.

Finally, notice that securing the mean-payoff by switching to phase (ii) of $\stratWNS$ \textit{is needed} to satisfy the worst-case requirement if $\playerTwo$ plays in $\edges \setminus \edgesNonZero$. Also, observe that it is still necessary to alternate according to $\stratComb$ in $\game_{\mpTrans} \reduc \ec_{2}$ and that playing $\stratExp$ is not sufficient to ensure the worst-case (because $\playerOne$ has to deal with the $-1$ edge from $\state_{8}$ to $\state_{6}$ that remains in $\edgesNonZero$).

\smallskip\noindent\textbf{Analysis of the witness-and-secure strategy.} We close our discussion of winning ECs with the formal proof of Thm.~\ref{thm:insideWeakly} through the use of the witness-and-secure strategy $\stratWNS$ (Def.~\ref{def:mp_stratWitAndFlee}).

\begin{proof}[Proof of Theorem~\ref{thm:insideWeakly}]
Assume an arbitrary $\varepsilon > 0$. Let $\stepsExp, \stepsWC \in \nat$ be such that the combined strategy $\stratComb \in \stratsPureFinite_{1}(\game_{\mpTrans} \reduc \ec)$, as defined in Def.~\ref{def:mp_stratComb}, satisfies the $\BWC$ problem for thresholds $(0,\, \optimalExp - \varepsilon)$ in $\game_{\mpTrans} \reduc \ec$. The existence of such values is guaranteed by Thm.~\ref{thm:insideWinning}. We build the finite-memory strategy $\stratWNS \in \stratsPureFinite_{1}(\game)$ according to Def.~\ref{def:mp_stratWitAndFlee} and claim it satisfies the $\BWC$ problem for thresholds $(0,\, \optimalExp - \varepsilon)$ in $\game$.

First, consider the worst-case. Let $\strat_{2} \in \strats_{2}(\game)$ be any strategy of $\playerTwo$ and $\play \in \outcomesGame{\game}{\initState}{\stratWNS}{\strat_{2}}$ be any outcome consistent with $\stratWNS$. Two cases are possible. One, $\playerTwo$ keeps choosing edges in $\edgesNonZero$ forever. That is, for $\play = \state_{0}\state_{1}\state_{2}\ldots{}$, for all $i \geq 0$ such that $\state_{i} \in \statesProb = \statesTwo$, we have that $(\state_{i}, \state_{i+1}) \in \edges_{\mpTrans}$. Then, the play is constrained to $\game_{\mpTrans} \reduc \ec$ and consistent with $\stratComb$. Hence, it follows from Thm.~\ref{thm:insideWinning} (and Lemma~\ref{lem:mp_swecWC}) that $\mpay(\play) > 0$. Two, $\playerTwo$ chooses some edges in $\edges \setminus \edgesNonZero$. That is, for $\play = \state_{0}\state_{1}\state_{2}\ldots{}$, there exists $i \geq 0$ such that $(\state_{i}, \state_{i+1}) \not\in \edges_{\mpTrans}$. Let $i_{0}$ be the smallest index where it happens. By definition of $\stratWNS$, we know that $\playerTwo$ switches to $\stratSecure$ at step $i_{0}$. Hence the suffix $\play' = \state_{i_{0}}\state_{i_{0}+1}\state_{i_{0}+2}\ldots{}$ is consistent with $\stratSecure$. Consequently, $\mpay(\play') > 0$. By prefix-independence of the mean-payoff value function, we conclude that $\mpay(\play) > 0$, which closes the case of the worst-case requirement.

Second, consider the expected value requirement. We claim that $\expect_{\initState}^{\game[\stratWNS, \stratStoch]}(\mpay) > \optimalExp - \varepsilon$. By definition, $\outcomesGame{\game}{\initState}{\stratWNS}{\stratStoch} = \outcomesMDP{\markovProcess}{\initState}{\stratWNS}$ only contains plays where $\playerTwo$ conforms to $\edgesNonZero$ at all times. Such plays never exit the EC and by Def.~\ref{def:mp_stratWitAndFlee}, we have that $\outcomesGame{\game}{\initState}{\stratWNS}{\stratStoch} = \outcomesGame{\game_{\mpTrans}}{\initState}{\stratComb}{\stratStoch}$. Also note that the probability measure of plays of those two sets is identical. Hence, by Thm.~\ref{thm:insideWinning} (and Lemma~\ref{lem:mp_swecExp}), we obtain that
$\expect_{\initState}^{\game[\stratWNS, \stratStoch]}(\mpay) = \expect_{\initState}^{\game_{\mpTrans}[\stratComb, \stratStoch]}(\mpay)  > \optimalExp - \varepsilon$.
This sets the case for the expectation and concludes our proof.
\end{proof}

\subsection{\textbf{Global strategy: favor reaching the highest valued winning end-components}}
\label{subsec:mpGlobal}

We now have all the elements needed to describe the final steps of algorithm $\mpAlgoName$ (lines~\ref{alg:mp_modifyWeights}-\ref{alg:mp_main_end}) and prove its correctness. In this section, we first describe (Def.~\ref{def:mp_weightsZeroInLosingECs}) how to modify the weights of the MDP $\markovProcess = \game[\stratStoch]$ such that a classical optimal expectation strategy in the modified MDP $\markovProcess'$ will naturally try to reach \textit{winning} ECs with the highest combined expectation. This step is a cornerstone of the \textit{global strategy} $\stratGlobal \in \stratsPureFinite_{1}(\game)$ that we define next (Def.~\ref{def:mp_globalStrategy}). This strategy is a by-product of algorithm $\mpAlgoName$.

We study the adequacy of algorithm $\mpAlgoName$ through two lemmas. In Lemma~\ref{lem:mp_algoCorrectness}, we prove its \textit{correctness}, i.e., that if it returns $\yes$, then the global strategy $\stratGlobal$ satisfies the $\BWC$ problem for the given thresholds. In Lemma~\ref{lem:mp_algoCompleteness}, we show its \textit{completeness}, i.e., that if it returns $\no$, then there exists no finite-memory strategy that satisfies the $\BWC$ problem. Combining those two lemmas and the analysis of the preprocessing conducted in Sect.~\ref{subsec:mpAssumptions}, we conclude that algorithm $\mpAlgoName$ is a valid algorithm to solve the $\BWC$ problem on any two-player game with the mean-payoff value function.

\smallskip\noindent\textbf{Modifying the MDP to naturally reach winning ECs.} Our motivation is the creation of an MDP $\markovProcess'$ such that an optimal strategy in $\markovProcess'$ maximizes the expectation without using negligible states (as defined in Sect.~\ref{subsec:mpClassification}, that is, with regard to $\game$ and $\markovProcess$). Indeed, we know by Lemma~\ref{lem:winningECsProbaOne} that winning ECs should be almost-surely eventually used in order to satisfy the worst-case requirement of the $\BWC$ problem. In particular, states in losing ECs and not in any winning sub-EC should be avoided in the long-run.

\begin{definition}
\label{def:mp_weightsZeroInLosingECs}
Given $\gameFull$, $\graphFull$ and $\markovProcess = \game[\stratStoch]$, we define $\game' = (\graph', \statesOne, \statesTwo)$, $\graph' = (\states, \edges, \weight')$ and $\markovProcess' = \game'[\stratStoch]$ by modifying the weight function as follows:
\begin{equation*}
\forall\, e = (\state_{1}, \state_{2}) \in \edges,\, \weight'(e) := \begin{cases}\weight(e) \text{ if } \exists\: \ec \in \maxWinningECs \text{ s.t. } \{\state_{1}, \state_{2}\} \subseteq \ec,\\0 \text{ otherwise.} \end{cases}
\end{equation*}
\end{definition}

Let $\stratExp \in \stratsPureMemoryless_{1}(\markovProcess')$ be a pure memoryless strategy of $\playerOne$ that maximizes the expected mean-payoff in $\markovProcess'$. Such a strategy always exists~\cite{filar1997}. Note that following $\stratExp$ does not suffice to satisfy the $\BWC$ problem in general (Rem.~\ref{rmk:mp_memorylessNotEnough}). This strategy will be part of the global strategy $\stratGlobal$ (Def.~\ref{def:mp_globalStrategy}): its role is to maximize the combined expectation of reachable winning ECs, while avoiding using negligible states, in particular states that only belong to losing ECs. The strategy to adopt inside winning ECs will be prescribed by another part of the global strategy, based on what we have established in Sect.~\ref{subsec:mpInsideWeakly}. Observe that it suffices to consider the \textit{maximal} winning ECs in order to maximize the expectation, as proved by Lemma~\ref{lem:mp_ecsMaximal}.

\begin{remark}
\label{rem:optiExpModifMDP}
Notice that $\stratExp$ is also well-defined in $\markovProcess$ and $\game$ thanks to the shared underlying graph. Also, recall that all states of $\game$ are worst-case winning due to the preprocessing. Let $\stratWC \in \stratsPureMemoryless_{1}(\game)$ be an optimal worst-case winning strategy. We observe that all states remain worst-case winning in $\game'$ for the reason that an optimal worst-case strategy only needs to visit edges involving negligible states finitely often. Indeed, either these negligible states do not belong to any EC, in which case $\playerOne$ cannot rely on them to satisfy the worst-case requirement (as he cannot ensure that he will be able to see them infinitely often), or they belong to losing ECs and no winning sub-EC, in which case $\playerOne$ only needs to visit them a finite number of times (basically to get out of the set $(\bigcup_{\ec \in \ecsSet} \ec) \setminus (\bigcup_{\ec \in \winningECs} \ec)$ if the play starts in it and reach the set $\bigcup_{\ec \in \winningECs} \ec$). Hence, the guaranteed mean-payoff is not impacted by the changes described in Def.~\ref{def:mp_weightsZeroInLosingECs}: it remains strictly positive in all states. By virtue of this, we deduce that
$\forall\, \state \in \states,\; \expect_{\state}^{\markovProcess'[\stratExp]}(\mpay) \geq \expect_{\state}^{\markovProcess'[\stratWC]}(\mpay) > 0$,
and as such, that strategy $\stratExp$ will not prescribe staying in the set $(\bigcup_{\ec \in \ecsSet} \ec) \setminus (\bigcup_{\ec \in \winningECs} \ec)$ forever. Indeed, it is always beneficial to exit it and obtain a strictly positive expectation instead of an expectation equal to zero (recall all edges involving negligible states are mapped to weight zero by Def.~\ref{def:mp_weightsZeroInLosingECs}).
\end{remark}

\smallskip\noindent\textbf{Defining a global strategy.} Based on the memoryless strategies $\stratExp$ and $\stratWC$ in $\game$ (as defined above), and the pure finite-memory witness-and-secure strategy $\stratWNS$ in winning ECs (as presented in Def.~\ref{def:mp_stratWitAndFlee} - parameters $\stepsExp$ and $\stepsWC$ may vary depending on the actual corresponding EC), we build a \textit{global strategy}~$\stratGlobal$ in $\game$ as follows. This strategy is parameterized by a natural constant $\stepsGlobal \in \nat$.
\begin{definition}
\label{def:mp_globalStrategy}
In a game $\game$, we define the {\em global strategy} $\stratGlobal \in \stratsPureFinite_{1}(\game)$ as follows.

\begin{itemize}
\item[$(a)$] Play $\stratExp \in \stratsPureMemoryless_{1}(\game)$ for $\stepsGlobal$ steps.
\item[$(b)$] Let $\state \in \states$ be the reached state.
\begin{itemize}
\item[$(b.1)$] If $\state \in \ec \in \maxWinningECs$, play the corresponding strategy $\stratWNS \in \stratsPureFinite_{1}(\game)$ forever.
\item[$(b.2)$] Else play $\stratWC \in \stratsPureMemoryless_{1}(\game)$ forever.
\end{itemize}
\end{itemize}
\end{definition}
Let us sketch this strategy. In phase $(a)$, the optimal expectation strategy in $\markovProcess'$ is followed. It will drive the outcomes towards the ECs with the highest expected values. By taking $\stepsGlobal$ large enough, we can ensure that the probability of being in an EC will be arbitrarily close to one (by Lemma~\ref{lem:ecsProbaOne}). As a result of the weights modification described in Def.~\ref{def:mp_weightsZeroInLosingECs}, we can further ensure that the probability of being inside a \textit{winning} EC will be arbitrarily close to one. Note that in phase $(b.1)$, the witness-and-secure strategy guarantees satisfaction of the worst-case requirement while yielding an expectation arbitrarily close to the optimal expectation of the EC (as proved in Sect.~\ref{subsec:mpInsideWeakly}). Also, in phase~$(b.2)$, the mean-payoff of outcomes is strictly positive, and the probability of being in $(b.2)$ can be arbitrarily close to zero for large enough values of $\stepsGlobal$. Overall, we obtain that $\stratGlobal$ satisfies the worst-case requirement because the strategies played in the two terminal phases, $(b.1)$ and $(b.2)$, all guarantee its satisfaction and the mean-payoff is prefix-independent (hence it is not impacted by phase $(a)$). Furthermore, the expectation of $\stratGlobal$ can be arbitrarily close to the maximal expectation $\optimalExp$ achievable in $\markovProcess'$ (i.e., the one achieved by~$\stratExp$) by taking sufficiently large values for the constants $\stepsExp$, $\stepsWC$ and~$\stepsGlobal$. Hence, if $\optimalExp > \thresholdExp$, $\stratGlobal$ is a proper $\BWC$ satisfying strategy for $\playerOne$.

Finally, $\optimalExp$ constitutes an upper bound to the expectation of any strategy of $\playerOne$ in $\markovProcess'$. By Lemma~\ref{lem:winningECsProbaOne}, it is also an upper bound on the expectation of any strategy that satisfies the worst-case requirement in the original game and MDP. It follows that if $\optimalExp \leq \thresholdExp$, then there exists no finite-memory strategy that satisfies the $\BWC$ problem.
As the validity of the preprocessing was shown in Sect.~\ref{subsec:mpAssumptions}, this let us conclude that algorithm $\mpAlgoName$ is both correct \textit{and} complete.

\smallskip\noindent\textbf{Illustration.} Consider the game $\game$ depicted in Fig.~\ref{fig:mpRunningExample} and the associated MDP $\markovProcess = \game[\stratStoch]$. Following Lemma~\ref{lem:ecsProbaOne}, analysis of the maximal ECs $\ec_{1}$, $\ec_{2}$ and $\ec_{3}$ reveals that the maximal expected mean-payoff achievable in $\markovProcess$ is $4$. It is for example obtained by the memoryless strategy that chooses to go to $\state_{2}$ from $\state_{1}$ and to $\state_{4}$ from $\state_{3}$. Observe that playing in $\ec_{1}$ forever is needed to achieve this expectation. By Lemma~\ref{lem:winningECsProbaOne}, this should not be allowed as the worst-case cannot be ensured if it is. Indeed, $\playerTwo$ can produce worst-case losing outcomes by playing the $-1$ edge. Clearly, the maximal expected value that $\playerOne$ can ensure while guaranteeing the worst-case requirement is thus bounded by the maximal expectation in $\markovProcess'$, i.e., by $3$. Let $\stratExp$ denote an optimal memoryless expectation strategy in $\markovProcess'$ that tries to enter~$\ec_{2}$ by playing $(\state_{1}, \state_{2})$ and $(\state_{3}, \state_{5})$, and then plays edge $(\state_{6}, \state_{8})$ forever.

Observe that algorithm $\mpAlgoName$ answers $\yes$ for any thresholds pair $(0,\, \thresholdExp)$ such that $\thresholdExp < 3$. For the sake of illustration, we construct the global strategy $\stratGlobal$ as presented in Def.~\ref{def:mp_globalStrategy} with $\stepsGlobal = 6$ and $\stepsExp = \stepsWC = 2$. For the first six steps, it behaves exactly as $\stratExp$. Note that after the six steps, the probability of being in $\ec_{2}$ is $1/4 + 1/8 = 3/8$. Then, $\stratGlobal$ switches to another strategy depending on the current state ($\stratWNS$ or $\stratWC$) and sticks to this strategy forever. Particularly, if the current state belongs to $\ec_{2}$, it switches to $\stratWNS$ as described in Def.~\ref{def:mp_stratWitAndFlee} for $\stepsExp = \stepsWC = 2$, which guarantees the worst-case threshold and induces an expectation of $13/6$ (Sect.~\ref{subsec:mpInsideWeakly}). By definition of $\stratGlobal$ on the sample game~$\game$, if the current state after six steps is not in $\ec_{2}$, then $\stratGlobal$ switches to $\stratWC$ which guarantees a mean-payoff of $1$ by reaching state $\state_{9}$ and then playing $(\state_{9}\state_{10})^{\omega}$. Overall, the expected mean-payoff of $\stratGlobal$ against $\stratStoch$ is $\expect_{\state_{1}}^{\game[\stratGlobal, \stratStoch]}(\mpay) \geq \frac{3}{8}\cdot\frac{13}{6} + \frac{5}{8}\cdot 1 = \frac{23}{16}$.
Notice that by taking $\stepsGlobal$, $\stepsExp$ and $\stepsWC$ large enough, it is possible to satisfy the $\BWC$ problem for any $\thresholdExp < 3$ with the strategy $\stratGlobal$. Also, observe that the winning EC $\ec_{2}$ is crucial to achieve expectations strictly greater than $2$, which is the upper bound when limited to EC $\ec_{3}$. For example, $\stepsGlobal = 25$ and $\stepsExp = \stepsWC = 2$ implies an expectation strictly greater than $2$ for the global strategy.

Lastly, note that in general, the maximal expectation achievable in $\markovProcess'$ (and thus in $\markovProcess$ when limited to strategies that respect the worst-case requirement) may depend on a combination of ECs instead of a unique one. This is transparent through the solving of the expected value problem in the MDP $\markovProcess'$. Hence, the approach followed by our algorithm is a way of solving a complex problem by breaking it into smaller pieces.

\smallskip\noindent\textbf{Correctness and completeness.} We start by proving the correctness of the algorithm $\mpAlgoName$ described in Alg.~\ref{alg:mp} and the soundness of the global strategy presented in Def.~\ref{def:mp_globalStrategy} to satisfy the $\BWC$ problem.

\begin{lemma}[correctness]
\label{lem:mp_algoCorrectness}
If algorithm $\mpAlgoName$ answers $\yes$, then there exist values of the parameters~$N$ and $K$ such that the global strategy $\stratGlobal \in \stratsPureFinite_{1}$ satisfies the $\BWC$ mean-payoff problem.
\end{lemma}

\begin{proof}
We assume the answer returned by $\mpAlgoName$ is $\yes$ and we prove the claim.

First, consider the worst-case requirement (Eq.~\eqref{eq:thresholdWC}). Let $\strat_{2} \in \strats_{2}$ be an arbitrary strategy of $\playerTwo$. Let~$\stepsGlobal$ take an arbitrary value in $\nat$, and for any winning EC $\ec \in \maxWinningECs$, let $\stepsExp_{\ec}$ take an arbitrary value in $\nat$ and $\stepsWC_{\ec}$ be defined according to Def.~\ref{def:mp_stepsWC} with regard to $\stepsExp_{\ec}$. Consider the outcomes consistent with $\stratGlobal$ and $\strat_{2}$. Our goal is to prove that for all outcomes $\play \in \outcomesGame{\game}{\initState}{\stratGlobal}{\strat_{2}}$, we have that $\mpay(\play) > 0$. Let $\play$ be an arbitrary outcome in this set, $\state = \last{\play(\stepsGlobal)}$ be the state reached after phase $(a)$ of the global strategy, and $\play'$ be the suffix play such that $\play = \play(\stepsGlobal)\cdot \play'$. Two cases are possible. First, assume $\state \in \ec$ for some maximal winning EC $\ec \in \maxWinningECs$. Then, $\play'$ is consistent with the witness-and-secure strategy $\stratWNS$ (as presented in Sect.~\ref{subsec:mpInsideWeakly} for initial states in $\ec$). By Thm.~\ref{thm:insideWeakly}, $\mpay(\play') > 0$. Second, assume $\state \not\in \bigcup_{\ec \in \maxWinningECs} \ec$, i.e., $\state \in \negligibleStates$. Then, $\play'$ is consistent with the worst-case winning strategy $\stratWC$ provided by the preprocessing, and we have that $\mpay(\play') > 0$. By prefix-independence of the mean-payoff value function, we conclude that in both cases, $\mpay(\play) = \mpay(\play') > 0$, proving that strategy~$\stratGlobal$ ensures the worst-case requirement.

Second, consider the expected value requirement (Eq.~\eqref{eq:thresholdExp}). We need to prove that $\expect_{\initState}^{\game[\stratGlobal, \stratStoch]}(\mpay) > \thresholdExp$ for some well-chosen values of $\stepsGlobal$ and $\stepsExp \in \nat$. Formally, the value $\stepsExp_{\ec}$ may be different in each winning EC $\ec \in\, \maxWinningECs$, so we will take a uniform value $\stepsExp$ sufficiently large to ensure that it works for all ECs. Values $\stepsWC_{\ec}(\stepsExp)$ are defined according to~Def.~\ref{def:mp_stepsWC}. Again noting that the weights encountered during phase $(a)$ of strategy $\stratGlobal$ have no impact on the mean-payoff of plays (because phase $(a)$ is of finite duration and all weights are also finite), we formulate the expectation as
\begin{equation}
\label{eq:mp_globalStratExpect}
\expect_{\initState}^{\game[\stratGlobal, \stratStoch]}(\mpay) = \sum_{\ec \in\, \maxWinningECs} \big[ p_{\stepsGlobal}(\ec) \cdot e_{\stepsExp}(\ec)\big]  + \sum_{\state \in \negligibleStates} \big[ p_{\stepsGlobal}(\state)\cdot e_{\textsc{wc}}(\state)\big],
\end{equation}
where $p_{\stepsGlobal}(\ec)$ denotes the probability to be in a state belonging to the maximal winning EC $\ec \in \maxWinningECs$ after~$\stepsGlobal$ steps of following strategy~$\stratExp$ (i.e., phase $(a)$); $e_{\stepsExp}(\ec)$ denotes the expectation of plays starting in $\ec$ and consistent with $\stratWNS$ for values $\stepsExp$ and $\stepsWC_{\ec}(\stepsExp)$ of the parameters (this expectation is identical for all initial states in the EC); $p_{\stepsGlobal}(\state)$ denotes the probability to be in a given negligible state $\state \in \negligibleStates$ (i.e., outside of winning ECs) after phase~$(a)$; and~$e_{\textsc{wc}}(s)$ denotes the expectation over plays that start in such a state $\state$ and are consistent with the worst-case strategy~$\stratWC$. Observe that $\sum_{\state \in \negligibleStates} p_{\stepsGlobal}(\state) = 1 - \sum_{\ec \in\, \maxWinningECs} p_{\stepsGlobal}(\ec)$.

Similarly, we write the expectation of the optimal expectation strategy $\stratExp$ in $\markovProcess'$ as $\expect_{\initState}^{\game'[\stratExp, \stratStoch]}(\mpay) = \sum_{\ec \in\, \maxWinningECs} \big[ p(\ec) \cdot e(\ec)\big]$, where $p(\ec)$ and $e(\ec)$ denote the probability and the expectation of maximal winning ECs when strategy $\stratExp$ is followed forever. Note that this equation depends uniquely on winning ECs by consequence of Rem.~\ref{rem:optiExpModifMDP}, and specifically maximal winning ECs by further application of Lemma~\ref{lem:mp_ecsMaximal}. In addition, observe that $\expect_{\initState}^{\game'[\stratExp, \stratStoch]}(\mpay) = \expect_{\initState}^{\game[\stratExp, \stratStoch]}(\mpay) = \optimalExp$, as the weight modification (Def.~\ref{def:mp_weightsZeroInLosingECs}) does not alter winning ECs.

We claim that $\expect_{\initState}^{\game[\stratGlobal, \stratStoch]}(\mpay)$ tends to $\optimalExp$ when $\stepsGlobal$ and $\stepsExp$ tend to infinity. We study the terms of Eq.~\eqref{eq:mp_globalStratExpect}. Note that $e_{\textsc{wc}}$ takes a bounded value (in $\left] 0,\, \largestW\right] $ by definition of $\stratWC$). By application of the analysis developed in Lemma~\ref{lem:mp_swecExp} and Theorem~\ref{thm:insideWeakly}, we have that $\forall\, \ec \in\, \maxWinningECs,\: e_{\stepsExp}(\ec) \xrightarrow[\stepsExp \rightarrow \infty]{}  e(\ec)$. Furthermore, by definition of $\stratExp$ we have that $\forall\, \ec \in\, \maxWinningECs,\: p_{\stepsGlobal}(\ec) \xrightarrow[\stepsGlobal \rightarrow \infty]{}  p(\ec)$,
and by definition of the modified weight function (Def.~\ref{def:mp_weightsZeroInLosingECs}) and Rem.~\ref{rem:optiExpModifMDP}, that $\sum_{\ec \in\, \maxWinningECs} p_{\stepsGlobal}(\ec) \xrightarrow[\stepsGlobal \rightarrow \infty]{}  \sum_{\ec \in\, \maxWinningECs} p(\ec) = 1$. Summing up, we obtain as needed that $\expect_{\initState}^{\game[\stratGlobal, \stratStoch]}(\mpay) \xrightarrow[\stepsGlobal,\, \stepsExp \rightarrow \infty]{} \optimalExp$.

By convergence, for all $\varepsilon > 0$, there exist $\stepsGlobal, \stepsExp \in \nat$ such that $\expect_{\initState}^{\game[\stratGlobal, \stratStoch]}(\mpay) \geq \optimalExp - \varepsilon$. Since algorithm $\mpAlgoName$ answered $\yes$, we have that $\optimalExp > \thresholdExp$. Hence, there exist values $\stepsGlobal, \stepsExp \in \nat$ such that $\expect_{\initState}^{\game[\stratGlobal, \stratStoch]}(\mpay) > \thresholdExp$. This concludes the proof.
\end{proof}

In order to prove that the algorithm solves the $\BWC$ problem (Def.~\ref{def:bwc_problem}) for the mean-payoff value function, we still need to establish its completeness: if the global strategy does not suffice to satisfy some thresholds pair, then no finite-memory strategy can do it.

\begin{lemma}[completeness]
\label{lem:mp_algoCompleteness}
If algorithm $\mpAlgoName$ answers $\no$, then there exists no finite-memory strategy of $\playerOne$ that satisfies the $\BWC$ mean-payoff problem.
\end{lemma}

\begin{proof}
By contradiction, assume there exists $\strat_{1}^{f} \in \stratsFinite_{1}$ that satisfies the $\BWC$ problem for thresholds $(0,\, \thresholdExp)$. We claim that algorithm $\mpAlgoName$ answers $\yes$. First, notice that the algorithm cannot answer $\no$ at line~\ref{alg:mp_losingNoB} since $\strat_{1}^{f}$ satisfies the worst-case requirement from the initial state $\initState$. Hence it remains to prove that~$\optimalExp$, as computed by the algorithm, is such that $\optimalExp > \thresholdExp$. If it is the case, the algorithm will answer $\yes$, which proves our claim.

By hypothesis, strategy $\strat_{1}^{f}$ induces an expectation $\expect_{\initState}^{\game[\strat_{1}^{f}, \stratStoch]}(\mpay) > \thresholdExp$. By Lemma~\ref{lem:winningECsProbaOne}, we have that
$\expect_{\initState}^{\game[\strat_{1}^{f}, \stratStoch]}(\mpay) = \expect_{\initState}^{\game'[\strat_{1}^{f}, \stratStoch]}(\mpay) = \expect_{\initState}^{\markovProcess'[\strat_{1}^{f}]}(\mpay)$, with $\game'$ the game obtained by the transformation defined in Def.~\ref{def:mp_weightsZeroInLosingECs}. Moreover, by definition of the optimal expectation, we have that for all $\strat_{1} \in \strats_{1}$, $\optimalExp \geq  \expect_{\initState}^{\markovProcess'[\strat_{1}]}(\mpay)$. In particular, this inequality is verified for strategy $\strat_{1}^{f}$. Hence, we obtain that $\optimalExp \geq  \expect_{\initState}^{\markovProcess'[\strat_{1}^{f}]}(\mpay) > \thresholdExp$. Consequently, the answer of the algorithm is $\yes$ and the lemma is proved.
\end{proof}

In summary, correctness and completeness of algorithm $\mpAlgoName$ as stated in Thm.~\ref{thm:mp_decisionProblem} follows from the combination of Lemma~\ref{lem:mp_algoCorrectness}, Lemma~\ref{lem:mp_algoCompleteness} and the validity of the preprocessing, as presented in Sect.~\ref{subsec:mpAssumptions}. The complexity of the algorithm is discussed in the next section (Lemma~\ref{lem:mp_algoComplexity}), as well as matching lower bounds for the $\BWC$ problem (Lemma~\ref{lem:mp_problemHardness}).

\subsection{\textbf{Complexity: algorithm and lower bound}}
\label{subsec:mpComplexity}

In this section, we prove that the BWC problem is in $\NPinter$ (Lemma~\ref{lem:mp_algoComplexity}), as algorithm $\mpAlgoName$ is polynomial in the complexity of the worst-case threshold problem, which also belongs to $\NPinter$~\cite{ZP96,jurdzinski98}. Whether the latter is in $\PTIME$ or not is a long-standing open problem~\cite{BCDGR11,Chatterjee201525}. In Lemma~\ref{lem:mp_problemHardness}, we establish that it reduces in polynomial time to the $\BWC$ problem. Given the outstanding nature of the worst-case threshold problem membership to $\PTIME$, algorithm $\mpAlgoName$ can thus be considered optimal. Furthermore, we observe that if the worst-case threshold problem were proved to be solvable in deterministic polynomial time, then algorithm $\mpAlgoName$ would also take polynomial time (Rem.~\ref{rem:mp_complexityCollapse}).

The size of the \textit{input} for algorithm $\mpAlgoName$ depends polynomially on (i) the number of states of the input game~$\vert\states^{i}\vert$, (ii)~the number of edges of the input game $\vert\edges^{i}\vert$, (iii) the number of bits of the encoding of weights $\bits^{i} = \log_{2} \largestW^{i}$, (iv) the size of the memory of the Moore machine $\vert\mooreMem\vert$, (v) the size of the supports and the length of the encoding of probabilities for the next-action function $\mooreNext$, and (vi) the encoding of the thresholds $\thresholdWC, \thresholdExp \in \rat$.

\begin{lemma}
\label{lem:mp_algoComplexity}
Algorithm $\mpAlgoName$ requires polynomial time plus a polynomially-many calls to an $\NPinter$ oracle solving the worst-case threshold problem. Hence the beyond worst-case problem for the mean-payoff is in $\NPinter$.
\end{lemma}

\begin{proof}
To prove $\NPinter$-membership, we review each step of the algorithm sequentially. Lines~\ref{alg:mp_thresholdsTest}-\ref{alg:mp_thresholds} and lines~\ref{alg:mp_losingNo}-\ref{alg:mp_mdp} are at most polynomial in the input. Line~\ref{alg:mp_winningStates} consists in solving the worst-case threshold problem on the input game $\game^{i}$: this can be done by calling an $\NPinter$ oracle~\cite{ZP96,jurdzinski98}. Overall, the preprocessing is in $\NPinter$ and produces a game $\game$ such that $\vert\game\vert \leq \vert\game^{i}\vert \cdot \vert\mooreMachine{\strat_{2}^{i}}\vert$, using the natural definitions of those sizes as polynomial functions of the values described in points (i) to (vi).

For the main algorithm, the complexities are as follows. Line~\ref{alg:mp_main} is the call to sub-algorithm $\mwecAlgoName(\markovProcessNonZero)$, which has been proved to only make a polynomial number of calls to an $\NPinter$ oracle in Lemma~\ref{lem:mpECsPartition}. Note that the size of $\markovProcess = \game[\stratStoch]$ is polynomial in the input. The weights modification (line~\ref{alg:mp_modifyWeights}) requires linear time (polynomial in the input game) as do lines~\ref{alg:mp_maxExpComp}-\ref{alg:mp_main_end}. Finally, computing the maximal expected value on $\markovProcess'$ (line~\ref{alg:mp_maxExp}) is polynomial in $\vert\markovProcess'\vert$ via linear programming~\cite{filar1997}, hence polynomial in the input size.

In conclusion, the algorithm executes a polynomial number of operations, and each of them belongs to $\PTIME$ or is a call to an $\NPinter$ oracle. Therefore, the $\BWC$ problem for the mean-payoff belongs to $\PTIME^{\NPinter}$. By~\cite{Bra79}, we know that $\PTIME^{\NPinter} = \NPinter$, which proves the claim.
\end{proof}

\begin{remark}
\label{rem:mp_complexityCollapse}
Assume an algorithm $\textsc{Ptime\_wc}$ is established to solve the worst-case threshold problem in deterministic polynomial time. Then, the complexities of algorithm $\mpAlgoName$ and sub-algorithm $\mwecAlgoName$ boil down to a polynomial number of polynomial time operations and external calls, and it follows that the BWC mean-payoff problem is in $P$.
\end{remark}

Reduction of the worst-case threshold problem to the $\BWC$ one seems natural by Eq.~\eqref{eq:thresholdWC}. Still, we need to pay attention to the strict inequality in the $\BWC$ problem definition: we use the existence of memoryless winning strategies for the worst-case problem and careful analysis of the domain of the mean-payoff values of outcomes to prove that it is not restrictive. The expected value part of the $\BWC$ problem can be defined arbitrarily under certain conditions. 
Our complexity results are summed up in Thm.~\ref{thm:mp_decisionProblem}.

\begin{lemma}
\label{lem:mp_problemHardness}
The worst-case threshold problem on mean-payoff games reduces in polynomial time to the $\BWC$ mean-payoff problem.
\end{lemma}

\begin{proof}
Given a game $\gameFull$, its underlying graph $\graphFull$, an initial state $\initState \in \states$, and the worst-case threshold $\thresholdWC = 0$ (without loss of generality), the worst-case threshold problem asks if the following proposition is true:
\begin{equation}
\label{eq:mp_wcReducProof}
\exists\, \strat_{1} \in \strats_{1},\,\forall\, \strat_{2} \in \strats_{2},\, \forall\, \play \in \outcomesGame{\game}{\initState}{\strat_{1}}{\strat_{2}},\, \mpay(\play) \geq 0.
\end{equation}
By the results of~\cite{liggett_SR69,EM79}, it is equivalent to restrict both players to pure memoryless strategies:
\begin{equation}
\label{eq:mp_wcReducProofB}
\exists\, \strat^{pm}_{1} \in \stratsPureMemoryless_{1},\,\forall\, \strat^{pm}_{2} \in \stratsPureMemoryless_{2},\, \forall\, \play \in \outcomesGame{\game}{\initState}{\strat^{pm}_{1}}{\strat^{pm}_{2}},\, \mpay(\play) \geq 0.
\end{equation}

It is well-known that in this context, $\mpay(\play) \geq 0 \Leftrightarrow \mpay(\play) > -\frac{1}{\vert\states\vert}$. Indeed, consider the following argument. First, the mean-payoff of any outcome $\play \in \outcomesGame{\game}{\initState}{\strat^{pm}_{1}}{\strat^{pm}_{2}}$ can be trivially bounded by $-\largestW \leq \mpay(\play) \leq \largestW$, with $\largestW$ the largest absolute value of any weight assigned by $\weight$ to edges of $\game$. Second, consider the decomposition of $\play$ into simple cycles (i.e., cycles with no repeated state except for the starting and ending state). Since weights are integers, any simple cycle has an associated mean-payoff belonging to $\{-\largestW, \ldots{}, -\frac{1}{\vert\states\vert}, 0, \frac{1}{\vert\states\vert}, \ldots{}, \largestW\}$. As both strategies are memoryless, any outcome $\play \in \outcomesGame{\game}{\initState}{\strat^{pm}_{1}}{\strat^{pm}_{2}}$ will ultimately consist in a repeated simple cycle. Hence we have that $\mpay(\play) \in \{-\largestW, \ldots{}, -\frac{1}{\vert\states\vert}, 0, \frac{1}{\vert\states\vert}, \ldots{}, \largestW\}$ and we observe that no value can be taken between $-\frac{1}{\vert\states\vert}$ and $0$.

Consequently, Eq.~\eqref{eq:mp_wcReducProofB} is equivalent to
\begin{equation}
\label{eq:mp_wcReducProofC}
\exists\, \strat^{pm}_{1} \in \stratsPureMemoryless_{1},\,\forall\, \strat^{pm}_{2} \in \stratsPureMemoryless_{2},\, \forall\, \play \in \outcomesGame{\game}{\initState}{\strat^{pm}_{1}}{\strat^{pm}_{2}},\, \mpay(\play) > -\frac{1}{\vert\states\vert}.
\end{equation}

To formulate Eq.~\eqref{eq:mp_wcReducProofC} in terms of a $\BWC$ problem, we have to define an expected value threshold $\thresholdExp \in \rat$ and a stochastic model $\stratStoch \in \stratsFinite_{2}$. Since all plays $\play \in \plays{\graph}$ satisfy $\mpay(\play) \geq -\largestW$ by definition of the weight function, we trivially have that $\forall\, \strat_{1} \in \strats_{1},\, \forall\, \stratStoch \in \stratsFinite_{2},\, \expect_{\initState}^{\game[\strat_{1}, \stratStoch]}(\mpay) \geq -\largestW$.
Hence, it suffices to fix an arbitrary stochastic model $\stratStoch \in \stratsFinite_{2}$ and an arbitrary expectation threshold $\thresholdExp < -\largestW$ to obtain that Eq.~\eqref{eq:mp_wcReducProof} is satisfied \textit{if and only if} $\playerOne$ has a strategy to satisfy the $\BWC$ problem for thresholds $(-\frac{1}{\vert\states\vert},\, \thresholdExp)$ against the stochastic model $\stratStoch$.
Notice this reduction is polynomial because we can choose a simple stochastic model (e.g., a memoryless strategy requires a Moore machine of size linear in the size of the game) and a value of $\thresholdExp$ that will not require a super-polynomial growth of the encodings (e.g., $\thresholdExp = -\largestW - 1$).
\end{proof}

\newpage
\subsection{\textbf{Memory requirements}}
\label{subsec:mpMemoryRequirements}

\begin{wrapfigure}{r}{75mm}
\vspace{-3mm}
  \centering   
  \scalebox{0.7}{\begin{tikzpicture}[->,>=latex,shorten >=1pt,auto,node
    distance=2.5cm,bend angle=45,font=\Large]
    \tikzstyle{p1}=[draw,circle,text centered,minimum size=10mm]
    \tikzstyle{p2}=[draw,rectangle,text centered,minimum size=10mm]
    \tikzstyle{empty}=[]
    \node[p1] (1) at (0,0) {$\state_{2}$};
    \node[p1] (9) at (-4,0) {$\state_{1}$};
    \node[p2] (10) at (-8,0) {$\state_{3}$};
    \node[empty] (a) at (-7.5, 0.9) {$\frac{1}{2}$};
    \node[empty] (b) at (-7.5, -0.9) {$\frac{1}{2}$};
    \coordinate[shift={(0mm,5mm)}] (init) at (9.north);
    \path
    (init) edge (9)
    ;
	\draw[->,>=latex] (1) to[out=140,in=40] node[above] {$1$} (9);
	\draw[->,>=latex] (9) to[out=320,in=220] node[below] {$1$} (1);
	\draw[->,>=latex] (9) to[out=180,in=0] node[below] {$0$} (10);
	\draw[->,>=latex] (10) to[out=40,in=140] node[above] {$-X$} (9);
	\draw[->,>=latex] (10) to[out=320,in=220] node[below] {$X+5$} (9);
      \end{tikzpicture}}
      \caption{Family of games $(\game(X))_{X \in \natStrict}$ requiring polynomial memory in $\largestW = X + 5$ to satisfy the $\BWC$ problem for thresholds $(0,\, \thresholdExp \in \left]  1,\, 5/4\right[ )$.}
\label{fig:mp_familyExpInLargestW}
\vspace{-3mm}
  \end{wrapfigure}

Across the previous sections, we have studied the complexity of deciding the $\BWC$ problem, i.e., deciding the existence of a \textit{finite-memory} strategy of $\playerOne$ satisfying Def.~\ref{def:bwc_problem} for the mean-payoff value function. Now, we focus on the \textit{size of the memory} used by such a strategy. In Thm.~\ref{thm:mp_memoryRequirements}, we give an upper bound for the memory of the global strategy described in Def.~\ref{def:mp_globalStrategy} (which has been shown to suffice if satisfaction of the $\BWC$ problem is possible). This is obtained through careful analysis of the structure of involved strategies (global, witness-and-secure, combined). All of them are based on alternation between well-chosen pure memoryless strategies, based on parameters $\stepsGlobal$, $\stepsExp$ and $\stepsWC \in \nat$. We prove that these values only need to be polynomial in the size of the game and the stochastic model, and in the values of weights and thresholds.

Furthermore, we prove this upper bound to be tight in the sense that polynomial memory in the values of weights is needed in general. To establish this result, we exhibit a family of games $(\game(X))_{X \in \natStrict}$, presented in Fig.~\ref{fig:mp_familyExpInLargestW}, which requires polynomial memory in the largest absolute weight.

\begin{theorem}
\label{thm:mp_memoryRequirements}
Memory of pseudo-polynomial size may be necessary and is always sufficient to
satisfy the $\BWC$ problem for the mean-payoff: polynomial in the size of the game and the stochastic model, and polynomial in the weight and threshold values.
\end{theorem}

\begin{proof}
We first consider the upper bound on memory, derived by analysis of the global strategy $\stratGlobal$ (Def.~\ref{def:mp_globalStrategy}). Observe that it follows the memoryless strategy $\stratExp$ for $\stepsGlobal$ steps before switching to phase \textit{(b)}. The correctness of the strategy (Lemma~\ref{lem:mp_algoCorrectness}) relies on the existence of a value $\stepsGlobal$ such that the probability of being in an EC after $\stepsGlobal$ steps is \textit{high enough}. We argue that $\stepsGlobal$ does not need to be exponentially large. 

Consider the probability to be outside of winning ECs after $\stepsGlobal$ steps. By classical results on MCs, this probability decreases exponentially fast when $\stepsGlobal$ grows. Indeed, to prove it, it suffices to consider the chain $\game[\stratExp, \stratStoch]$, replace BSCCs by absorbing states and observe that the probability of absorption tends towards one exponentially fast~\cite{grinstead_AMS1997}. Now, consider the expectation of the global strategy for given constants~$\stepsGlobal$ and~$\stepsExp$, as given in Eq.~\eqref{eq:mp_globalStratExpect}. Let $\thresholdExp < \optimalExp - \varepsilon$ be the expected value threshold considered in the $\BWC$ problem (as before, we assume $\thresholdWC < \thresholdExp < \optimalExp$ otherwise the problem is trivial). Assume that $\stepsExp$ is sufficiently large to have $\sum_{\ec \in \maxWinningECs} p(\ec) \cdot e_{\stepsExp}(\ec) > \optimalExp - \varepsilon'$, with $\varepsilon' < \varepsilon$. We want to establish how large $\stepsGlobal$ needs to be to ensure an overall expectation strictly greater than $\thresholdExp$. Since~$e_{\textsc{wc}}(s)$ can be trivially lower bounded by zero for all $\state \in \negligibleStates$, it is clear that to obtain the needed property, we need to have values $p_{\stepsGlobal}(\ec)$ growing polynomially with $\varepsilon$ and~$\optimalExp$. As the growth of $p_{\stepsGlobal}(\ec)$ is exponential in the growth of the value $\stepsGlobal$, we obtain that a logarithmic value of~$\stepsGlobal$, hence \textit{polynomial in the encoding}, suffices to achieve the desired expected value.

Similarly, we study the strategies followed in phase \textit{(b)} of the global strategy (Def.~\ref{def:mp_globalStrategy}). The case \textit{(b.2)} is the easiest: the worst-case strategy $\stratWC$ is memoryless. In case \textit{(b.1)}, the witness-and-secure strategy $\stratWNS$ is used. By Def.~\ref{def:mp_stratWitAndFlee}, this strategy needs polynomial memory to witness the use of edges in $\edges \setminus \edgesNonZero$ and to implement the memoryless secure strategy $\stratSecure$. It also needs to implement the combined strategy $\stratComb$, based on alternation between memoryless strategies. The size of the memory of $\stratComb$ is polynomial in~$\stepsExp$,~$\stepsWC$ and the largest absolute value taken by $\cmbSum$, as well as in the size of the game. The proof of Lemma~\ref{lem:mp_swecExpDecrease} guarantees that for constant $\stepsExp$, a value polynomial in the size of the input game and the stochastic model, as well as in the values of weights and thresholds, suffices. By Def.~\ref{def:mp_stepsWC}, an identical situation is verified for~$\stepsWC$. Finally, the running sum $\cmbSum$ takes values in $\{-\stepsExp \cdot \largestW, \ldots{}, \stepsExp \cdot \largestW\}$, hence it also verifies such bounds. Overall, the memory needed by the combined strategy is polynomial in the size of the input game and the stochastic model, and in the weight and threshold values.

Aggregating all these bounds, we conclude that the global strategy also requires memory at most polynomial in the size of the game and the stochastic model, and in the values, thus proving the upper bound.

It remains to show that pseudo-polynomial memory is really necessary in general. In order to achieve this, we introduce a family of games, $(\game(X))_{X \in \natStrict}$, such that winning the $\BWC$ problem on $\game(X)$ requires memory polynomial in the largest weight $\largestW = X + 5$. This family is presented in Fig.~\ref{fig:mp_familyExpInLargestW}. Let the worst-case threshold be $\thresholdWC = 0$ and the expectation threshold be an arbitrary value $\thresholdExp \in \left]  1,\, 5/4\right[$. Thanks to Thm.~\ref{thm:insideWinning}, the $\BWC$ problem is satisfiable, because~$\game(X)$ is reduced to a winning EC with no edge of probability zero, and the optimal expectation is $5/4 > \thresholdExp$ (expectation achieved by the memoryless strategy that always chooses edge $(\state_{1}, \state_{3})$). Notice that it cannot be achieved by the memoryless strategy that always chooses edge $(\state_{1}, \state_{2})$ since this strategy induces a mean-payoff equal to $1 < \thresholdExp$. Hence it is mandatory to choose $(\state_{1}, \state_{3})$ infinitely often in order to achieve the expected value requirement (Eq.~\eqref{eq:thresholdExp}).

Let $\strat_{1} \in \stratsFinite_{1}$ be a finite-memory strategy of $\playerOne$ that satisfies the $\BWC$ problem. Observe that it may as well be pure, i.e., $\strat_{1} \in \stratsPureFinite_{1}$ as choosing edge $(\state_{1}, \state_{3})$ with a non-zero probability recurrently yields consistent outcomes that do not satisfy the worst-case requirement (Eq.~\eqref{eq:thresholdWC}). Also observe that anytime edge $(\state_{1}, \state_{3})$ is chosen, there is a probability~$1/2$ that the edge of weight $-X$ is taken to come back. Hence, from some point on, every appearence of this edge of weight $-X$ must be \textit{eventually} counteracted in order to preserve the worst-case requirement. A finite number of non-compensated occurences is not a problem thanks to the prefix-independence of the mean-payoff value function. Looking at the involved weights, it is clear that taking the edge $(\state_{1}, \state_{2})$ for $(\left\lfloor X/2\right\rfloor + 1)$ times is necessary to counteract the negative edge of weight $-X$. Hence, memory polynomial in $X$ (hence in $\largestW$) is needed to ensure both the worst-case and the expected value requirements for the given thresholds. This concludes our proof.
\end{proof}

\subsection{\textbf{Infinite-memory strategies}}
\label{subsec:mpInfiniteMemory}

We close our study of the $\BWC$ problem for the mean-payoff value function by considering what happens when $\playerOne$ is allowed to use infinite-memory strategies. Specifically, we show that in this context, infinite-memory strategies are in general strictly more powerful than finite-memory ones: they can exploit losing ECs to benefit from their possibly higher optimal expected value; and even inside a single winning EC, they can be optimal with regard to the expectation whereas finite-memory ones are limited to $\varepsilon$-optimality. Nonetheless, as discussed in the introduction, such strategies are ill-suited for the synthesis of implementable controllers for real-world applications, hence our focus on finite memory in the previous results.

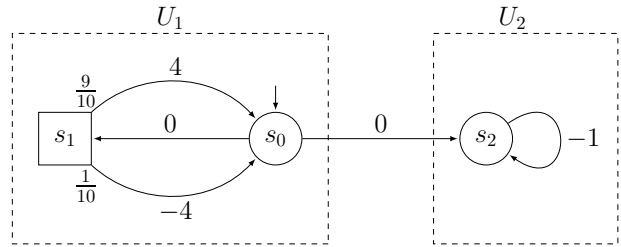
\begin{wrapfigure}{r}{84mm}
  \centering   
  \scalebox{0.7}{\begin{tikzpicture}[->,>=latex,shorten >=1pt,auto,node
    distance=2.5cm,bend angle=45,font=\Large]
    \tikzstyle{p1}=[draw,circle,text centered,minimum size=10mm]
    \tikzstyle{p2}=[draw,rectangle,text centered,minimum size=10mm]
    \tikzstyle{empty}=[]
    \node[p1] (0) at (0,0) {$\state_{0}$};
    \node[p2] (1) at (-4,0) {$\state_{1}$};
    \node[p1] (2) at (4,0) {$\state_{2}$};

    \node[empty] (proba1a) at (-3.6, 0.9) {$\frac{9}{10}$};
    \node[empty] (proba1b) at (-3.6, -0.9) {$\frac{1}{10}$};

    \coordinate[shift={(0mm,5mm)}] (init) at (0.north);
    \path
    (init) edge (0)
    (2) edge [loop right, out=35, in=325,looseness=2, distance=16mm] node [right] {$-1$} (2)
    ;

    \path (0) edge node[above] {$0$} (1) ;
    \path (1) edge [bend right] node[below] {$-4$} (0) ;
    \path (1) edge [bend left] node[above] {$4$} (0) ;
    \path (0) edge node[above] {$0$} (2) ;
 
    \draw[dashed,-] (1,-2) -- (-5,-2) -- (-5,2) -- (1,2) -- (1,-2);
    \draw[dashed,-] (3,2) -- (6.5,2) -- (6.5,-2) -- (3,-2) -- (3,2);

    \node[empty] (lec) at (-2, 2.3) {$\ec_{1}$};
    \node[empty] (swec) at (4.5, 2.3) {$\ec_{2}$};
      \end{tikzpicture}}
      \caption{Infinite-memory strategies may use losing ECs forever with a non-zero probability in order to increase the expected value.}
\label{fig:mp_infiniteMemory}
  \vspace{-2mm}
  \end{wrapfigure}

\smallskip\noindent\textbf{Losing end-components may still be useful.} Let us consider the game depicted in Fig.~\ref{fig:mp_infiniteMemory}, together with 
a memoryless stochastic model of the adversary $\stratStoch \in \stratsMemoryless_{2}$, modeled by the probabilities $1/10$ and $9/10$ on the edges leaving $\state_{1}$. The MDP $\markovProcess = \game[\stratStoch]$ can be decomposed into two
end-components $\ec_{1}$ and $\ec_{2}$, as depicted by the dashed lines. Assume the worst-case threshold is $\thresholdWC = -3/2$ (notice for once we take it different than zero), then $\ec_{1}$ is losing (because~$\playerTwo$ can induce an outcome of mean-payoff value $-4/2 \leq -3/2$, and he can do that by choosing edges in $\edgesNonZero$) and $\ec_{2}$ is winning (as the only outcome yields mean-payoff $-1 > -3/2$), following~Def.~\ref{def:classificationECs}. 
As shown in Lemma~\ref{lem:winningECsProbaOne}, any finite-memory strategy of $\playerOne$ which ensures a mean-payoff strictly greater than $\thresholdWC$, leaves $\ec_{1}$ with probability one against $\stratStoch$, because states of $\ec_{1}$ are classified as negligible. Therefore, in order to satisfy the worst-case requirement of the $\BWC$ problem, the
expected mean-payoff of any finite-memory strategy of~$\playerOne$ is $\optimalExp_{2} = -1$, i.e., the expectation obtained in $\ec_{2}$ by the only possible outcome. Notice however that if we forget about the worst-case requirement, the maximal expectation that $\playerOne$ could achieve in $\ec_{1}$ is $\optimalExp_{1} = \frac{1}{2}\cdot (\frac{9}{10}\cdot 4 + \frac{1}{10} \cdot (-4)) = \frac{8}{5}$.

We now show that $\playerOne$ can ensure the worst-case requirement and obtain
an expected value strictly greater than $-1$, if he is allowed to use infinite memory. We define a pure infinite-memory strategy $\strat_{1}$ for $\playerOne$
as follows: $\strat_{1}$ stores the running sum along the prefix played so far, and chooses to move from $\state_0$ to $\state_1$ as long as this sum is strictly greater than zero (except in the first round where it moves directly to $\state_1$). First, notice that this strategy trivially guarantees a mean-payoff greater than or equal to $-1 > \thresholdWC$. Indeed, either the running sum always stays strictly positive, implying that the mean-payoff is at least zero, or the running sum gets negative or null at some point, in which case the strategy switches to $\ec_{2}$ and the mean-payoff takes value $-1$. Second, let us compute the expected mean-payoff of this strategy against $\stratStoch$. Let $p_{{\sf switch}}$ denote the probability to switch to $\ec_{2}$ along a play, as prescribed by the strategy $\strat_{1}$. By definition, it is equal to the probability, when playing inside the EC $\ec_{1}$ and starting from an initial credit of zero in state $\state_{0}$, to come back to $\state_0$ with a credit less than or equal to zero after an arbitrary number of steps. Formally, let $\markovChain$ be the MC induced by the subgame $G\reduc \ec_{1}$ and $\stratStoch$ (note that $\playerOne$ has no choice in it). We have that
$p_{{\sf switch}} = \proba^{M}_{\state_{0}} \big( \{\play \in \outcomesMC{M}{s_0} \mid \exists\, i > 0,\, \last{\play(i)} = \state_{0} \wedge \tpay(\play(i)) \leq 0\}\big)$.

Determining the probability $p_{{\sf switch}}$ of the running sum hitting zero is equivalent to a well-studied problem on MCs: the \textit{gambler's ruin problem}~\cite{grinstead_AMS1997}. Applying results on this problem, we obtain that, if the probabilities $9/10$ and $1/10$ are respectively replaced by arbitrary probabilities $p$ and $q$ such that $p~>~q$, the probability that the gambler is eventually ruined is $\frac{q}{p}$. In our example, this implies that $p_{{\sf switch}} = 1/9$.

We are now able to use this result to provide a lower bound for the expected value of the strategy $\strat_{1}$. Consider the set of outcomes $\outcomesGame{\game}{\state_{0}}{\strat_{1}}{\stratStoch}$: it can be partitioned into the set of plays that stay in $\ec_{1}$, for which the mean-payoff is trivially bounded by zero as discussed before; and the set of plays that reach $\ec_{2}$, for which the mean-payoff is equal to $-1$. Hence, the overall expectation respects the following inequality:
$\expect_{\state_{0}}^{\game[\strat_{1}, \stratStoch]}(\mpay) \geq p_{{\sf switch}}\cdot (-1) + (1 - p_{{\sf switch}})\cdot 0 = -\frac{1}{9} > -1$.
Strategy $\strat_{1}$ yields an expectation at least equal to $-1/9$, hence strictly greater than the expectation achievable by any finite-memory strategy satisfying the $\BWC$ problem, which we have shown to be equal to $-1$.

Intuitively, the added power of infinite memory comes from the possibility to memorize an unbounded running sum, whereas finite memory implies an upper bound on such a sum. In the first case, $\playerOne$ will be able to properly acknowledge that some plays see their running sum diverging without ever dropping to zero (the set of such plays has a strictly positive probability in our example), which lets him benefit from the added expectation without endangering the worst-case requirement. In the second case, $\playerOne$ sees all running sums as upper bounded by some value~$X \in \nat$ due to its limited memory. As such, when he sees a sequence of weights whose total sum is~$-X$, an event that occurs almost-surely infinitely often when an outcome $\play$ is such that $\infVisited{\play} = \ec$ for some EC $\ec \in \losingECs = \ecsSet \setminus \winningECs$, $\playerOne$ will \textit{believe} its running sum hits zero, \textit{whether it really does or not}. Consequently, he has to leave $\ec_{1}$ eventually to ensure the worst-case requirement.

\smallskip\noindent\textbf{Optimal expected values can be reached in winning end-components.} Consider a setting satisfying Assumption~\ref{assump:uniqueWEC}: a game $\game$ reduced to a winning EC such that $\edgesNonZero = \edges$. Let $\thresholdWC = 0$, as usual. In Sect.~\ref{subsec:mpInsideStrongly}, we have seen that, for all $\varepsilon>0$, it is possible to combine a worst-case strategy $\stratWC$ with an optimal expectation strategy~$\stratExp$ into a finite-memory strategy $\stratComb$ that ensures satisfaction of the $\BWC$ problem for thresholds $(0,\, \optimalExp - \varepsilon)$, where $\optimalExp$ is the maximal expected value in $P = \game[\stratStoch]$.

In general, it is not possible to construct a finite-memory strategy $\stratComb$ that ensures the worst-case while inducing an expected value exactly equal to $\optimalExp$ against the stochastic model. See for example the game $\game \reduc \ec_{3}$ in Fig.~\ref{fig:mpRunningExample}: clearly, $\playerOne$ has to use $(\state_{10}, \state_{9})$ infinitely often to ensure the worst-case, and when using finite memory, the contribution (in terms of proportion of cycles played) of the corresponding cycle in the overall expectation can be lower bounded by a strictly positive probability, hence inducing an expected value strictly lower than $\optimalExp = 2$.

Nonetheless, it is possible to build an infinite-memory strategy,
$\strat_{1}^{inf}$, that exactly achieves this expectation while verifying the worst-case threshold. It is in essence similar to the finite-memory combined strategy (Def.~\ref{def:mp_stratComb}). Consider the following argument. Observe that in the analysis of the combined strategy (Sect.~\ref{subsec:mpInsideStrongly}), we show that when parameters $\stepsExp$ and $\stepsWC(\stepsExp)$ tend to infinity, the expectation induced by $\stratComb$ tends to $\optimalExp$. Moreover, the worst-case is always ensured by choice of $\stepsWC(\stepsExp)$. Hence, we possess all the elements needed to construct $\strat_{1}^{inf}$: it suffices to implement a strategy that plays as $\stratComb$, but sequentially increasing the values of $\stepsExp$ and $\stepsWC(\stepsExp)$ up to infinity. Formally, let $(K_i)_{i\in \nat}$ be a strictly increasing sequence of naturals, and
for all $i\geq 0$, let $\stepsWC(K_i)$ be the natural given by Def.~\ref{def:mp_stepsWC}.

The strategy $\strat_{1}^{inf}$ is defined as follows: 
\begin{itemize}
\item[$(init)$] Initialize $i$ to $0$.
\item[$\typeA$] Play $\stratExp$ for $K_i$ steps and memorize $\cmbSum \in \integ$, the sum of weights encountered during these $K_i$ steps.
\item[$\typeB$] If $\cmbSum > 0$, then go to $\typeA$.

Else, play $\stratWC$ during $\stepsWC(K_i)$ steps, then increment $i$ and go to $\typeA$.
\end{itemize}
By doing so, it is possible to show that $\strat_{1}^{inf}$ ensures an expected mean-payoff exactly equal to $\optimalExp$, as well as the worst-case requirement. For the worst-case, it suffices to apply the reasoning developed in Lemma~\ref{lem:mp_swecWC}. To show that~$\strat_{1}^{inf}$ achieves the expected value $\optimalExp$, the intuitive argument is that the probability that a period of type $\typeA$ is followed by a period of type $\typeB$ tends to zero as 
$K_i$ grows, since $\optimalExp > 0$. Therefore, the probability that $\stratWC$
is played infinitely many times is zero.
Indeed, consider the
example of Fig.~\ref{fig:mpRunningExample}. By playing
$\strat_{1}^{inf}$ as above, $\playerOne$ can ensure the worst-case requirement and induce the optimal expected mean-payoff $2$, because the proportion of time spent following strategy $\stratExp$ will tend to one as the parameter $K_{i}$ tends to infinity. 

\smallskip\noindent\textbf{Solving the $\BWC$ problem for infinite-memory strategies.} The problem was recently studied by Clemente and Raskin in~\cite{CR15}, where they showed that it also belongs in $\NPinter$.

\section{Truncated Sum Value Function - Shortest Path Problem}
\label{sec:shortest_path}

Consider a game $\gameFull$ with an underlying graph $\graphFull$ such that the weight function $\weight\colon \edges \rightarrow \natStrict$ assigns \textit{strictly positive} integer weights to all edges, and a target set of states $\truncatedTarget \subseteq \states$ that $\playerOne$ wants to reach with a path of bounded value. That is, $\playerOne$ aims to ensure some threshold on the \textit{truncated sum} value function $\truncatedSum{\truncatedTarget}$. In other words, we study the $\BWC$ synthesis problem for the \textit{shortest path problem}~\cite{bertsekas_MOR1991,deAlfaro_CONCUR1999}. More precisely, given an initial state $\initState \in \states$, a worst-case threshold $\thresholdWC \in \nat$ (we assume a natural threshold w.l.o.g.~as all weights take positive integer values and so does the truncated sum function for any play reaching~$\truncatedTarget$), an expected value threshold $\thresholdExp \in \rat$ and a finite-memory stochastic model $\stratStoch$ for the adversary, the problem is to decide if $\playerOne$ has a finite-memory strategy $\strat_{1} \in \stratsFinite_{1}$ such that
\begin{equation*}
\begin{cases}
      \forall\, \strat_{2} \in \strats_{2},\, \forall\, \play \in \outcomesGame{\game}{\initState}{\strat_{1}}{\strat_{2}},\, \truncatedSum{\truncatedTarget}(\play) < \thresholdWC\\
      \expect_{\initState}^{\game[\strat_{1}, \stratStoch]}(\truncatedSum{\truncatedTarget}) < \thresholdExp
    \end{cases}
\end{equation*}
Hence, regarding Def.~\ref{def:bwc_problem}, the inequalities are reversed. We assume that $\thresholdExp < \thresholdWC$. Equivalently, the problem could be stated with value function $-\truncatedSum{\truncatedTarget}$ without changing the definition.

We establish the following results. First, satisfaction of the $\BWC$ problem for the truncated sum value function can be decided in pseudo-polynomial time (Sect.~\ref{subsec:ts_pseudoPolyAlgo}). Second, pseudo-polynomial memory may be necessary to satisfy the $\BWC$ problem and that it always suffices (Sect.~\ref{subsec:ts_memory}). Third, the decision problem cannot be solved in polynomial time unless $\PTIME = \NPTIME$ by providing an $\NPTIME$-hardness result (Sect.~\ref{subsec:ts_hardness}).

\subsection{\textbf{A pseudo-polynomial time algorithm}}
\label{subsec:ts_pseudoPolyAlgo}

To solve the decision problem, we proceed as follows. First, we show how to construct, from the original game $\game$ and the worst-case threshold $\thresholdWC$, a new game $\gameTS$ such that there is a one-to-one correspondence between the strategies of~$\playerOne$ in $\gameTS$ and the strategies of $\playerOne$ in the original game $\game$ that are winning for the worst-case requirement (Eq.~\eqref{eq:thresholdWC}). To construct this game, we unfold the original graph $\graph$, tracking the current value of the truncated sum \textit{up to the worst-case threshold $\thresholdWC$}, and integrating this value in the states of an expanded graph $\graph'$. In the corresponding game $\game'$, we then compute the set of states $R$ from which $\playerOne$ can reach the target set with cost lower than the worst-case threshold and we define the subgame $\gameTS = \game' \reduc R$ such that any path in the graph of $\gameTS$ satisfies the worst-case requirement. Second, assuming that $\gameTS$ is not empty, we can now combine it with the the stochastic Moore machine $\mooreMachine{\stratStoch}$ of the adversary to construct an MDP in which we search for a $\playerOne$ strategy that ensures reachability of $\truncatedTarget$ with an expected cost strictly lower than $\thresholdExp$ (Eq.~\eqref{eq:thresholdExp}). If such a strategy exists, it is guaranteed that it will also satisfy the worst-case requirement against any strategy of $\playerTwo$ thanks to the bijection evoked earlier.

Hence, for the shortest path, our approach is sequential, first solving the worst-case, then optimizing the expected value among the worst-case winning strategies. Observe that this approach is not applicable to the mean-payoff, as in that case there exists no obvious finite representation of the worst-case winning strategies.

 \begin{theorem}
 \label{thm:ts_pseudoPoly}
The beyond worst-case problem for the shortest path can be solved in pseudo-polynomial time: polynomial in the size of the game graph, the Moore machine for the stochastic model of the adversary and the encoding of the expected value threshold, and polynomial in the value of the worst-case threshold.
 \end{theorem}
 \begin{proof}
Let $\gameFull$ be the two-player game, $\graphFull$ its underlying graph, $\weight\colon \edges \rightarrow \natStrict$ its weight function, $\initState \in \states$ the initial state, $\truncatedTarget \subseteq \states$ the target set, $\stratStoch \in \stratsFinite_{2}$ the stochastic model of $\playerTwo$, with $\mooreMachineFull{\stratStoch}$ its Moore machine, $\thresholdWC \in \nat$ the worst-case threshold, and $\thresholdExp \in \rat$ the expected value threshold.
Based on $\game$ and $\thresholdWC$, we define the game $\game' = (\graph', \statesOne', \statesTwo')$. Its underlying graph $\graph' = (\states', \edges', \weight')$ is built by unfolding the original graph $\graph$, tracking the current value of the truncated sum \textit{up to the worst-case threshold $\thresholdWC$}, and integrating this value in the states of $\graph'$. Formally, we have that
  \begin{itemize}
\item $\statesOne'=\statesOne \times \left( \{0,1,\dots,\thresholdWC-1\} \cup \{\tsFailSymbol\}\right) $, $\statesTwo'=\statesTwo \times \left( \{0,1,\dots,\thresholdWC-1\} \cup \{\tsFailSymbol\}\right) $, and $\states' = \statesOne' \cup \statesTwo'$;
	\item $\edges'= \left\lbrace \big((\state_1, u_1),(\state_2,u_2)\big) \in \states' \times \states' \mid (\state_1, \state_2) \in \edges \,\land\, u_2 = u_1 + \weight((\state_1, \state_2)) \right\rbrace$, with the convention that, for all $c \in \nat$, $\tsFailSymbol + c = \tsFailSymbol$, and, for all $u \in \nat$, $u + c = \tsFailSymbol$ if $u + c \geq \thresholdWC$;
	\item $\forall\, e = \big((\state_1, u_1),(\state_2,u_2)\big) \in \edges',\, \weight'(e) = \weight((s_1,s_2))$.
  \end{itemize}
The symbol $\tsFailSymbol$ represents costs exceeding the worst-case threshold $\thresholdWC$. The initial state in $\game'$ is $\initState' = (\initState, 0)$, and the target set is $\truncatedTarget' = \truncatedTarget \times \{0, 1, \ldots,\thresholdWC-1\}$, i.e., in $\game'$ the target set is restricted to copies of states of the original target set that are reached with a sum less than $\thresholdWC$. Notice that for any state $\state_{1}' = (\state_{1}, \tsFailSymbol) \in \states'$, all its successors in $\graph'$ are of the form $\state_{2}' = (\state_{2}, \tsFailSymbol)$.

Now, we compute in $\game'$ the set of states $R \subseteq \states'$ from which $\playerOne$ has a strategy to force reaching $\truncatedTarget'$ using a classical attractor computation, i.e., $R = \attr_{\game'}^{\playerOne}(\truncatedTarget')$ (cf. Sect.~\ref{sec:preliminaries}). Clearly, all states outside this attractor set are losing for the worst-case requirement. Indeed, from any state outside of $R$, either $\playerOne$ cannot force reaching a state $\state' = (\state, u)$ with $\state \in \truncatedTarget$, or he can only do it for $u = \tsFailSymbol$. In particular, if $\initState'=(\initState,0) \not\in R$ then we know that $\playerOne$ cannot enforce the worst-case threshold in the original game $\game$, and we can stop here in this case: no strategy exists for the $\BWC$ problem. 

Assume that  $\initState'=(\initState,0) \in R$. Let us write $\gameTS = \game' \reduc R$. Note that there will be deadlocks in $\gameTS$ (i.e., states with no successors): this is guaranteed since the sum of weights is strictly growing (recall $\weight\colon \edges \rightarrow \natStrict$) and states of the form $(\state, \tsFailSymbol)$ do not belong to $R$ by definition. However, the only deadlocks will be on states that are in $\truncatedTarget' \subseteq R$ (by definition of $R$ as the attractor of $\truncatedTarget'$). Hence, we get rid of them by adding self-loops of weight zero (this does not change the truncated sum up to $\truncatedTarget'$ by definition of $\truncatedSum{\truncatedTarget'}$).
It is easy to see that all strategies $\strat_{1} \in \strats_{1}(\gameTS)$ are winning for the worst-case requirement, and that there is a bijection between the winning strategies for the worst-case requirement in the original game $\game$ and the strategies in $\gameTS$. 

We are now equipped to handle the expectation objective. We proceed as follows. First, we take the product of~$\gameTS$ and $\mooreMachineFull{\stratStoch}$, following the construction of Lemma~\ref{lem:mp_memorylessStochModel} (the proof holds for the truncated sum value function as well). On the product game, we again preserve correspondence with the worst-case winning strategies in the original game $\game$. Applying the memoryless stochastic model (resulting of Lemma~\ref{lem:mp_memorylessStochModel}) on the product game, we obtain an MDP $\markovProcess$. It is then clear that $\playerOne$ has a strategy to enforce an expected value strictly less than the threshold $\thresholdExp$ in $\markovProcess$ \textit{if and only if} $\playerOne$ has a strategy that \textit{both} enforces the worst-case threshold against any strategy $\strat_{2} \in \strats_{2}(\game)$, and the expectation threshold against $\stratStoch$ in $\game$. To decide if such a strategy exists, we compute the minimal achievable expected value on $\markovProcess$ and we compare it against the threshold $\thresholdExp$. Thanks to the reduction from truncated sum to total-payoff proposed in the preliminaries (Sect.~\ref{sec:preliminaries}), we know that this optimal value can be achieved by a memoryless strategy and its computation can be executed in polynomial time in the size of the encoding of $\markovProcess$ via linear programming~\cite{filar1997}. Hence, it requires time polynomial in the size of the encoding of $\game$ and $\mooreMachine{\stratStoch}$, and polynomial in the value $\thresholdWC$ (since $\vert\states'\vert = \vert\states\vert \cdot (\thresholdWC + 1)$).
\end{proof}

\subsection{\textbf{Memory requirements}}
\label{subsec:ts_memory}
  
In the next theorem, we characterize the memory needed by strategies satisfying the $\BWC$ problem.
 
\begin{theorem}
\label{thm:ts_memory}
Memory of pseudo-polynomial size may be necessary and is always sufficient to
satisfy the $\BWC$ problem for the shortest path: polynomial in the size of the game and the stochastic model, and polynomial in the worst-case threshold value.
\end{theorem}
  
\begin{proof}
The upper bound on the size of the memory can be obtained directly from the construction exposed in the proof of Thm.~\ref{thm:ts_pseudoPoly}. Indeed, we have shown that if the $\BWC$ problem can be satisfied, the memoryless strategy that minimizes the expectation in the MDP $\markovProcess$ does satisfy it. Translated back to the original game, this strategy has a memory which is polynomial in $\vert\game\vert$, $\vert\mooreMachine{\stratStoch}\vert$, and the value of $\thresholdWC$. Intuitively, the strategy needs to memorize the current value of the sum of weights, up to the value of the worst-case threshold (at which point it does not matter to bookkeep it anymore as $\playerOne$ has already failed to enforce the worst-case requirement). Hence, such a strategy requires memory polynomial in the input game and the stochastic model, and in the threshold.

To prove that pseudo-polynomial memory may be necessary, we introduce a family $(\game(\thresholdWC))_{\thresholdWC \in \{13 + k\cdot 4 \mid k \in \nat\}}$ of games, indexed by the value of the worst-case threshold. This value is taken in a specific set $\{13 + k\cdot 4 \mid k \in \nat\}$ mostly to ease the following calculations. The family is presented in Fig.~\ref{fig:ts_familyExpInThresholdWC}: it consists of three states $\states = \{\state_{1}, \state_{2}, \state_{3}\}$. The weight function only assigns strictly positive weights as assumed in the setting of the shortest path problem. All weights are equal to $1$ except for edge $(\state_{1}, \state_{3})$ which has a weight $\left\lfloor\frac{\thresholdWC}{2}\right\rfloor$. Notice that $\thresholdWC$ is chosen odd and such that $\left\lfloor\frac{\thresholdWC}{2}\right\rfloor$ is even.
We will consider the values of the expectation threshold $\thresholdExp$ that can be ensured by a $\BWC$ strategy in such a game, under the chosen worst-case threshold $\thresholdWC$ and against a stochastic model assigning uniform distributions, and show that to minimize this value, $\playerOne$ needs to use linear memory in $\thresholdWC$, hence proving the claim.

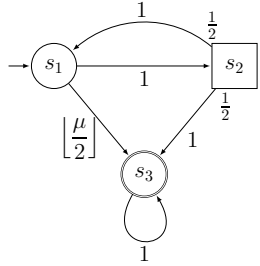
\begin{figure}[thb]
  \centering
  \scalebox{0.6}{\begin{tikzpicture}[->,>=latex,shorten >=1pt,auto,node
    distance=2.5cm,bend angle=45,font=\Large]
    \tikzstyle{p1}=[draw,circle,text centered,minimum size=10mm]
    \tikzstyle{p2}=[draw,rectangle,text centered,minimum size=10mm]
    \tikzstyle{empty}=[]
    \node[p1] (1) at (0,0) {$\state_{1}$};
    \node[p2] (2) at (4,0) {$\state_{2}$};
    \node[p1,double] (3) at (2,-2.4) {$\state_{3}$};
    \node[empty] (a) at (3.5,0.9) {$\frac{1}{2}$};
    \node[empty] (b) at (3.8,-0.9) {$\frac{1}{2}$};
    \coordinate[shift={(-5mm,0mm)}] (init) at (1.west);
    \path
    (1) edge node[below] {$1$} (2)
    (2) edge node[right,near end,xshift=2mm] {$1$} (3)
    (1) edge node[left,near end,xshift=-2mm] {$\left\lfloor\dfrac{\thresholdWC}{2}\right\rfloor$} (3)
    (3) edge [loop below, out=240, in=300,looseness=2, distance=16mm] node [below] {$1$} (3)
    (init) edge (1)
    ;
	\draw[->,>=latex] (2) to[out=140,in=40] node[above] {$1$} (1);
      \end{tikzpicture}}
  \vspace{-2mm}
      \caption{Family of games requiring memory linear in $\thresholdWC$ for the $\BWC$ problem.}
\label{fig:ts_familyExpInThresholdWC}
\end{figure}

First, observe that if the running sum of weights (which is an integer value) gets strictly larger than~$\left\lfloor\frac{\thresholdWC}{2}\right\rfloor$, then $\playerOne$ has lost the worst-case requirement (Eq.~\eqref{eq:thresholdWC}) as playing $(\state_{1}, \state_{2})$ does not guarantee reaching $\truncatedTarget$, and playing $(\state_{1}, \state_{3})$ induces a total cost at least equal to $\thresholdWC$. Hence, when in $\state_{1}$ with a running sum equal to $\left\lfloor\frac{\thresholdWC}{2}\right\rfloor$, $\playerOne$ has no valid choice but to take the edge $(\state_{1}, \state_{3})$. Since randomization clearly does not help (as it will produce consistent outcomes that are losing if the edge $(\state_{1}, \state_{2})$ is repeatedly assigned a non-zero probability), defining the optimal strategy of $\playerOne$ boils down to deciding for how long he should take the edge $(\state_{1}, \state_{2})$ before switching (if at all).

We claim that it should maximize the number of passes in $\state_{2}$ (Fig.~\ref{fig:ts_optimalStratForBWC}). Let $n$ denotes the number of times $\playerOne$ chooses $(\state_{1}, \state_{2})$ before switching. Clearly, to guarantee satisfaction of the worst-case requirement, we need $2\cdot n + \left\lfloor\frac{\thresholdWC}{2}\right\rfloor < \thresholdWC$. Since the threshold is odd, we have $2\cdot n + \left\lfloor\frac{\thresholdWC}{2}\right\rfloor < 2\cdot n + \frac{\thresholdWC}{2}$. Hence, it suffices to have $2\cdot n + \frac{\thresholdWC}{2} \leq \thresholdWC$, or equivalently, $n \leq \frac{\thresholdWC}{4}$. Note that this bound is linear in the value of $\thresholdWC$. What remains to prove is that increasing the number of passes results in a decrease of the expected value. Let $e(n)$ denotes the expected value induced by the strategy that plays edge $(\state_{1}, \state_{2})$ for $n$ times before switching. Careful computation reveals that $e(n)$ can be expressed as follows:
$e(n) = \sum_{i=0}^{n-1} \frac{1}{2^{i-1}} + \frac{1}{2^{n}}\cdot \left\lfloor\frac{\thresholdWC}{2}\right\rfloor$.
Our thesis is that for all $n \geq 0$, $e(n) < e(n-1)$. By the previous equation, that is
\begin{align*}
\sum_{i=0}^{n-1} \frac{1}{2^{i-1}} + \frac{1}{2^{n}}\cdot \left\lfloor\frac{\thresholdWC}{2}\right\rfloor < \sum_{i=0}^{n-2} \frac{1}{2^{i-1}} + \frac{1}{2^{n-1}}\cdot \left\lfloor\frac{\thresholdWC}{2}\right\rfloor \:\Leftrightarrow\:
\frac{1}{2^{n-2}} - \frac{1}{2^{n}}\cdot \left\lfloor\frac{\thresholdWC}{2}\right\rfloor < 0 \:\Leftrightarrow\:
\left\lfloor\frac{\thresholdWC}{2}\right\rfloor > \frac{2^{n}}{2^{n-2}} = 4 \:\Leftrightarrow\:
\thresholdWC > 9,
\end{align*}
and the last inequality is true thanks to the hypothesis that $\thresholdWC \in \{13 + k\cdot 4 \mid k \in \nat\}$. This shows that increasing $n$ decreases the expectation, as wanted.

\begin{figure}[htb]
  \centering   
  \scalebox{0.4}{\begin{tikzpicture}[->,>=latex,shorten >=1pt,auto,node
    distance=2.5cm,bend angle=45,font=\huge,scale=1.2]
    \tikzstyle{p1}=[draw,circle,text centered,minimum size=35mm]
    \tikzstyle{p2}=[draw,rectangle,text centered,minimum size=20mm]
    \tikzstyle{empty}=[]
    \node[p1] (1) at (0,0) {$\state_{1}, 0$};
    \node[p2] (2) at (5,0) {$\state_{2}, 1$};
    \node[p1] (3) at (10,0) {$\state_{1}, 2$};
    \node[p2] (4) at (15,0) {$\state_{2}, 3$};
    \node[p2] (5) at (22,0) {$\state_{2}, \left\lfloor\dfrac{\thresholdWC}{2}\right\rfloor - 1$};
    \node[p1] (6) at (27,0) {$\state_{1}, \left\lfloor\dfrac{\thresholdWC}{2}\right\rfloor$};
    \node[p1] (7) at (5,-5) {$\state_{3}, 2$};
    \node[p1] (8) at (15,-5) {$\state_{3}, 4$};
    \node[p1] (9) at (22,-5) {$\state_{3}, \left\lfloor\dfrac{\thresholdWC}{2}\right\rfloor$};
    \node[p1] (10) at (27,-5) {$\state_{3}, 2\cdot\left\lfloor\dfrac{\thresholdWC}{2}\right\rfloor$};
    \node[empty] (a) at (4.8,-1.4) {$\frac{1}{2}$};
    \node[empty] (b) at (14.8,-1.4) {$\frac{1}{2}$};
    \node[empty] (c) at (21.8,-1.4) {$\frac{1}{2}$};
    \node[empty] (d) at (6.2,0.4) {$\frac{1}{2}$};
    \node[empty] (e) at (16.2,0.4) {$\frac{1}{2}$};
    \node[empty] (f) at (23.8,0.4) {$\frac{1}{2}$};
    \coordinate[shift={(-5mm,0mm)}] (init) at (1.west);
    \path
    (1) edge node[above] {$1$} (2)
    (2) edge node[above] {$1$} (3)
    (3) edge node[above] {$1$} (4)
    (5) edge node[above] {$1$} (6)
    (2) edge node[right] {$1$} (7)
    (4) edge node[right] {$1$} (8)
    (5) edge node[right] {$1$} (9)
    (6) edge node[right] {$\left\lfloor\dfrac{\thresholdWC}{2}\right\rfloor$} (10)
    (init) edge (1)
    ;
	\draw[->,>=latex] (4) -- (18,0);
	\draw[loosely dashed,-] (18,0) -- (20.6,0);
      \end{tikzpicture}}
      \caption{Partial representation of the Markov chain induced by the $\BWC$ strategy that minimizes the expected cost to target in $\game(\thresholdWC)$, $\thresholdWC \in \{13 + k\cdot 4 \mid k \in \nat\}$.}
\label{fig:ts_optimalStratForBWC}
  \end{figure}

In conclusion, the optimal $\BWC$ strategy for the expected value criterion consists in playing $(\state_{1}, \state_{2})$ for exactly $n = \left\lfloor\frac{\thresholdWC}{4}\right\rfloor$ times, then swithing to $(\state_{1}, \state_{3})$ to ensure the worst-case (the corresponding MC is represented in Fig.~\ref{fig:ts_optimalStratForBWC}). Following our computations, it is possible to impose that playing this strategy is necessary to satisfy the $\BWC$ problem by taking the expected value threshold such that $e(n) < \thresholdExp \leq e(n-1)$. This proves that memory linear in $\thresholdWC$ is needed for the given family of games.
 \end{proof}
 
\begin{remark}
\label{rem:ts_infMem}
In contrast to the case of the mean-payoff value function (Sect.~\ref{subsec:mpInfiniteMemory}), infinite memory gives no additional power here. Indeed, the proof of Thm.~\ref{thm:ts_pseudoPoly} gives a complete representation of worst-case winning strategies through the game $\game_{\thresholdWC}$ and it is further proved that finite memory suffices to define an optimal strategy with regard to the expected value criterion among these worst-case winning strategies.
\end{remark}

\subsection{\textbf{$\NPTIME$-hardness of the decision problem}}
\label{subsec:ts_hardness}
 
We conclude our study of the $\BWC$ problem in the shortest path setting (i.e., for the truncated sum value function) by showing that it is very unlikely that a truly-polynomial (i.e., also polynomial in the size of the encoding of the worst-case threshold) time algorithm exists, as we establish in Thm.~\ref{thm:ts_NPHardness} that the decision problem is $\NPTIME$-hard. We prove it by reduction from the \textit{$K^{th}$ largest subset problem}~\cite{garey_FNY1979}. 
A recent paper by Haase and Kiefer~\cite{HaasePP} shows that this $K^{th}$ largest subset problem is actually $\PPTIME$-complete. Thus, it suggests that the $\BWC$ shortest path problem does not belong to $\NPTIME$ at all, otherwise the polynomial hierarchy would collapse to $\PTIME^{\NPTIME}$ by Toda's theorem~\cite{toda1991pp}.

The $K^{th}$ largest subset problem is expressed as follows. Given a finite set $\kthSet$, a size function $\kthSizeFctFull$ assigning strictly positive integer values to elements of $\kthSet$, and two naturals $\kthSetsNbr, \kthSetMaxSum \in \nat$, decide if there exist $\kthSetsNbr$ distinct subsets $\kthSubset_{i} \subseteq \kthSet$, $1 \leq i \leq \kthSetsNbr$, such that $\kthSizeFct(\kthSubset_{i}) = \sum_{\kthElem \in \kthSubset_{i}} \kthSizeFct(\kthElem) \leq \kthSetMaxSum$ for all $\kthSetsNbr$ subsets. The $\NPTIME$-hardness of this problem was proved in~\cite{johnson_JACM1978} via a Turing reduction from the partition problem.

\begin{wrapfigure}{r}{80mm}
  \centering   
  \vspace{-6mm}
  \scalebox{0.68}{\begin{tikzpicture}[->,>=latex,shorten >=1pt,auto,node
    distance=2.5cm,bend angle=45,font=\Large]
    \tikzstyle{p1}=[draw,circle,text centered,minimum size=10mm]
    \tikzstyle{p2}=[draw,rectangle,text centered,minimum size=10mm]
    \tikzstyle{empty}=[]
    \node[p1] (1) at (0,0) {{\small choice}};
    \node[p2] (2) at (4,-2) {$\kthWCState$};
    \node[p2] (3) at (4,2) {$\kthExpState$};
    \node[p1] (4) at (8,0) {{\small target}};
    \node[empty] (a) at (4.7,2.3) {$0$};
    \node[empty] (b) at (3.85,1.25) {$1$};
    \node[empty] (C) at (4.7,-2.3) {$1$};
    \coordinate[shift={(-5mm,0mm)}] (init) at (1.west);
    \path
    (1) edge node[below] {$1$} (2)
    (1) edge node[above] {$1$} (3)
    (4) edge [loop right, out=35, in=325,looseness=2, distance=16mm] node [right] {$1$} (4)
    (init) edge (1)
    ;
	\draw[->,>=latex] (2) to[out=0,in=230] node[below,xshift=2mm] {$\kthWeightC$} (4);
	\draw[->,>=latex] (3) to[out=270,in=180] node[below] {$\kthWeightB$} (4);
	\draw[->,>=latex] (3) to[out=0,in=100] node[above] {$\kthWeightA$} (4);
      \end{tikzpicture}}
      \caption{Choice gadget: choosing $\kthExpState$ is best for the expected value, but it is safe with regard to the worst-case if and only if the random subset selection produced a subset $\kthSubset$ such that $\kthSizeFct(\kthSubset) \leq \kthSetMaxSum$.}
\label{fig:ts_NPhardChoiceGadget}
\vspace{-8mm}
  \end{wrapfigure}
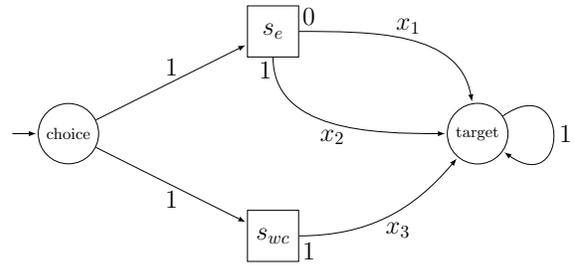

The crux of the reduction is as follows. We build a game composed of two gadgets. The \textit{random subset selection gadget} (Fig.~\ref{fig:ts_NPhardRandomSelectGadget}) stochastically generates paths that represent subsets of $\kthSet$: all subsets are equiprobable. The \textit{choice gadget} follows (Fig.~\ref{fig:ts_NPhardChoiceGadget}): $\playerOne$ decides either to go to $\kthExpState$, which leads to lower expected values (and lower is better in our setting) but may be dangerous for the worst-case requirement, or to go to~$\kthWCState$, which is always safe with regard to the worst-case but induces a higher expected cost. The trick is to define values of the thresholds and the weights used in the gadgets such that an optimal (i.e., minimizing the expectation while guaranteeing the worst-case threshold) strategy for $\playerOne$ consists in choosing $\kthExpState$ only when the randomly generated subset $\kthSubset \subseteq \kthSet$ satisfies $\kthSizeFct(\kthSubset) \leq \kthSetMaxSum$, as asked by the~$K^{th}$ largest subset problem; and such that this strategy satisfies the $\BWC$ problem if and only if there exist~$\kthSetsNbr$ distinct subsets that verify this bound, i.e., if and only if the answer to the~$K^{th}$ largest subset problem is $\yes$.

\begin{figure}[htb]
  \centering   
  \scalebox{0.6}{\begin{tikzpicture}[->,>=latex,shorten >=1pt,auto,node
    distance=2.5cm,bend angle=45,font=\Large]
    \tikzstyle{p1}=[draw,circle,text centered,minimum size=10mm]
    \tikzstyle{p2}=[draw,rectangle,text centered,minimum size=10mm]
    \tikzstyle{empty}=[]
    \node[p2] (1) at (0,0) {$\kthElem_{1}$};
    \node[p2] (2) at (4,0) {$\kthElem_{2}$};
    \node[p2] (3) at (8,0) {$\kthElem_{3}$};
    \node[p2] (4) at (12,0) {$\kthElem_{\kthSetSize}$};
    \node[p1] (5) at (16,0) {{\small choice}};
    \node[empty] (a) at (0.5,0.9) {$\frac{1}{2}$};
    \node[empty] (b) at (0.5,-0.9) {$\frac{1}{2}$};
    \node[empty] (c) at (4.5,0.9) {$\frac{1}{2}$};
    \node[empty] (d) at (4.5,-0.9) {$\frac{1}{2}$};
    \node[empty] (e) at (12.5,0.9) {$\frac{1}{2}$};
    \node[empty] (f) at (12.5,-0.9) {$\frac{1}{2}$};
    \coordinate[shift={(-5mm,0mm)}] (init) at (1.west);
    \path
    (init) edge (1)
    ;
	\draw[->,>=latex] (1) to[out=40,in=140] node[above] {$\kthNewSizeFct(\kthElem_{1})$} (2);
	\draw[->,>=latex] (2) to[out=40,in=140] node[above] {$\kthNewSizeFct(\kthElem_{2})$} (3);
	\draw[->,>=latex] (4) to[out=40,in=140] node[above] {$\kthNewSizeFct(\kthElem_{\kthSetSize})$} (5);
	\draw[->,>=latex] (1) to[out=320,in=220] node[below] {$1$} (2);
	\draw[->,>=latex] (2) to[out=320,in=220] node[below] {$1$} (3);
	\draw[->,>=latex] (4) to[out=320,in=220] node[below] {$1$} (5);
	\draw[-,loosely dashed] (3) to (4);
      \end{tikzpicture}}
      \caption{Random subset selection gadget: an element is selected in the subset if the upper edge is taken when leaving the corresponding state.}
\label{fig:ts_NPhardRandomSelectGadget}
  \end{figure}
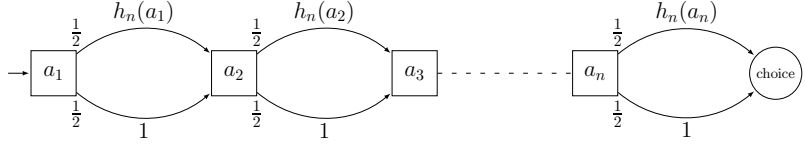
\begin{theorem}
\label{thm:ts_NPHardness}
The beyond worst-case problem for the shortest path is $\NPTIME$-hard.
\end{theorem}

\begin{proof}
We establish a reduction from the $K^{th}$ largest subset problem: given a finite set $\kthSet = \{\kthElem_{1}, \ldots{}, \kthElem_{\kthSetSize}\}$ (hence $\kthSetSize = \vert\kthSet\vert$), a size function $\kthSizeFctFull$, and two naturals $\kthSetsNbr, \kthSetMaxSum \in \nat$, decide if there exist $\kthSetsNbr$ distinct subsets $\kthSubset_{i} \subseteq \kthSet$, $1 \leq i \leq \kthSetsNbr$, such that $\kthSizeFct(\kthSubset_{i}) = \sum_{\kthElem \in \kthSubset_{i}} \kthSizeFct(\kthElem) \leq \kthSetMaxSum$ for all $\kthSetsNbr$ subsets. This problem is known to be $\NPTIME$-hard~\cite{johnson_JACM1978,garey_FNY1979}. Note that the restriction to~$\natStrict$ for the codomain of $\kthSizeFct$ in place of $\nat$ is w.l.o.g. as the problem is satisfied for $\kthSet$, $\kthSetsNbr$, $\kthSetMaxSum$ and $\kthSizeFct\colon \kthSet \rightarrow \nat$ if and only if it is satisfied for $\kthSet' = \kthSet \setminus \{\kthElem \in \kthSet \mid \kthSizeFct(\kthElem) = 0\}$, $\kthSetsNbr' = \left\lfloor\frac{\kthSetsNbr}{2^{\vert\kthSet\vert - \vert\kthSet'\vert}}\right\rfloor$, $\kthSetMaxSum' = \kthSetMaxSum$ and $\kthSizeFct'\colon \kthSet' \rightarrow \natStrict$ such that for all $\kthElem \in \kthSet'$, $\kthSizeFct'(\kthElem) = \kthSizeFct(\kthElem)$. Obviously, we should have $\kthSetsNbr \leq 2^{\kthSetSize}$, otherwise the problem is trivial since we cannot find a sufficient number of \textit{distinct} subsets.

Before giving the details of our reduction, we define, given $\kthSet$ and $\kthSizeFct$, the function $\kthNewSizeFct\colon \kthSet \rightarrow \natStrict$ such that for each $\kthElem \in \kthSet$, $\kthNewSizeFct(\kthElem) = (\kthSetSize + 1) \cdot \kthSizeFct(\kthElem)$. Clearly, it satisfies the following property: 
\begin{equation}
\label{eq:ts_scaling}
\forall\, \kthSubset \subseteq \kthSet,\, \kthSizeFct(\kthSubset) \leq \kthSetMaxSum \Leftrightarrow \kthNewSizeFct(\kthSubset) \leq (\kthSetSize + 1) \cdot \kthSetMaxSum.
\end{equation}

We present two gadgets used to construct a game and a $\BWC$ shortest path problem such that the answer to the $K^{th}$ largest subset problem is $\yes$ if and only if the answer to the $\BWC$ problem is also $\yes$.

First, the fragment of the game graph depicted in Fig.~\ref{fig:ts_NPhardRandomSelectGadget} is called the {\em random subset selection gadget}. All its states belong to $\playerTwo$, except for the last one, and model the selection (or not) of an element of $\kthSet$ in a subset. Basically, there is a bijection between paths in this gadget and subsets of $\kthSet$: an element $\kthElem_{i} \in \kthSet$ is selected by the gadget if the outgoing upper edge is taken when leaving state $\kthElem_{i}$, and not selected when the outgoing lower edge is taken. To be able to formally distinguish between such paths, which we usually define as sequence of \textit{states}, we should introduce dummy states to split edges. We omit this technical trick for the sake of simplicity. The stochastic model followed by $\playerTwo$ in the $\BWC$ shortest path problem we construct is the uniform distribution: the upper and lower edges are equiprobable in all states. This gadget verifies the following important properties. 
\begin{enumerate}
\item All subsets are equiprobable: they have probability $\frac{1}{2^\kthSetSize}$ to be selected.
\item If the gadget selects a subset $\kthSubset \subseteq \kthSet$ through the corresponding path $\kthRandomPath{\kthSubset}$, the total sum of weights along $\kthRandomPath{\kthSubset}$, denoted by $\kthPathSizeFct(\kthRandomPath{\kthSubset})$, is equal to $\kthNewSizeFct(\kthSubset) + n - \vert\kthSubset\vert$.
\end{enumerate}

By Eq.~\eqref{eq:ts_scaling}, we have that 
\begin{equation}
\label{eq:ts_equivalencePathSetForMax}
\forall\, \kthSubset \subseteq \kthSet,\, \kthSizeFct(\kthSubset) = \kthSetMaxSum \Leftrightarrow (\kthSetSize + 1) \cdot \kthSetMaxSum \leq \kthPathSizeFct(\kthRandomPath{\kthSubset}) < (\kthSetSize + 1) \cdot (\kthSetMaxSum + 1).
\end{equation}
Indeed, consider the following. Observe that $0 \leq \kthSetSize - \vert\kthSubset\vert \leq \kthSetSize$, for any subset $\kthSubset \subseteq \kthSet$. Hence the left-to-right implication is trivial. For the converse, we directly deduce the following equivalent expression:
$\kthSetMaxSum - \frac{\kthSetSize - \vert\kthSubset\vert}{\kthSetSize + 1} \leq \kthSizeFct(\kthSubset) < \kthSetMaxSum + 1 - \frac{\kthSetSize - \vert\kthSubset\vert}{\kthSetSize + 1}$.
The left inequality implies that $\kthSetMaxSum - 1 < \kthSizeFct(\kthSubset)$, and since $\kthSizeFct(\kthSubset) \in \nat$, that $\kthSetMaxSum \leq \kthSizeFct(\kthSubset)$. The right inequality implies that $\kthSizeFct(\kthSubset) < \kthSetMaxSum + 1$, and using the same argument, that $\kthSizeFct(\kthSubset) \leq \kthSetMaxSum$. We conclude that Eq.~\eqref{eq:ts_equivalencePathSetForMax} is true. Consequently, we define the value $\kthPathBound = (\kthSetSize + 1) \cdot (\kthSetMaxSum + 1) - 1$, which is an upper bound on the value $\kthPathSizeFct(\kthRandomPath{\kthSubset}) \leq \kthPathBound$ of a path corresponding to a subset $\kthSubset \subseteq \kthSet$ such that $\kthSizeFct(\kthSubset) \leq \kthSetMaxSum$.

Now consider the second gadget, called the {\em choice gadget} and depicted in Fig.~\ref{fig:ts_NPhardChoiceGadget}. This gadget comes after the random subset selection gadget. Its aim is to discriminate subsets generated by the preceding gadget based on whether or not they satisfy the upper bound  $\kthSizeFct(\kthSubset) \leq \kthSetMaxSum$. Observe the shared \textit{choice} state. There, $\playerOne$ has the choice to go up to state $\kthExpState$ or down to state $\kthWCState$. Both belong to $\playerTwo$. Again, probabilities for the stochastic model of the adversary are depicted in Fig.~\ref{fig:ts_NPhardChoiceGadget}. So, in $\kthExpState$, an arbitrary strategy of $\playerTwo$ can decide to impose cost $\kthWeightA$ or cost $\kthWeightB$ before reaching the target set of the game. Nonetheless, the stochastic model $\stratStoch$ of $\playerTwo$ assigns probability zero to the edge of weight $\kthWeightA$: the expectation of any strategy of $\playerOne$ against this stochastic model will be independent of the value $\kthWeightA$. In $\kthWCState$, the cost added is always equal to~$\kthWeightC$. Intuitively, we will choose values so that to minimize his expected cost-to-target, $\playerOne$ should choose~$\kthExpState$, but also so that the worst-case requirement implies that it is only safe to choose this state if the previous path defined a subset that satisfies the bound $\kthSizeFct(\kthSubset) \leq \kthSetMaxSum$ given by the $K^{th}$ largest subset problem.

To complete the description of the reduction, we need to precise the values of the thresholds $\thresholdWC$ and $\thresholdExp$, and the weights $\kthWeightA$, $\kthWeightB$ and $\kthWeightC$. Assume that we choose the worst-case threshold and the weights such that:
\begin{itemize}
\item[\textit{(a)}] $\kthPathBound + 1 + \kthWeightA + 1 \geq \thresholdWC$, i.e., going to $\kthExpState$ with a path $\kthRandomPath{\kthSubset}$ (obtained in the random subset selection gadget) of cost $\kthPathSizeFct(\kthRandomPath{\kthSubset}) > \kthPathBound$ (i.e., with a selected subset $\kthSubset \subseteq \kthSet$ such that $\kthSizeFct(\kthSubset) > \kthSetMaxSum$ by Eq.~\eqref{eq:ts_equivalencePathSetForMax}) is losing for the worst-case threshold if $\playerTwo$ takes the edge of weight $\kthWeightA$;
\item[\textit{(b)}] $\kthPathBound + \kthWeightA + 1 < \thresholdWC$ and $\kthPathBound + \kthWeightB + 1 < \thresholdWC$, i.e., going to $\kthExpState$ with a path $\kthRandomPath{\kthSubset}$ of cost $\kthPathSizeFct(\kthRandomPath{\kthSubset}) \leq \kthPathBound$ (i.e., $\kthSizeFct(\kthSubset) \leq \kthSetMaxSum$) is safe for the worst-case requirement whatever the choice of $\playerTwo$;
\item[\textit{(c)}] for all $\kthSubset \subseteq \kthSet$, we have that $\kthPathSizeFct(\kthRandomPath{\kthSubset}) + \kthWeightC + 1 < \thresholdWC$, i.e., going to $\kthWCState$ is always safe for the worst-case requirement.
\end{itemize}
Then clearly, $\playerOne$ can always choose to go to $\kthWCState$ and ensure the worst-case threshold, but he can go up only if the chosen subset $\kthSubset$ satisfies $\kthSizeFct(\kthSubset) \leq \kthSetMaxSum$, which is equivalent to say that $\kthPathSizeFct(\kthRandomPath{\kthSubset}) \leq \kthPathBound$. We add the following constraints to the choices of the expectation threshold and the weights:
\begin{itemize}
\item[\textit{(d)}] in order to minimize the expected truncated sum in the MDP defined by the stochatic model, the optimal choice for~$\playerOne$ is to always take $\kthExpState$ when possible (i.e., when $\kthSizeFct(\kthSubset) \leq \kthSetMaxSum$, or equivalently $\kthPathSizeFct(\kthRandomPath{\kthSubset}) \leq \kthPathBound$ because of the constraint~\textit{(a)} defined above);
\item[\textit{(e)}] the expected value $\optimalExp$ of this optimal choice satisfies the expectation requirement (i.e., $\optimalExp < \thresholdExp$) if and only if the number of distinct subsets $\kthSubset_{i} \subseteq \kthSet$ verifying $\kthSizeFct(\kthSubset_{i}) \leq \kthSetMaxSum$ is larger than or equal to $K$.
\end{itemize}

We will now define values such that properties \textit{(a)} through \textit{(e)} are ensured. First, let $\kthPathMax = \max \{\kthPathSizeFct(\kthRandomPath{\kthSubset}) \mid \kthSubset \subseteq \kthSet\} = \kthPathSizeFct(\kthRandomPath{\kthSet})$ be the maximal cost of a path in the random subset selection gadget (the equality with $\kthPathSizeFct(\kthRandomPath{\kthSet})$ is thanks to the size function $\kthSizeFct$ assigning strictly positive values). We claim the needed properties are verified for the following values:
$\thresholdWC = 2^{\kthSetSize + 1} \cdot \kthSetSize \cdot (\kthPathMax + 2)$, $\thresholdExp = \big(\kthSetsNbr \cdot (\kthPathBound + 2) + (2^{\kthSetSize} - \kthSetsNbr) \cdot \thresholdWC\big)/2^{\kthSetSize}$,
$\kthWeightA = \thresholdWC - \kthPathBound - 2$, $\kthWeightB = 1$, and $\kthWeightC = \thresholdWC - \kthPathMax - 2$.
Using these, we review each property one-by-one. For \textit{(a)}, we obtain by simple substitutions
$\textit{(a)} \:\Leftrightarrow\: \kthPathBound + 1 + \thresholdWC - \kthPathBound - 2 + 1 \geq \thresholdWC \:\Leftrightarrow\: 0 \geq 0$,
which is obvious. Similarly, for \textit{(b)}, we have that
\begin{equation}
\label{eq:ts_hardnessProofEqA}
\textit{(b)} \quad\Leftrightarrow\quad (\kthPathBound + \thresholdWC - \kthPathBound - 2 + 1 < \thresholdWC) \wedge (\kthPathBound + 1 + 1 < \thresholdWC) \quad\Leftrightarrow\quad (-1 < 0) \wedge (\kthPathBound + 2 < 2^{\kthSetSize + 1} \cdot \kthSetSize \cdot (\kthPathMax + 2)).
\end{equation}
The first term of the conjunction is trivially true so we focus on the second one. Without loss of generality, we can assume that $\kthSetMaxSum < \kthSizeFct(\kthSet)$ as otherwise the $K^{th}$ largest subset problem reduces to decide if $\kthSetsNbr \leq 2^{\kthSetSize}$. Thus, we deduce the inequality $\kthPathBound  < (\kthSetSize + 1) \cdot (\kthSizeFct(\kthSet) + 1) - 1$. Also note that, by definition, we have that $\kthPathMax = \kthPathSizeFct(\kthRandomPath{\kthSet}) = \kthNewSizeFct(\kthSet) = (\kthSetSize + 1) \cdot \kthSizeFct(\kthSet)$. Using these inequalities in Eq.~\eqref{eq:ts_hardnessProofEqA}, we derive that proving the following central inequality suffices to obtain \textit{(b)}:
\begin{equation*}
\kthPathBound + 2 < (\kthSetSize + 1) \cdot (\kthSizeFct(\kthSet) + 1) + 1 \leq 2^{\kthSetSize + 1} \cdot \kthSetSize \cdot ((\kthSetSize + 1) \cdot \kthSizeFct(\kthSet) + 2) = 2^{\kthSetSize + 1} \cdot \kthSetSize \cdot (\kthPathMax + 2).
\end{equation*}
This boils down to
$(2^{\kthSetSize + 1}\cdot \kthSetSize - 1)\cdot (\kthSetSize + 1)\cdot \kthSizeFct(\kthSet) + (2\cdot 2^{\kthSetSize+1} - 1)\cdot \kthSetSize - 2 \geq 0$,
which is true for $\kthSetSize \geq 1$ (which we can assume otherwise $\kthSet = \emptyset$ and the problem is trivial). Hence, property \textit{(b)} is verified by our choice of values. Now, consider property \textit{(c)}. We have
$\textit{(c)} \:\Leftrightarrow\: \kthPathSizeFct(\kthRandomPath{\kthSubset}) + \thresholdWC - \kthPathMax - 2 + 1 < \thresholdWC \:\Leftrightarrow\:  \kthPathSizeFct(\kthRandomPath{\kthSubset}) < \kthPathMax + 1$, which is true by definition of $\kthPathMax$ as the maximum over the values of paths. Regarding property \textit{(d)}, we have to show that choosing $\kthExpState$ gives an expectation strictly lower (recall we want to minimize it) than choosing~$\kthWCState$. Observe that due to the particular structure of the game graph, the strategy of $\playerOne$ is restricted to this one-shot choice of edge. Note that in this expected value context, the actual value obtained in the random subset selection gadget does not matter to decide whether to go to $\kthExpState$ or to $\kthWCState$: hence it suffices to look at the expectation from the \textit{choice} state up to the \textit{target} state. For $\kthExpState$, it is trivially equal to $1 + 1 = 2$ as the stochastic model $\stratStoch$ of $\playerTwo$ always chooses the edge of weight~$\kthWeightB$. For $\kthWCState$, this expectation is equal to
\begin{equation*}
1 + \kthWeightC = 1 + \thresholdWC - \kthPathMax - 2 = 2^{\kthSetSize + 1} \cdot \kthSetSize \cdot (\kthPathMax + 2) - \kthPathMax - 1 \geq (2^{\kthSetSize + 1} \cdot \kthSetSize - 1) \cdot (\kthPathMax + 2) \geq (\kthPathMax + 2) > 2,
\end{equation*}
and we obtain the claim \textit{(d)}. Note that an actual strategy that satisfies the $\BWC$ problem will only be able to choose $\kthExpState$ if the selected path satisfies the bound $\kthPathSizeFct(\kthRandomPath{\kthSubset}) \leq \kthPathBound$, as discussed in properties \textit{(a)} and \textit{(b)}.

Finally, it remains to show the most involved property \textit{(e)}: proving it will conclude our reduction as we will obtain that the answer to the $K^{th}$ largest subset problem is $\yes$ if and only if the answer to the $\BWC$ problem we have defined is $\yes$. Note that combining the already proved properties \textit{(a)} through \textit{(d)}, we know that the strategy $\strat_{1} \in \stratsPureFinite_{1}$ that chooses state $\kthExpState$ when $\kthPathSizeFct(\kthRandomPath{\kthSubset}) \leq \kthPathBound$ and state $\kthWCState$ otherwise, yields the minimal expectation value $\optimalExp$ \textit{under the worst-case constraint} of threshold $\thresholdWC$. Hence, it suffices to study this strategy to answer the $\BWC$ problem. Our claim is thus that
\begin{equation}
\label{eq:ts_hardnessProofEqC}
\optimalExp = \expect^{\game[\strat_{1}, \stratStoch]}_{a_{1}} < \thresholdExp \quad\Leftrightarrow\quad \Big\vert \left\lbrace \kthSubset \subseteq \kthSet \mid \kthSizeFct(\kthSubset) \leq \kthSetMaxSum \right\rbrace \Big\vert \geq \kthSetsNbr,
\end{equation}
with $\game$ and $\stratStoch$ the game and stochastic model we defined.

For the left-to-right implication, we reason by contradiction and show that if there is only $\kthSetsNbr - 1$ (or less) distinct subsets whose sum is less than or equal to $\kthSetMaxSum$, then strategy $\strat_{1}$ has an expected cost larger than or equal to $\thresholdExp$. To show that, we use the fact that all paths (i.e., subsets) have equal probability in the random subset selection gadget, and establish that a lower bound on the sum of all the paths under this strategy reaches or exceeds $2^{\kthSetSize} \cdot \thresholdExp$. Recall that $\playerTwo$ follows its stochastic model $\stratStoch$ for this matter. First, let $\kthLBExp = 0$ which is trivially a lower bound for the cost of all the paths that goes through $\kthExpState$. Second, let $\kthLBWC = (2^{\kthSetSize} - (\kthSetsNbr - 1)) \cdot (\kthPathBound + 1 + \kthWeightC + 1)$: it is clearly a lower bound for the sum of the values of paths that go through state $\kthWCState$ when $\playerOne$ follows strategy $\strat_{1}$. We have that $2^{\kthSetSize} \cdot \optimalExp \geq \kthLBExp + \kthLBWC$. Let us now establish that $\kthLBWC \geq 2^{\kthSetSize} \cdot \thresholdExp$ and we will be done. We proceed as follows.
\begin{align*}
\kthLBWC - 2^{\kthSetSize} \cdot \thresholdExp &= (2^{\kthSetSize} - \kthSetsNbr + 1) \cdot (\kthPathBound + 1 + \thresholdWC - \kthPathMax - 2 + 1) - \kthSetsNbr \cdot (\kthPathBound + 2) - (2^{\kthSetSize} - \kthSetsNbr)\cdot \thresholdWC\\
&= \thresholdWC + (2^{\kthSetSize} - \kthSetsNbr + 1) \cdot (\kthPathBound - \kthPathMax) - \kthSetsNbr \cdot (\kthPathBound + 2)\\
&= 2^{\kthSetSize + 1} \cdot \kthSetSize \cdot (\kthPathMax + 2) + (2^{\kthSetSize} - \kthSetsNbr + 1) \cdot (\kthPathBound - \kthPathMax) - \kthSetsNbr \cdot (\kthPathBound + 2)
\end{align*}
Recall that $\kthPathBound, \kthPathMax \geq 0$, $\kthSetSize \geq 1$ and $0 \leq \kthSetsNbr \leq 2^{\kthSetSize}$ (otherwise the answer is trivially $\no$). Furthermore, $\kthPathBound$ is the upper bound on the values of paths $\kthRandomPath{\kthSubset}$ representing good subsets, i.e., subsets $\kthSubset \subseteq \kthSet$ such that $\kthSizeFct(\kthSubset) \leq \kthSetMaxSum$. This value is used by the strategy $\strat_{1}$ implemented by $\playerOne$ to decide whether going to $\kthExpState$ is safe with regard to the worst-case requirement or not. As such, we can assume that $\kthPathBound \leq \kthPathMax$, otherwise all paths are safe and the answer to the problem is trivially $\yes$ (since all subsets respect the bound and $\kthSetsNbr \leq 2^{\kthSetSize}$). Using $\kthSetsNbr \leq 2^{\kthSetSize}$, $\kthPathBound \geq 0$ and $\kthPathBound \leq \kthPathMax$, we can write $\kthLBWC - 2^{\kthSetSize} \cdot \thresholdExp \geq (2^{\kthSetSize + 1} \cdot \kthSetSize - 2^{\kthSetSize} + \kthSetsNbr - 1 -\kthSetsNbr) \cdot \kthPathMax + (2^{\kthSetSize + 1} \cdot \kthSetSize \cdot 2 - 2 \cdot \kthSetsNbr)$.

To prove that this last expression is non-negative, we analyze its terms. We know that $\kthPathMax \geq 0$. For its coefficient, we have 
$2^{\kthSetSize + 1} \cdot \kthSetSize - 2^{\kthSetSize} - 1 \geq 2^{\kthSetSize} - 1 \geq 0$ because $\kthSetSize \geq 1$. For the last term, we use $\kthSetsNbr \leq 2^{\kthSetSize}$ and obtain that $2^{\kthSetSize + 2} \cdot \kthSetSize - 2 \cdot \kthSetsNbr \geq 2^{\kthSetSize + 2} - 2^{\kthSetSize + 1} = 2^{\kthSetSize + 1} \geq 0$.
Hence all terms of the last expression are non-negative and $\kthLBWC \geq 2^{\kthSetSize} \cdot \thresholdExp$, proving that the left-to-right implication of Eq.~\eqref{eq:ts_hardnessProofEqC} is verified.

It remains to prove the right-to-left implication. Assume there are exactly $\kthSetsNbr$ distinct subsets of sum less than or equal to $\kthSetMaxSum$ (if there are more, then the bounds below are easier to obtain). We claim that strategy $\strat_{1}$ ensures an expected truncated sum $\optimalExp$ strictly lower than $\thresholdExp$. To show this, we establish that the total sum of the outcomes under this strategy of $\playerOne$ and the stochastic model of $\playerTwo$ is strictly bounded from above by $2^{\kthSetSize} \cdot \thresholdExp$, and the claim follows thanks to all paths being equiprobable in the random subset selection gadget. First, consider the paths that go through $\kthExpState$ (i.e., all the paths corresponding to subsets~$\kthSubset$ such that $\kthSizeFct(\kthSubset) \leq \kthSetMaxSum$). By definition of $\strat_{1}$ and our hypothesis, there are exactly $\kthSetsNbr$ such paths. We define $\kthUBExp = \kthSetsNbr \cdot (\kthPathBound + 2)$, a clear upper bound for the sum of the values of these paths, by definition of $\kthPathBound$ and~$\stratStoch$. Second, there are $(2^{\kthSetSize} - \kthSetsNbr)$ paths that go through $\kthWCState$. Let $\kthUBWC = (2^{\kthSetSize} - \kthSetsNbr) \cdot (\kthPathMax + 1 + \kthWeightC) = (2^{\kthSetSize} - \kthSetsNbr) \cdot (\thresholdWC - 1)$ be a bound for the sum of the values of all these paths. Clearly,
\begin{equation*}
\small
2^{\kthSetSize}\cdot \optimalExp \leq \kthUBExp + \kthUBWC = \kthSetsNbr \cdot (\kthPathBound + 2) + (2^{\kthSetSize} - \kthSetsNbr) \cdot (\thresholdWC - 1) < \kthSetsNbr \cdot (\kthPathBound + 2) + (2^{\kthSetSize} - \kthSetsNbr) \cdot \thresholdWC = 2^{\kthSetSize}\cdot \thresholdExp
\end{equation*}
by definition of the expected value threshold $\thresholdExp$, and so we are done for this direction.

Having verified both directions of the equivalence given in Eq.~\eqref{eq:ts_hardnessProofEqC}, the correctness of our reduction from the $K^{th}$ largest subset problem is established. Note that it requires values of the thresholds that are exponential in the size of the set $\kthSet$ and polynomial in the value of the largest weight assigned by the size function $\kthSizeFct$ (or equivalently, exponential in its encoding) for the $K^{th}$ largest subset problem. It also requires to use edge weights that are polynomial in these values. Observe that this is not a problem, as all those values may be represented using a logarithmic number of bits, hence polynomially in the characteristics of the initial $K^{th}$ largest subset problem. Finally, notice that we do need to consider exponential constants in our game to obtain the $\NPTIME$-hardness of our problem, as for values polynomial in the size of the game, the algorithm described in Sect.~\ref{subsec:ts_pseudoPolyAlgo} actually operates in truly-polynomial time. This concludes our proof.
\end{proof}

\section{Conclusion}
\label{sec:conclusion}

In this paper, we paved the way to a new approach, combining worst-case and expected value requirements in what we named the \textit{beyond worst-case synthesis problem}. We believe this setting is adequate to the synthesis of controllers that must ensure strict guarantees under all circumstances, and prove to be more efficient in reasonable conditions, a problem for which few theoretical frameworks exist.

We thoroughly studied the $\BWC$ synthesis problem in the context of two well-known quantitative measures: the \textit{mean-payoff} and the \textit{shortest path}. For the mean-payoff, we proved the problem to be in $\NPinter$, matching the complexity of the worst-case threshold problem \cite{ZP96,jurdzinski98,gawlitza2009}, which we encompass. Hence, the $\BWC$ setting provides additional modeling power at no complexity cost (in terms of problem solving), a remarkably positive result. For the shortest path, the $\BWC$ problem proves to be harder than the worst-case threshold problem, going from polynomial to pseudo-polynomial time, with an $\NPTIME$-hardness result. In both cases, synthesized strategies may require pseudo-polynomial memory, but accept natural, elegant representations, based on states of the game and simple integer counters.

Possible future work include study of the $\BWC$ problem for other quantitative measures, extension of our results for the mean-payoff and the shortest path to more general settings (multi-dimension~\cite{DBLP:journals/iandc/VelnerC0HRR15,DBLP:journals/acta/ChatterjeeRR14}, decidable classes of games with imperfect information~\cite{degorre_CSL2010,HunterPR14}, etc), and application of the $\BWC$ problem to various practical cases.
Given the relevance of the framework for practical applications, it would certainly be worthwhile to develop tool suites supporting it. We could for example build on symblicit implementations recently developed for monotonic Markov decision processes by Bohy et al.~\cite{DBLP:journals/corr/BohyBR14a}. Following the publication of~\cite{bruyere_STACS2014}, the concept of beyond worst-case synthesis has also been applied to priced timed games in~\cite{larsen2014}.

\bibliographystyle{plain}
\bibliography{bwc_bib}

\begin{thebibliography}{10}

\bibitem{baier_MIT08}
C.~Baier and J.-P. Katoen.
\newblock {\em Principles of model checking}.
\newblock MIT Press, 2008.

\bibitem{bertsekas_MOR1991}
D.P. Bertsekas and J.N. Tsitsiklis.
\newblock An analysis of stochastic shortest path problems.
\newblock {\em Mathematics of Operations Research}, 16:580--595, 1991.

\bibitem{DBLP:journals/corr/BohyBR14a}
A.~Bohy, V.~Bruy{\`{e}}re, and J.-F. Raskin.
\newblock Symblicit algorithms for optimal strategy synthesis in monotonic
  {M}arkov decision processes.
\newblock In {\em Proc. of SYNT}, {EPTCS 157}, pages 51--67, 2014.

\bibitem{Bra79}
G.~Brassard.
\newblock A note on the complexity of cryptography (corresp.).
\newblock {\em IEEE Transactions on Information Theory}, 25(2):232--233, 1979.

\bibitem{brazdil_LICS2013}
T.~Br{\'a}zdil, K.~Chatterjee, V.~Forejt, and A.~Kucera.
\newblock Trading performance for stability in {M}arkov decision processes.
\newblock In {\em Proc. of LICS}, pages 331--340. IEEE Computer Society, 2013.

\bibitem{BCDGR11}
L.~Brim, J.~Chaloupka, L.~Doyen, R.~Gentilini, and J.-F. Raskin.
\newblock Faster algorithms for mean-payoff games.
\newblock {\em Formal Methods in System Design}, 38(2):97--118, 2011.

\bibitem{DBLP:journals/corr/BruyereFRR14}
V.~Bruy{\`{e}}re, E.~Filiot, M.~Randour, and J.-F. Raskin.
\newblock Expectations or guarantees? {I} want it all! {A} crossroad between
  games and {MDP}s.
\newblock In {\em Proc. of SR}, {EPTCS} 146, pages 1--8, 2014.

\bibitem{bruyere_STACS2014}
V.~Bruy{\`{e}}re, E.~Filiot, M.~Randour, and J.-F. Raskin.
\newblock Meet your expectations with guarantees: Beyond worst-case synthesis
  in quantitative games.
\newblock In {\em Proc. of STACS}, LIPIcs 25, pages 199--213. Schloss Dagstuhl
  - LZI, 2014.

\bibitem{chatterjee_MEMICS11}
K.~Chatterjee and L.~Doyen.
\newblock Games and {M}arkov decision processes with mean-payoff parity and
  energy parity objectives.
\newblock In {\em Proc. of MEMICS}, LNCS 7119, pages 37--46. Springer, 2011.

\bibitem{Chatterjee201525}
K.~Chatterjee, L.~Doyen, M.~Randour, and J.-F. Raskin.
\newblock Looking at mean-payoff and total-payoff through windows.
\newblock {\em Information and Computation}, 242:25 -- 52, 2015.

\bibitem{DBLP:journals/jacm/ChatterjeeH14}
K.~Chatterjee and M.~Henzinger.
\newblock Efficient and dynamic algorithms for alternating {B}{\"{u}}chi games
  and maximal end-component decomposition.
\newblock {\em J. {ACM}}, 61(3):15:1--15:40, 2014.

\bibitem{DBLP:journals/acta/ChatterjeeRR14}
K.~Chatterjee, M.~Randour, and J.-F. Raskin.
\newblock Strategy synthesis for multi-dimensional quantitative objectives.
\newblock {\em Acta Informatica}, 51(3-4):129--163, 2014.

\bibitem{CR15}
L.~Clemente and J.-F. Raskin.
\newblock Controller synthesis for the multidimensional beyond worst-case and
  almost-sure problems.
\newblock In {\em Proc. of LICS}, pages 257--268. IEEE Computer Society, 2015.

\bibitem{courcoubetis_JACM1995}
C.~Courcoubetis and M.~Yannakakis.
\newblock The complexity of probabilistic verification.
\newblock {\em J. ACM}, 42(4):857--907, 1995.

\bibitem{larsen2014}
A.~David, P.G. Jensen, K.G. Larsen, A.~Legay, D.~Lime, M.G. Sørensen, and J.H.
  Taankvist.
\newblock On time with minimal expected cost!
\newblock In {\em Proc. of ATVA}, LNCS 8837, pages 129--145. Springer, 2014.

\bibitem{de1997formal}
L.~de~Alfaro.
\newblock {\em Formal verification of probabilistic systems}.
\newblock PhD thesis, Stanford University, 1997.

\bibitem{deAlfaro_CONCUR1999}
L.~de~Alfaro.
\newblock Computing minimum and maximum reachability times in probabilistic
  systems.
\newblock In {\em Proc. of CONCUR}, LNCS 1664, pages 66--81. Springer, 1999.

\bibitem{degorre_CSL2010}
A.~Degorre, L.~Doyen, R.~Gentilini, J.-F. Raskin, and S.~Torunczyk.
\newblock Energy and mean-payoff games with imperfect information.
\newblock In {\em Proc. of CSL}, LNCS 6247, pages 260--274. Springer, 2010.

\bibitem{EM79}
A.~Ehrenfeucht and J.~Mycielski.
\newblock Positional strategies for mean payoff games.
\newblock {\em Int. Journal of Game Theory}, 8(2):109--113, 1979.

\bibitem{filar1997}
J.~Filar and K.~Vrieze.
\newblock {\em Competitive {M}arkov decision processes}.
\newblock Springer, 1997.

\bibitem{FKR95}
J.A. Filar, D.~Krass, and K.W. Ross.
\newblock Percentile performance criteria for limiting average {M}arkov
  decision processes.
\newblock {\em Transactions on Automatic Control}, pages 2--10, 1995.

\bibitem{garey_FNY1979}
M.R. Garey and D.S. Johnson.
\newblock {\em Computers and intractability: a guide to the Theory of
  {NP}-Completeness}.
\newblock Freeman New York, 1979.

\bibitem{gawlitza2009}
T.~Gawlitza and H.~Seidl.
\newblock Games through nested fixpoints.
\newblock In {\em Proc. of CAV}, LNCS 5643, pages 291--305. Springer, 2009.

\bibitem{gimbert2004}
H.~Gimbert and W.~Zielonka.
\newblock When can you play positionally?
\newblock In {\em Proc. of MFCS}, LNCS 3153, pages 686--697. Springer, 2004.

\bibitem{glynn_SPL2002}
P.W. Glynn and D.~Ormoneit.
\newblock Hoeffding's inequality for uniformly ergodic {M}arkov chains.
\newblock {\em Statistics \& Probability Letters}, 56(2):143--146, 2002.

\bibitem{grinstead_AMS1997}
C.M. Grinstead and J.L. Snell.
\newblock {\em Introduction to probability}.
\newblock American Mathematical Society, 1997.

\bibitem{HaasePP}
C.~Haase and S.~Kiefer.
\newblock The complexity of the {K}th largest subset problem and related
  problems.
\newblock {\em Information Processing Letters}, 2015.
\newblock To appear.

\bibitem{hoeffding_JASA1963}
W.~Hoeffding.
\newblock Probability inequalities for sums of bounded random variables.
\newblock {\em Journal of the American Statistical Association},
  58(301):13--30, 1963.

\bibitem{HunterPR14}
P.~Hunter, G.A. P{\'{e}}rez, and J.-F. Raskin.
\newblock Mean-payoff games with partial-observation.
\newblock In {\em Proc. of RP}, LNCS 8762, pages 163--175. Springer, 2014.

\bibitem{johnson_JACM1978}
D.B. Johnson and S.D. Kashdan.
\newblock Lower bounds for selection in {X + Y} and other multisets.
\newblock {\em Journal of the ACM}, 25(4):556--570, 1978.

\bibitem{jurdzinski98}
M.~Jurdzi\'nski.
\newblock Deciding the winner in parity games is in {UP} $\cap$ co-{UP}.
\newblock {\em Information Processing Letters}, 68(3):119--124, 1998.

\bibitem{DBLP:dblp_journals/mst/KhachiyanBBEGRZ08}
L.~Khachiyan, E.~Boros, K.~Borys, K.M. Elbassioni, V.~Gurvich, G.~Rudolf, and
  J.~Zhao.
\newblock On short paths interdiction problems: Total and node-wise limited
  interdiction.
\newblock {\em Theory of Computing Systems}, 43:204--233, 2008.

\bibitem{liggett_SR69}
T.M. Liggett and S.A. Lippman.
\newblock Stochastic games with perfect information and time average payoff.
\newblock {\em Siam Review}, 11(4):604--607, 1969.

\bibitem{mannor_ICML2011}
S.~Mannor and J.N. Tsitsiklis.
\newblock Mean-variance optimization in {M}arkov decision processes.
\newblock In {\em Proc. of ICML}, pages 177--184. Omnipress, 2011.

\bibitem{martin_AM75}
D.A. Martin.
\newblock Borel determinacy.
\newblock {\em Annals of Mathematics}, 102(2):363--371, 1975.

\bibitem{Puterman94}
M.L. Puterman.
\newblock {\em Markov decision processes: discrete stochastic dynamic
  programming}.
\newblock John Wiley \& Sons, Inc., New York, NY, USA, 1st edition, 1994.

\bibitem{Ran14}
M.~Randour.
\newblock {\em Synthesis in Multi-Criteria Quantitative Games}.
\newblock PhD thesis, University of Mons, Belgium, 2014.

\bibitem{RRS15a}
M.~Randour, J.-F. Raskin, and O.~Sankur.
\newblock Percentile queries in multi-dimensional {M}arkov decision processes.
\newblock In {\em Proc. of CAV}, LNCS 9206, pages 123--139. Springer, 2015.

\bibitem{RRS15b}
M.~Randour, J.-F. Raskin, and O.~Sankur.
\newblock Variations on the stochastic shortest path problem.
\newblock In {\em Proc. of VMCAI}, LNCS 8931, pages 1--18. Springer, 2015.

\bibitem{toda1991pp}
S.~Toda.
\newblock {PP} is as hard as the polynomial-time hierarchy.
\newblock {\em SIAM Journal on Computing}, 20(5):865--877, 1991.

\bibitem{tracol_ORL2009}
M.~Tracol.
\newblock Fast convergence to state-action frequency polytopes for {MDPs}.
\newblock {\em Oper. Res. Lett.}, 37(2):123--126, 2009.

\bibitem{vardi_FOCS85}
M.Y. Vardi.
\newblock Automatic verification of probabilistic concurrent finite-state
  programs.
\newblock In {\em Proc. of FOCS}, pages 327--338. IEEE Computer Society, 1985.

\bibitem{DBLP:journals/iandc/VelnerC0HRR15}
Y.~Velner, K.~Chatterjee, L.~Doyen, T.A. Henzinger, A.M. Rabinovich, and J.-F.
  Raskin.
\newblock The complexity of multi-mean-payoff and multi-energy games.
\newblock {\em Information and Computation}, 241:177--196, 2015.

\bibitem{WL99}
C.~Wu and Y.~Lin.
\newblock Minimizing risk models in {M}arkov decision processes with policies
  depending on target values.
\newblock {\em Journal of Mathematical Analysis and Applications},
  231(1):47--67, 1999.

\bibitem{ZP96}
U.~Zwick and M.~Paterson.
\newblock The complexity of mean payoff games on graphs.
\newblock {\em Theoretical Computer Science}, 158:343--359, 1996.

\end{thebibliography}

\end{document}